\def\llncs{0}
\def\fullpage{1}
\def\anonymous{0}
\def\authnote{0}
\def\notxfont{0}
\def\submission{0}
\def\reply{0}
\def\cameraready{0}

\ifnum\submission=1
\def\anonymous{1}
\def\authnote{0}
\def\llncs{1}
\else
\fi

\ifnum\cameraready=1
\def\llncs{1}
\def\anonymous{0}
\def\authnote{0}
\else
\fi

\ifnum\llncs=1
	\documentclass[runningheads]{llncs}
\else
	\documentclass[letterpaper,hmargin=1.05in,vmargin=1.05in]{article}
			\ifnum\fullpage=1
		\usepackage{fullpage}
		\fi
\fi

\ifnum\reply=1
\usepackage{ulem}
\renewcommand{\emph}{\textit}
\else
\fi

\usepackage[%
  colorlinks=true,
  citecolor=darkgreen,
  linkcolor=darkblue,
  urlcolor=darkblue,
  pagebackref=true
]{hyperref}

\usepackage{amsmath, amsfonts, amssymb, mathtools,amscd}

\usepackage{amsthm}

\usepackage{lmodern}
\usepackage[T1]{fontenc}
\usepackage[utf8]{inputenc}

\usepackage{etex}

\usepackage{arydshln} 
\usepackage{url}
\usepackage{ifthen}
\usepackage{bm}
\usepackage{multirow}
\usepackage[dvips]{graphicx}
\usepackage[usenames]{color}
\usepackage{xcolor,colortbl} 
\usepackage{threeparttable}
\usepackage{comment}
\usepackage{paralist,verbatim}
\usepackage{cases}
\usepackage{booktabs}
\usepackage{braket}
\usepackage{cancel} 
\usepackage{ascmac} 
\usepackage{framed}
\usepackage{authblk}
\usepackage{pifont}
\usepackage{physics}
\definecolor{darkblue}{rgb}{0,0,0.6}
\definecolor{darkgreen}{rgb}{0,0.5,0}
\definecolor{maroon}{rgb}{0.5,0.1,0.1}
\definecolor{dpurple}{rgb}{0.2,0,0.65}
\definecolor{chocolate}{rgb}{0.8,0.4,0.1}

\usepackage[capitalise,noabbrev]{cleveref}
\usepackage[absolute]{textpos}
\usepackage[final]{microtype}
\usepackage[absolute]{textpos}
\usepackage{everypage}
\usepackage{autonum}

\usepackage{dsfont}
\DeclareMathAlphabet{\mathpzc}{OT1}{pzc}{m}{it}

\ifnum\submission=1
\renewcommand*{\backref}[1]{}
\def\notxfont{1}
\else
\fi

\ifnum\llncs=1
\renewcommand{\subparagraph}{\paragraph}
\else
\ifnum\notxfont=1
\else
\usepackage{mathpazo}
\usepackage{newtxtext}
\usepackage{helvet}
\fi
\fi

\newtheoremstyle{thicktheorem}%
{\topsep}
{\topsep}
{\itshape}{}%
{\bfseries}%
{.}
{ }%
{\thmname{#1}\thmnumber{ #2}%
		\thmnote{ (#3)}%
}

\newtheoremstyle{remark}
{\topsep}
{\topsep}
	{}
	{}
	{}
	{.}
	{ }
	{\textit{\thmname{#1}}\thmnumber{ #2}
			\thmnote{ (#3)}%
	}

\ifnum\llncs=0
	\theoremstyle{thicktheorem}
	\newtheorem{theorem}{Theorem}[section]
	\newtheorem{lemma}[theorem]{Lemma}
	\newtheorem{corollary}[theorem]{Corollary}
	\newtheorem{proposition}[theorem]{Proposition}
	\newtheorem{definition}[theorem]{Definition}

	\theoremstyle{remark}
	\newtheorem{claim}[theorem]{Claim}
	\newtheorem{remark}[theorem]{Remark}

\fi

	\crefname{theorem}{Theorem}{Theorems}
	\crefname{assumption}{Assumption}{Assumptions}
	\crefname{construction}{Construction}{Constructions}
	\crefname{corollary}{Corollary}{Corollaries}
	\crefname{conjecture}{Conjecture}{Conjectures}
	\crefname{definition}{Definition}{Definitions}
	\crefname{exmaple}{Example}{Examples}
	\crefname{experiment}{Experiment}{Experiments}
	\crefname{counterexample}{Counterexample}{Counterexamples}
	\crefname{lemma}{Lemma}{Lemmata}
	\crefname{observation}{Observation}{Observations}
	\crefname{proposition}{Proposition}{Propositions}
	\crefname{remark}{Remark}{Remarks}
	\crefname{claim}{Claim}{Claims}
	\crefname{fact}{Fact}{Facts}
	\crefname{note}{Note}{Notes}

\ifnum\llncs=1
 \crefname{appendix}{App.}{Appendices}
 \crefname{section}{Sec.}{Sections}
\else
\fi

\ifnum\llncs=1
\renewcommand*{\backref}[1]{}
\else
	\renewcommand*{\backref}[1]{(Cited on page~#1.)}
	\ifnum\notxfont=1
	\else
		\usepackage{newtxtext}
	\fi
\fi

\newcommand*{\keys}[1]{\mathsf{#1}}

\newcommand{\Oracle}[1]{O_{\mathtt{#1}}}

\newcommand*{\algo}[1]{\ensuremath{\mathsf{#1}}}
\newcommand*{\qalgo}[1]{\ensuremath{\mathpzc{#1}}}
\newcommand*{\qstate}[1]{\mathpzc{#1}}
\newcommand*{\qreg}[1]{{\color{gray}{\mathsf{#1}}}}

\newcommand*{\entity}[1]{\mathcal{#1}}

\newcounter{expitem}

\usepackage{mathtools}

\newcommand{\chosen}{\leftarrow}
\newcommand{\sample}{\leftarrow}
\newcommand{\lrun}{\leftarrow}

\newcommand{\la}{\leftarrow}
\newcommand{\ra}{\rightarrow}
\renewcommand{\gets}{\leftarrow}

\newcommand{\seteq}{\coloneqq}

\newcommand{\tensor}{\otimes}
\newcommand{\concat}{\|}

\newcommand{\setbk}[1]{\{#1\}}

\newcommand{\cC}{\mathcal{C}}
\newcommand{\cD}{\mathcal{D}}
\newcommand{\cE}{\mathcal{E}}
\newcommand{\cF}{\mathcal{F}}

\newcommand{\cI}{\mathcal{I}}

\newcommand{\cM}{\mathcal{M}}

\newcommand{\cO}{\mathcal{O}}
\newcommand{\cP}{\mathcal{P}}
\newcommand{\cQ}{\mathcal{Q}}
\newcommand{\cR}{\mathcal{R}}

\newcommand{\cT}{\mathcal{T}}

\newcommand{\qA}{\qalgo{A}}
\newcommand{\qB}{\qalgo{B}}
\newcommand{\qC}{\qalgo{C}}

\newcommand{\qG}{\qalgo{G}}

\def\makeuppercase#1{
\expandafter\newcommand\csname sf#1\endcsname{\mathsf{#1}}
\expandafter\newcommand\csname frak#1\endcsname{\mathfrak{#1}}
\expandafter\newcommand\csname bb#1\endcsname{\mathbb{#1}}
\expandafter\newcommand\csname bf#1\endcsname{\textbf{#1}}
}

\def\makelowercase#1{
\expandafter\newcommand\csname frak#1\endcsname{\mathfrak{#1}}
\expandafter\newcommand\csname bf#1\endcsname{\textbf{#1}}
}

\newcounter{char}
\loop
   \edef\letter{\alph{char}}
   \edef\Letter{\Alph{char}}
   \expandafter\makelowercase\letter
   \expandafter\makeuppercase\Letter
   \stepcounter{char}
   \unless\ifnum\thechar>26
\repeat

\def\makeuppercase#1{
\expandafter\newcommand\csname tl#1\endcsname{\widetilde{#1}}
}

\def\makelowercase#1{
\expandafter\newcommand\csname tl#1\endcsname{\widetilde{#1}}
}

\newcommand{\R}{\mathbb{R}}

\newcommand{\bit}{\{0,1\}}

\newcommand{\Fs}{\mathcal{F}}
\newcommand{\Ms}{\mathcal{M}}

\newcommand{\secp}{\lambda}
\newcommand{\secpar}{\secp}

\newcommand{\coin}{\keys{coin}}

\newcommand{\sig}{\sigma}
\newcommand{\cert}{\keys{cert}}

\newcommand{\aux}{\mathsf{aux}}

\newcommand{\A}{\entity{A}}

\newcommand{\IND}{\textrm{IND}}
\newcommand{\CPA}{\textrm{CPA}}

\newcommand{\INDCPA}{\IND\textrm{-}\CPA}

\newcommand*{\xProblem}[1]{\ensuremath{\mathrm{#1}}}

\newcommand*{\GapSVP}{\xProblem{GapSVP}}
\newcommand*{\SIS}{\xProblem{SIS}}

\newcommand{\advb}[3]{\mathsf{Adv}_{#1}^{\mathsf{#2} \mbox{-} \mathsf{#3}}}
\newcommand{\advc}[4]{\mathsf{Adv}_{#1}^{\mathsf{#2} \mbox{-} \mathsf{#3} \mbox{-} \mathsf{#4}}}

\newcommand{\expb}[3]{\mathsf{Exp}_{#1}^{ \mathsf{#2} \mbox{-} \mathsf{#3}}}
\newcommand{\expc}[4]{\mathsf{Exp}_{#1}^{ \mathsf{#2} \mbox{-} \mathsf{#3} \mbox{-} \mathsf{#4}}}

\newcommand{\hybi}[1]{\mathsf{Hyb}_{#1}}

\newcommand*{\pk}{\keys{pk}}
\newcommand*{\sk}{\keys{sk}}
\newcommand*{\dk}{\keys{dk}}
\newcommand*{\ek}{\keys{ek}}
\newcommand*{\dvk}{\keys{dvk}}
\newcommand*{\vk}{\keys{vk}}

\newcommand*{\svk}{\keys{svk}}

\newcommand*{\msk}{\keys{msk}}

\newcommand*{\ct}{\keys{ct}}

\newcommand*{\hatct}{\widehat{\ct}}

\newcommand*{\msg}{\keys{m}}

\newcommand{\qsk}{\qstate{sk}}

\newcommand{\qdk}{\qstate{dk}}

\newenvironment{boxfig}[2]{\begin{figure}[#1]\fbox{\begin{minipage}{0.97\linewidth}
                        \vspace{0.2em}
                        \makebox[0.025\linewidth]{}
                        \begin{minipage}{0.95\linewidth}
            {{
                        #2 }}
                        \end{minipage}
                        \vspace{0.2em}
                        \end{minipage}}
                        }
                        {\end{figure}}

\newcommand{\Setup}{\algo{Setup}}

\newcommand{\KG}{\algo{KG}}
\newcommand{\Enc}{\algo{Enc}}
\newcommand{\Dec}{\algo{Dec}}

\newcommand{\Sign}{\algo{Sign}}
\newcommand{\Vrfy}{\algo{Vrfy}}
\newcommand{\DelVrfy}{\algo{DelVrfy}}

\newcommand{\SigVrfy}{\algo{SigVrfy}}
\newcommand{\qSign}{\qalgo{Sign}}

\newcommand{\qKG}{\qalgo{KG}}

\newcommand{\qDec}{\qalgo{Dec}}

\newcommand{\TrapGen}{\algo{TrapGen}}

\newcommand{\EvalF}{\algo{EvalF}}
\newcommand{\EvalFX}{\algo{EvalFX}}

\newcommand\PKE{\algo{PKE}}

\newcommand{\SKE}{\algo{SKE}}

\newcommand{\E}{\algo{E}}
\newcommand{\D}{\algo{D}}

\newcommand{\Sim}{\algo{Sim}}

\newcommand{\Eval}{\algo{Eval}}

\newcommand{\negl}{{\mathsf{negl}}}

\newcommand{\poly}{{\mathrm{poly}}}

\newcommand{\zo}[1]{\{0,1\}^{#1}}
\newcommand{\bin}{\{0,1\}}

\newcommand{\Ext}{\mathrm{Ext}}

\newcommand{\UPFSKL}{\mathsf{UPFSKL}}

\newcommand{\PKESKL}{\algo{PKESKL}}
\newcommand{\PRFSKL}{\algo{PRFSKL}}
\newcommand{\DSSKL}{\algo{DSSKL}}

\newcommand{\qst}{\qstate{st}}
\newcommand{\qstateq}{\qstate{q}}
\newcommand{\mat}[1]{\boldsymbol{#1}}

\newcommand{\cPrev}{\overline{\cP}}

\newcommand{\sigvk}{\mathsf{svk}}

\newcommand{\qsigk}{\qstate{sigk}}

\newcommand{\API}{\qalgo{API}}

\newcommand{\projimp}{\algo{ProjImp}}

\newcommand{\qDel}{\qalgo{Del}}

\newcommand{\qExt}{\qalgo{Ext}}

\newcommand{\qaux}{\qstate{aux}}

\newcommand{\returned}{\top}
\newcommand{\unreturned}{\bot}
\newcommand{\requestchallenge}{\mathsf{RequestChallenge}}

\newcommand{\InvertLeft}{\mathsf{InvertLeft}}
\newcommand{\InvertRight}{\mathsf{InvertRight}}
\newcommand{\Invert}{\mathsf{Invert}}

\newcommand{\sigk}{\mathsf{sigk}}
\newcommand{\td}{\mathsf{td}}
\newcommand{\Sam}{\mathsf{Sam}}
\newcommand{\SamPre}{\mathsf{SamPre}}

\newcommand{\QSampGauss}{\qalgo{QSampGauss}}
\newcommand{\signature}{\mathsf{sig}}
\newcommand{\descf}{\langle f \rangle}
\newcommand{\QSign}{\qalgo{QSign}}

\newcommand{\Gauss}[2]{D_{#1, #2}} 
\newcommand{\NN}{\mathbb{N}}   
\newcommand{\ZZ}{\mathbb{Z}}
\newcommand{\RR}{\mathbb{R}}  
\newcommand{\CC}{\mathbb{C}} 
\newcommand{\mA}{\mathbf{A}}    \newcommand{\mB}{\mathbf{B}}
    \newcommand{\mG}{\mathbf{G}}
\newcommand{\mH}{\mathbf{H}}    
\newcommand{\mI}{\mathbf{I}}    \newcommand{\mR}{\mathbf{R}}
    
\newcommand{\mV}{\mathbf{V}}    
\newcommand{\mC}{\mathbf{C}}

\newcommand{\mS}{\mathbf{S}}

\newcommand{\vc}{\mathbf{c}}    
\newcommand{\ve}{\mathbf{e}}    
    
\newcommand{\vs}{\mathbf{s}}    \newcommand{\vt}{\mathbf{t}}
\renewcommand{\vu}{\mathbf{u}} 
 
\newcommand{\vx}{\mathbf{x}}    \newcommand{\vy}{\mathbf{y}}
\newcommand{\vz}{\mathbf{z}}    

\newcommand{\vzero}{\mathbf{0}} 
 
\newcommand{\smudpar}{\Gamma}
\newcommand{\trans}{\top}
\newcommand{\samp}{\leftarrow}

\newcommand{\bd}{\begin{description}}
\newcommand{\ed}{\end{description}}
\newcommand{\bt}{\begin{theorem}}
\newcommand{\et}{\end{theorem}}
\newcommand{\bc}{\begin{claim}}
\newcommand{\ec}{\end{claim}}
\newcommand{\bl}{\begin{lemma}}
\newcommand{\el}{\end{lemma}}
\newcommand{\br}{\begin{remark}}
\newcommand{\er}{\end{remark}}
\newcommand{\bde}{\begin{definition}}
\newcommand{\ede}{\end{definition}}
\newcommand{\bi}{\begin{itemize}}
\newcommand{\ei}{\end{itemize}}
\newcommand{\be}{\begin{enumerate}}
\newcommand{\ee}{\end{enumerate}}

\newcommand{\authornote}[3]{\textcolor{#3}{[\textbf{#1:} {#2}]}}
\ifnum\authnote=1
\newcommand{\fuyuki}[1]{\authornote{fuyuki}{#1}{chocolate}}
\newcommand{\takashi}[1]{\authornote{takashi}{#1}{dpurple}}
\newcommand{\mor}[1]{\authornote{tomoyuki}{#1}{red}}
\else
\newcommand{\fuyuki}[1]{}
\newcommand{\takashi}[1]{}
\newcommand{\mor}[1]{}
\fi

\ifnum\llncs=1
\let\oldvec\vec
\let\vec\oldvec
\makeatletter
\renewcommand*\l@author[2]{}
\renewcommand*\l@title[2]{}
\makeatletter
\fi    

\theoremstyle{remark}

\ifnum\submission=1
\title{
\textbf{A Simple Framework for Secure Key Leasing}
}
\else
\title{
\textbf{A Simple Framework for Secure Key Leasing}
}
\fi

\begin{document}

\ifnum\anonymous=1 
\ifnum\llncs=1
\author{\empty}\institute{\empty}
\else
\author{}
\fi
\else
%
%
\ifnum\llncs=1
\author{
	Fuyuki Kitagawa\inst{1,2} \and 
Tomoyuki Morimae\inst{3} \and
Takashi Yamakawa\inst{1,2,3}
}
\institute{ 
NTT Social Informatics Laboratories, Tokyo, Japan \and
NTT Research Center for Theoretical Quantum Information, Atsugi, Japan \and
    Yukawa Institute for Theoretical Physics, Kyoto University, Kyoto, Japan
}
\else
%
%
\author[1,2]{\hskip 1em Fuyuki Kitagawa}
\author[3]{\hskip 1em Tomoyuki Morimae}
\author[1,2,3]{\hskip 1em Takashi Yamakawa}
\affil[1]{{\small NTT Social Informatics Laboratories, Tokyo, Japan}\authorcr{\small \{fuyuki.kitagawa,takashi.yamakawa\}@ntt.com}}
\affil[2]{{\small NTT Research Center for Theoretical Quantum Information, Atsugi, Japan}}
\affil[3]{{\small Yukawa Institute for Theoretical Physics, Kyoto University, Kyoto, Japan}\authorcr{\small tomoyuki.morimae@yukawa.kyoto-u.ac.jp}}
\renewcommand\Authands{, }
\fi 
\fi

\ifnum\llncs=1
\date{}
\else
\ifnum\anonymous=0
\date{\today}
\else
\date{}
\fi
\fi

\maketitle

\begin{abstract}
Secure key leasing (a.k.a. key-revocable cryptography) enables us to lease a cryptographic key as a quantum state in such a way that the key can be later revoked in a verifiable manner. 
We propose a simple framework for constructing cryptographic primitives with secure key leasing via the certified deletion property of BB84 states. Based on our framework, we obtain the following schemes. 
\begin{itemize}
\item A public key encryption scheme with secure key leasing that has  classical revocation based on any IND-CPA secure public key encryption scheme. Prior works rely on either quantum revocation or stronger assumptions such as the quantum hardness of the learning with errors (LWE) problem.
\item A pseudorandom function with secure key leasing that has classical revocation based on one-way functions. Prior works rely on stronger assumptions such as the quantum hardness of the LWE problem.
\item A digital signature scheme with secure key leasing that has classical revocation based on the quantum hardness of the short integer solution (SIS) problem. 
Our construction has static signing keys, i.e., the state of a signing key almost does not change before and after signing. Prior constructions either rely on non-static signing keys or indistinguishability obfuscation to achieve a stronger goal of copy-protection.  
\end{itemize}
In addition, all of our schemes remain secure even if a verification key for revocation is leaked after the adversary submits a valid certificate of deletion. To our knowledge, all prior constructions are totally broken in this setting. Moreover, in our view, our security proofs are much simpler than those for existing schemes.    
\end{abstract}

\ifnum\llncs=1
\else
\newpage
\setcounter{tocdepth}{2}
\tableofcontents

\newpage
\fi


\newcommand{\qLEval}{\qalgo{LEval}}
\newcommand{\Domprf}{D_{\mathsf{prf}}}
\newcommand{\Ranprf}{R_{\mathsf{prf}}}
\newcommand{\Constrain}{\algo{Constrain}}
\newcommand{\CEval}{\algo{CEval}}
\newcommand{\stinfo}{\mathsf{st}}
\newcommand{\SMoEBB}{\mathsf{SMoEBB84}}
\newcommand{\expo}{\mathsf{expo}}
\newcommand{\ctlen}{\ell_{\ct}}

\newcommand{\CPRF}{\algo{CPRF}}
\newcommand{\inplen}{{\ell_{\mathsf{inp}}}}
\newcommand{\outlen}{{\ell_{\mathsf{out}}}}
\newcommand{\prfinp}{s}
\newcommand{\prfout}{t}
\newcommand{\cprfout}{\CPRF.\prfout}
\newcommand{\tlp}{\tilde{p}}
\newcommand{\tlD}{\tilde{D}}
\newcommand{\tlrho}{\tilde{\qsk}}
\newcommand{\qExtComp}{\qalgo{ExtComp}}
\newcommand{\ConfPRF}{\algo{ConfPRF}}
\newcommand{\TPRF}{\algo{TPRF}}
\newcommand{\qQKG}{\qalgo{QKG}}
\newcommand{\qQEval}{\qalgo{QEval}}
\renewcommand{\Ext}{\algo{Ext}}

\newcommand{\Dreal}[1]{D_{\mathsf{real},#1}}
\newcommand{\Sample}{\algo{Smp}}
\newcommand{\qp}{\qalgo{qp}}

\newcommand{\ObvKG}{\mathsf{ObvKG}}
\newcommand{\key}{\mathsf{sk}}
\newcommand{\TEPRF}{\mathsf{TEPRF}}

\newcommand{\chall}{\mathsf{chall}}
\newcommand{\ans}{\mathsf{ans}}
\newcommand{\Rel}{\mathsf{Rel}}
\newcommand{\str}{\mathsf{str}}
\renewcommand{\Ext}{\mathsf{Ext}}
\newcommand{\qSimKG}{\qalgo{SimKG}}
\newcommand{\hatP}{\widehat{P}}
\newcommand{\tlprfinp}{\widetilde{\prfinp}}
\newcommand{\hatprfinp}{\widehat{\prfinp}}
\newcommand{\hatprfout}{\widehat{\prfout}}
\newcommand{\rep}{\mathsf{rep}}
\newcommand{\hatm}{\widehat{m}}

\newcommand{\CS}{\mathsf{CS}}

\newcommand{\DS}{\mathsf{DS}}

\renewcommand{\key}{\mathsf{key}}

\newcommand{\hatD}{\widehat{D}}

\newcommand{\ExtComp}{\mathsf{ExtComp}}

\section{Introduction}

\subsection{Backgrounds}

\paragraph{Leasing cryptographic keys.}
The delegation and revocation of privileges are required in many application scenarios of cryptographic primitives. For example, the following situation is common in many applications: The lessor Alice leases cryptographic secret key $\sk$ to the lessee Bob so that Bob can have the privilege for a certain task by using $\sk$ for a limited time period, and at later point, Alice needs to revoke Bob's privilege. The de fact solution for this problem in classical cryptography is key update. Namely, Alice simply stops using $\sk$ when Alice revokes Bob, and starts to use a new key $\sk^\prime$. However, updating key is undesirable for many reasons and it should be avoided as long as possible. If $\sk$ is a secret key of public key encryption (PKE) or digital signature scheme, Alice needs to update not only $\sk$ but also its corresponding public key and announce the new public key all over the world, which is highly costly. Even worse, updating public key could cause vulnerability based on the time-lag in the update. For example, it might be the case that Charlie who does not realize that Alice's public encryption key is updated sends a ciphertext generated by using the old public encryption key after Bob is revoked. Since Bob has the old secret key needed to decrypt the ciphertext, Bob can get the message sent by Charlie even though Bob is already revoked.

\paragraph{Key leasing without key update?}
The above situation raises one natural question whether we can achieve key leasing capability without key update. It is clearly impossible in classical cryptography. For example, if Alice keeps using the un-updated public encryption key, it is impossible to prevent Bob from decrypting future ciphertexts.
However, surprisingly, key leasing without key update is in fact possible by utilizing quantum information together with cryptography.

\paragraph{Secure key leasing.}
Secure key leasing (a.k.a. key-revocable cryptography)~\cite{AC:KitNis22,EC:AKNYY23,TCC:AnaPorVai23} is a quantum cryptographic primitive that achieves key leasing capability without key update by considering key deletion protocol. In a secure key leasing scheme, we can generate a quantum secret key (that is, secret key encoded into a quantum state) together with a secret deletion verification key. Moreover, the secure key leasing scheme has the deletion algorithm and deletion verification algorithm, in addition to the algorithms that form a standard cryptographic primitive (such as the key generation, encryption, and decryption algorithms for PKE). The deletion verification key and the additional two algorithms are used in the deletion protocol. Consider again the situation where Alice leases (quantum) secret key $\qsk$ to Bob, and at later point, Alice needs to revoke Bob. In the secure key leasing scheme, when Alice revokes Bob, Alice requires Bob to send a deletion certificate by executing the deletion algorithm to the leased secret key $\qsk$. On receiving the deletion certificate from Bob, Alice can check its validity by using the deletion verification algorithm with the deletion verification key. The security of the secure key leasing scheme guarantees that if the deletion verification algorithm accepts the deletion certificate, Bob's privilege is lost, that is, Bob can no longer perform the task that requires $\qsk$. For example, in PKE with secure key leasing, Bob cannot decrypt a newly seen ciphertext once Bob deletes the leased key and the generated deletion certificate is accepted.
We emphasize that the security guarantee holds even if the ciphertext is generated by using the unupdated public encryption key, and it achieves key leasing capability without any key update.

In secure key leasing, the deletion certificate can either be a quantum state or a classical string. If a secure key leasing scheme has completely classical deletion certificate and deletion verification algorithm, 
we say that it has classical revocation.

\paragraph{Previous results on secure key leasing.}
A notable feature of secure key leasing is that it can be realized based on standard cryptographic assumptions.
This is contrast to the related notion of unclonable cryptographic keys that currently require highly strong assumption of post quantum indistinguishability obfuscation (iO)~\cite{C:CLLZ21,TCC:LLQZ22}.
Kitagawa and Nishimaki \cite{AC:KitNis22} constructed bounded-time secure secret key encryption (SKE) with secure key leasing based on any one-way function (OWF).\footnote{More precisely, they constructed bounded collusion secure secret key functional encryption with secure key leasing.}
Agrawal, Kitagawa, Nishimaki, Yamada, and Yamakawa \cite{EC:AKNYY23} and Ananth, Poremba, and Vaikuntanathan \cite{TCC:AnaPorVai23} concurrently and independently introduced PKE with secure key leasing (PKE-SKL)\footnote{\cite{TCC:AnaPorVai23} calls it revocable PKE.} and gave its constructions with incomparable features. The construction of \cite{EC:AKNYY23} is based on any standard (classical) PKE, but needs quantum revocation, while the construction of \cite{TCC:AnaPorVai23}  has classical revocation, but relies on a stronger assumption called the learning with errors (LWE) assumption and an unproven conjecture. 
The reliance on the unproven conjecture of \cite{TCC:AnaPorVai23} was removed by a follow-up work by Ananth, Hu, and Huang~\cite{AnaHuHua24}, but they still rely on the LWE assumption,  which is stronger than the mere existence of standard PKE. 
In addition to PKE, these works construct more primitives with secure key leasing.  
Specifically, \cite{EC:AKNYY23} extended their PKE-SKL to more advanced encryption primitives such as identity-based encryption (IBE), attribute-based encryption (ABE), and functional encryption (FE), with secure key leasing by assuming only their standard counterparts. 
\cite{TCC:AnaPorVai23,AnaHuHua24} extended their PKE-SKL to fully-homomorphic encryption with secure key leasing (FHE-SKL) and also constructed pseudorandom function with secure key leasing (PRF-SKL)  having classical revocation based on the LWE assumption.  
More recently, Morimae, Poremba, and Yamakawa \cite{TQC:MorPorYam24} introduced digital signatures with secure key leasing (DS-SKL)\footnote{They call it digital signatures with revocable signing keys (DSR-Key).} and constructed it with classical revocation based on the LWE assumption. 
However, their construction has a drawback that a signing key evolves whenever a new signature is issued, and in particular, the sizes of signing keys and signatures linearly grow with the total number of issued signatures.\footnote{They also sketched an idea to reduce the signature size to be logarithmic in the number of issued signatures, but it does not reduce the signing key size.} 
They also give another construction based on one-way group actions (OWGA)~\cite{TCC:JQSY19}. but the construction only supports quantum revocation and has a similar drawback that the sizes of signing keys and signatures linearly grow with the total number of issued signatures. 

Although we already have some positive results as listed above, secure key leasing is a relatively new primitive and our understanding for it has not yet developed enough.
In fact, we have many natural questions raised from the initial results.
This includes the following.
\begin{itemize}
\item Can we construct PKE-SKL with classical revocation from any standard PKE?
\item Can we construct PRF-SKL from standard PRFs or equivalently from any OWFs?
\item Can we construct DS-SKL such that the sizes of signing keys and signatures are independent of  the number of issued signatures?   
\end{itemize}

\begin{table}[t]
\centering
\caption{Comparison of PKE-SKL}
\label{table:pkeskl}
\begin{tabular}{cccc}
\toprule
Reference& Revocation&VRA security&Assumption\\
\midrule
\cite{EC:AKNYY23}& Quantum & &PKE\\
\midrule
\cite{TCC:AnaPorVai23,AnaHuHua24}& Classical & &LWE\\
\midrule
Ours& Classical & \checkmark&PKE\\
\bottomrule
\end{tabular}
\end{table}

\begin{table}[t]
\centering
\caption{Comparison of PRF-SKL}
\label{table:prfskl}
\begin{tabular}{cccc}
\toprule
Reference& Revocation&VRA security&Assumption\\
\midrule
\cite{TCC:AnaPorVai23,AnaHuHua24}& Classical & &LWE\\
\midrule
Ours& Classical & \checkmark& OWF\\
\bottomrule
\end{tabular}
\end{table}

\begin{table}[t]
\centering
\caption{Comparison of DS-SKL}
\label{table:dsskl}
\begin{tabular}{cccccc}
\toprule
Reference& Revocation&VRA security&Signing key size&Signature size &Assumption\\
\midrule
\cite{TQC:MorPorYam24}& Classical& &$n\cdot \poly(\secp)$& $n\cdot \poly(\secp)$& LWE\\
\midrule
\cite{TQC:MorPorYam24}& Quantum& &$n\cdot \poly(\secp)$& $n\cdot \poly(\secp)$& OWGA\\
\midrule
Ours& Classical&\checkmark &$\poly(\secp)$&$\poly(\secp)$& SIS\\
\bottomrule
\end{tabular}\\
$n$: the number of issued signatures, $\secp$: the security parameter.
\end{table}

\subsection{Our Results}\label{sec-our-results}

In this work, we propose a new unified framework for realizing secure key leasing schemes, and solve the above questions affirmatively.
More specifically, we propose a simple framework for constructing secure key leasing schemes via the certified deletion property of BB84 states.
Based on our framework, we obtain the following schemes.
\begin{itemize}
\item PKE-SKL with classical revocation based on any IND-CPA secure PKE scheme.
\item PRF-SKL with classical revocation based on OWFs.
\item DS-SKL with classical revocation based on the 
short integer solution (SIS) assumption. Our scheme has static signing keys, i.e., the states of a signing key almost does not change before and after signing. In particular, this implies that the sizes of signing keys and signatures are independent of the number of issued signatures.    
\end{itemize}
In addition, all of our schemes remain secure even if a deletion verification key used by the deletion verification algorithm is leaked after the adversary submits a valid certificate of deletion. We call it security against verification key revealing attack (VRA security).
To our knowledge, all prior constructions are totally broken in this setting.   
Moreover, in our view, our security proofs are much simpler than those for existing schemes. 
Comparisons among existing and our constructions can be found in \Cref{table:pkeskl,table:prfskl,table:dsskl}. 

\paragraph{On security definitions.} 
We introduce new security definitions for PKE-SKL, PRF-SKL, and DS-SKL that capture VRA security explained above. While ours are not implied by any of existing security definitions, some of existing definitions are incomparable to ours. For clarity, we summarize relationships among our and existing security definitions. 
\begin{description}
\item[\bf PKE-SKL:] Our security definition, called IND-VRA security, implies the security definition from \cite{EC:AKNYY23}, known as IND-KLA security.\footnote{Strictly speaking, the implication holds only when the message space is polynomial-size, but this is not an issue in terms of existential implication. That is, if we have an IND-VRA secure PKE-SKL scheme for single-bit messages, then the scheme also satisfies IND-KLA security, and its parallel repetition yields an IND-VRA secure PKE-SKL scheme for arbitrarily long messages.
} 
While IND-VRA security does not imply the definition in \cite{TCC:AnaPorVai23,AnaHuHua24}, referred to as APV-security, this is because APV-security requires ciphertext pseudorandomness rather than a more standard indistinguishability-based definition. If we consider a slightly weakened variant of APV-security with an indistinguishability-based definition, then it would be implied by IND-VRA security. 
\ifnum\llncs=0 See the last paragraph of \Cref{sec:pke-skl-defs} and \Cref{sec:comparison_def} for more details. \fi
\item[\bf PRF-SKL:] Our security definition, called PR-VRA security, implies the security definition of \cite{TCC:AnaPorVai23,AnaHuHua24}.
\item[\bf DS-SKL:] Our security definition, called RUF-VRA security, is incomparable to the security definition of \cite{TQC:MorPorYam24}. In fact, their security definition is impossible to achieve by a scheme that has static signing keys as in our scheme. Thus, we introduce a new security definition that is achievable by schemes with static signing keys, which resembles the security definition of copy-protection of signing keys in \cite{TCC:LLQZ22}. 
\ifnum\llncs=0 See also \Cref{rem:comparison_DS}. \fi   
\end{description}
\subsection{More on Related Works}
\paragraph{Secure software leasing} Ananth and La Placa \cite{EC:AnaLaP21} introduced a notion called secure software leasing (SSL).
SSL is a relaxed variant of quantum copy protection~\cite{aaronson2009quantum} for general functions.
In this primitive, a lessor who has a program can generate a quantum state encoding the program and can lease it to a lessee.
The security notion of SSL roughly guarantees that once the lessee deletes it correctly with the designated deletion process, the lessee can no longer compute the program.
They provided the SSL scheme supporting a subclass of evasive functions\footnote{Evasive functions is a class of functions for which it is hard to find an accepting input given only black-box access to the function.} based on public key quantum money and the hardness of the LWE assumption.
Subsequently, Coladangelo, Majenz, and Poremba~\cite{ARXIV:ColMajPor20} and Kitagawa, Nishimaki, and Yamakawa~\cite{TCC:KitNisYam21} showed more positive results on SSL.

However, the main focus of \cite{EC:AnaLaP21} was impossibility on quantum copy protection for general functions and its relaxed variants, thus they defined the security notion for SSL with a highly strong restriction, which makes their impossibility result strong.
Roughly speaking, SSL provides the security guarantee only when the lessee tries to compute the leased program on a correct platform designated by the lessor, which leaves a possibility that the lessee can keep computing the program with other platforms even after deleting the program state.
The restriction is fatal especially when we consider the notion in cryptographic applications.
Then, Kitagawa and Nishimaki~\cite{AC:KitNis22} introduced the notion of secure key leasing as a relaxation of unclonable cryptographic key (that is, quantum copy protection for cryptographic functionalities), removing the restriction in the security notion of SSL.

\paragraph{Unclonable cryptographic key.} 
Georgiou and Zhandry~\cite{EPRINT:GeoZha20} initiated the study of PKE with unclonable secret  keys and construct it based on very strong assumptions including one-shot signatures and extractable witness encryption. 
Coladangelo, Liu, Liu, and Zhandry~\cite{C:CLLZ21} constructed PKE with unclonable secret keys and PRF with unclonable evaluation keys assuming sub-exponentially secure iO and the hardness of the LWE assumption.
Liu, Liu, Qian, and Zhandry~\cite{TCC:LLQZ22} extended the constructions by \cite{C:CLLZ21} into ones with bounded collusion resistance.
Bounded collusion resistant unclonability roughly guarantees that given $k$ quantum secret key state for a-priori bounded $k$, any adversary cannot generate $k+1$ copies of the secret key state.
The paper also proposed a construction of digital signature with bounded collusion resistant unclonable signing key assuming sub-exponentially secure iO and the hardness of the LWE assumption.
Recently,  {\c{C}}akan and Goyal\cite{CG24} constructed PKE with fully collusion resistant unblonable secret key from a similar set of assumptions.\footnote{More precisely, they constructed public key functional encryption with fully collusion resistant unblonable secret key.}

Unclonable cryptographic key directly inherits the unclonability of quantum states and seems more desirable compared to secure key leasing.
However, unclonable cryptographic key would be too strong to be realized from standard assumptions.
As we saw above, all existing constructions of it requires strong assumption of iO that currently does not have any instantiation based on well-studied post quantum assumptions.
In fact, the use of iO is somewhat inherent with the current techniques from the fact that they imply publicly verifiable quantum money\footnote{More precisely, PKE and digital signature with unclonable secret key imply publicly verifiable quantum money, but it is not clear if PRF with unclonable evaluation key does.} whose realization without post quantum secure iO is one of the biggest challenge in quantum cryptography.
The situation encourages to study relaxed notions including secure key leasing that is still useful and can be realized from standard assumptions.

\paragraph{Secure key leasing with classical communication.} \takashi{I added this paragraph.}
Chardouvelis, Goyal, Jain, and Liu~\cite{EPRINT:CGJL23} constructed PKE-SKL with a classical lessor, i.e., a leased quantum secret key can be generated by interaction between quantum lessee and classical lessor over classical channel and revocation can also be done classically. Their construction is based on the LWE assumption. 
It is unlikely to have PKE-SKL with a classical lessor only assuming the existence of PKE since PKE-SKL with a classical lessor immediately implies proofs of quantumness~\cite{JACM:BCMVV21} but it is a major open problem to construct proofs of quantumness from PKE.

\paragraph{Certified deletion.} \takashi{I added this paragraph.}
An encryption scheme with certified deletion~\cite{TCC:BroIsl20} enables us to delete a quantum ciphertext in a verifiable manner. 
Certified deletion is conceptually related to secure key leasing since secure key leasing can be regarded as certified deletion of keys rather than ciphertexts. 
Broadbent and Islam~\cite{TCC:BroIsl20} constructed an unconditionally secure one-time SKE with certified deletion. The follow-up works 
\cite{AC:HMNY21,ITCS:Poremba23,C:BarKhu23,EC:HKMNPY24,EC:BGKMRR24} study certified deletion for more advanced types of encryption including PKE, attribute-based encryption, functional encryption, quantum fully homomorphic encryption, and witness encryption. 
The works by \cite{TCC:KitNisYam23,TCC:BKMPW23} extend the above schemes to publicly verifiable ones from minimal assumptions. 

\ifnum\llncs=1
\subsection{Organization} 
In \Cref{sec:overview}, we give an overview of our new framework and how we apply it to construct new PKE-SKL, PRF-SKL, and DS-SKL schemes.  
Due to space limitations, the remainder of the main body focuses on the construction of PKE-SKL, while the constructions of PRF-SKL and DS-SKL are deferred to the full version.  
\fi

\section{Technical Overview}\label{sec:overview}
We provide the technical overview.

\subsection{Definition of Secure Key Leasing}\label{sec-overview-def-skl}
We first introduce the definition of secure key leasing.
Below, we use calligraphic font to denote that the object is a quantum algorithm or quantum state (e.g., $\qKG$ and $\qstate{q}$).

\paragraph{Syntax of secure key leasing with classical revocation.}
The focus of this work is secure key leasing primitives having classical revocation.
The syntax of X-SKL with classical revocation for a classical cryptographic primitive X$\in\{\textrm{PKE,PRF,DS}\}$ is the same as that of X except for the following.
\begin{itemize}
\item The key generation algorithm $\qKG$ is a quantum algorithm that outputs a quantum secret key $\qsk$ instead of a classical one, together with a classical secret deletion verification key $\dvk$. If X is PKE (resp. DS), it also outputs public encryption (resp. signature verification) key. Moreover, if X is PRF, it also outputs a classical master secret key $\msk$.
\item X-SKL has the deletion algorithms $\qDel$ and the deletion verification algorithm $\DelVrfy$ in addition to the algorithms that form X. $\qDel$ takes as input a quantum secret key $\qsk^\prime$, and outputs a classical deletion certificate $\cert$. $\DelVrfy$ takes as input a deletion verification key $\dvk$ and a deletion certificate $\cert$, and outputs $\top$ or $\bot$.

\item In addition to the correctness as a cryptographic primitive X, X-SKL satisfies deletion verification correctness. It guarantees that for a pair of a quantum secret key $\qsk$ and the corresponding deletion verification key $\dvk$ output by $\qKG$, if we generate $\cert\gets\qDel(\qsk)$, then we have $\DelVrfy(\dvk,\cert)=\top$, which means a correctly generated deletion certificate always passes the verification.
\end{itemize}

\paragraph{Security notion for secure key leasing.}
The security of the secure key leasing scheme roughly guarantees that once the adversary who is given a quantum secret key $\qsk$ outputs a deletion certificate that is accepted by $\DelVrfy(\dvk,\cdot)$, the adversary can no longer perform the task that requires $\qsk$.
As stated in \cref{sec-our-results}, our schemes remain secure even if the secret deletion verification key is given to the adversary after the adversary outputs a deletion certificate.
We introduce security definitions that capture such attacks and call it
security against verification revealing attacks (VRA security).
We prove that VRA security is stronger than security under the existence of the verification oracle for secure key leasing schemes with classical revocation, which was studied in a previous work~\cite{EC:AKNYY23}. \ifnum\llncs=0(See \Cref{sec:comparison_def}.)\fi

We can define both unpredictability style security definition and indistinguishability style one for PKE-SKL and PRF-SKL, though it is inherently unpredictability style for DS-SKL.
For PKE-SKL, the former can be upgraded into the latter by using the standard (quantum) Goldreich-Levin technique as shown in \cite{EC:AKNYY23}.
We show that a similar upgrade is possible for PRF-SKL.
Thus, in this work, we primarily focus on constructing secure key leasing primitives satisfying unpredictability style security notions.

The unpredictability style VRA security experiment for X-SKL is abstracted by using a predicate $P$ that determines the adversary's goal.
It is described as follows.

\begin{enumerate}
\item The challenger generates $\qsk$ and $\dvk$ by executing the key generation algorithm $\qKG$ and sends $\qsk$ to $\qA$. (If X is PKE or DS, $\qA$ is also given the corresponding public key  output by $\qKG$.)
\item $\qA$ outputs a classical string $\cert$. If $\DelVrfy(\dvk,\cert)=\bot$, the experiment ends. Otherwise, the challenger generates a challenge $\chall$ and corresponding auxiliary information $\aux$, and sends $\dvk$ and $\chall$ to $\qA$. Note that $\aux$ is used by the predicate $P$ to determine whether $\qA$ wins the experiment or not.
\item $\qA$ outputs $\ans$.
\end{enumerate}
$\qA$ wins the experiment if both $\DelVrfy(\dvk,\cert)=\top$ and $P(\aux,\ans)=1$ holds simultaneously.
For example, if X is PKE, $\chall$ is a ciphertext of a uniformly random message $m$, $\aux$ is $m$, and $\ans$ is the guess $m^\prime$ for $m$ by $\qA$. Also, $P(\aux=m,\ans=m^\prime)$ outputs $1$ if and only if $m^\prime=m$. 
It guarantees that the winning probability of any QPT $\qA$ is negligible in the security parameter.


\subsection{Certified Deletion Property of BB84 States.}
As stated above, we propose a new framework for achieving secure key leasing schemes.
Its security is based on the certified deletion property of BB84 states proven in \cite{EPRINT:BehSatShi21}.
More specifically, we use its variant with ``verification key revealing'' that can be reduced to the original version.
It is defined as follows.

For strings $x\in\bit^n$ and $\theta\in\bit^n$, we let $\ket{x^\theta} \seteq H^{\theta[1]}\ket{x[1]}\tensor\cdots \tensor H^{\theta[n]}\ket{x[n]}$, where  
$H$ is the Hadamard operator, and $x[i]$ and $\theta[i]$ are the $i$-th bits of $x$ and $\theta$, respectively.
A state of the form $\ket{x^\theta}$ is called a BB84 state~\cite{BB84}. 
Consider the following experiment.
\begin{enumerate}
\item The challenger picks uniformly random strings $x\in\bit^n$ and $\theta\in\bit^n$ and sends the BB84 state $\ket{x^\theta}$ to $\qA$. 
\item $\qA$ outputs a classical string $y\in\bit^n$.
\item $\qA$ is then given $\theta$ and $(x[i])_{i\in[n]:\theta[i]=1}$, and outputs a classical string $z\in\bit^n$.
\end{enumerate}
$\qA$ wins the experiment if $y[i]=x[i]$ for every $i\in[n]$ such that $\theta[i]=1$ and $z[i]=x[i]$ for every $i\in[n]$ such that $\theta[i]=0$.
The certified deletion property of BB84 states guarantees that the winning probability of any adversary $\qA$ is negligible in $n$.

The string $y$ output by $\qA$ can be seen as a deletion certificate that is accepted if it agrees with $x$ in the Hadamard basis positions. 
The verification of the deletion certificate is done by the deletion verification key $(\theta, (x[i])_{i\in[n]:\theta[i]=1})$.
Note that $\qA$ can generate an accepting deletion certificate by just measuring the given BB84 state $\ket{x^\theta}$ in the Hadamard basis.
Then, the property says that once $\qA$ outputs an accepting deletion certificate, it is impossible for $\qA$ to obtain the value of $x$ in all of the computational basis positions even given the deletion verification key $(\theta,(x[i])_{i\in[n]:\theta[i]=1})$.

\subsection{The General Framework}\label{sec-overview-framework}
We now move on to our general framework for realizing secure key leasing schemes with VRA security.\footnote{We do not provide the framework formally in the main body, and show each construction directly, because introducing it in the main body requires defining many abstracted notions formally that could prevent readers from verifying the correctness of our schemes efficiently. Nevertheless, we choose to provide the framework in this overview because we believe it is highly helpful to acquire the unified design strategy lying behind all of our constructions.}
The way we construct key generation, deletion, and deletion verification algorithms is common among all of our constructions. 
\takashi{It may be better to mention that we actually do not explicitly give the "framework" in the main body.}\fuyuki{I added a footnote.}

\paragraph{Key generation algorithm.}
The key generation algorithm $\qKG$ first generates $2n$ classical secret keys $(\sk_{i,b})_{i\in[n],b\in\bit}$ and 
$(x,\theta)\gets\bit^n\times\bit^n$. 
It then applies the map $M_i$ that acts as
\begin{align}
\ket{b}\ket{c} \ra \ket{b}\ket{c \oplus \sk_{i,c}}
\end{align}
to the $i$-th qubit of the BB84 state $\ket{x^\theta}$ and $\ket{0\cdots0}$, and obtains a quantum state $\qsk_i$ for every $i\in[n]$.
Namely, we have
    \begin{align}
    \qsk_i\seteq
    \begin{cases}
    \ket{x[i]}\ket{\sk_{i,x[i]}} & (\textrm{if~}\theta[i]=0)\\
    \frac{1}{\sqrt{2}}\left(\ket{0}\ket{\sk_{i,0}}+(-1)^{x[i]}\ket{1}\ket{\sk_{i,1}}\right) & (\textrm{if~}\theta[i]=1).
    \end{cases}
    \end{align}
The resulting secret key is set to $\qsk:=(\qsk_1,\ldots,\qsk_n)$.
The corresponding deletion verification key is set to $\dvk\seteq (\theta,(x[i])_{i\in[n]:\theta[i]=1},(\sk_{i,0},\sk_{i,1})_{i\in[n]:\theta[i]=1})$.

We explain how the public keys of PKE-SKL and DS-SKL and the master secret key of PRF-SKL are generated when we provide concrete constructions.


\paragraph{Deletion and deletion verification algorithms.}
The deletion algorithm $\qDel$, given the above $\qsk=(\qsk_1,\ldots,\qsk_n)$, measures $\qsk_i$ in the Hadamard basis and obtains $(e_i,d_i)$ for every $i\in[n]$. The resulting deletion certificate is set to $\cert:=(e_i,d_i)_{i\in[n]}$. The deletion verification algorithm, given a deletion verification key $\dvk\seteq (\theta,(x[i])_{i\in[n]:\theta[i]=1},\allowbreak (\sk_{i,0},\sk_{i,1})_{i\in[n]:\theta[i]=1})$ and a deletion certificate $\cert^\prime=(e^\prime_i,d^\prime_i)_{i\in[n]}$, outputs $\top$ if and only if $e^\prime_i=x[i]\oplus d^\prime_i \cdot (\sk_{i,0}\oplus\sk_{i,1})$ holds for every $i\in[n]$ such that $\theta[i]=1$. We can confirm the deletion verification correctness of the construction by simple calculations.

\fuyuki{Yamakawa san's suggestion is adopted here.}
Given the construction of the above three algorithms, we show that the construction satisfies a nice property that we can convert a valid deletion certificate $\cert$ for the quantum secret key $\qsk$ to a valid deletion certificate $y$ for the BB84 state $\ket{x^\theta}$ by using $2n$ secret keys $(\sk_{i,b})_{i\in[n],b\in\bit}$.
Suppose $\cert^\prime=(e^\prime_i,d^\prime_i)_{i\in[n]}$ is a valid deletion certificate for $\qsk$.
Then, we have $e^\prime_i=x[i]\oplus d^\prime_i \cdot (\sk_{i,0}\oplus\sk_{i,1})$  for every $i\in[n]$ such that $\theta[i]=1$. Thus, if we set 
$y[i]=e^\prime_i\oplus d^\prime_i \cdot (\sk_{i,0}\oplus\sk_{i,1})$ for every $i\in[n]$, then we have $y[i]=x[i]$ for all every $i\in[n]$ such that $\theta[i]=1$, and thus $y$ is a valid deletion certificate for the BB84 state $\ket{x^\theta}$.

\paragraph{Special challenge generation.}
In our secure key leasing schemes, we construct the remaining algorithms with the property that allows us to reduce their security to the certified deletion property of BB84 states.
The property is that we can generate a special challenge $\chall_\theta$ and corresponding special auxiliary information $\aux_{x,\theta}$ with the following two requirements, where $\chall_\theta$ can depend on $\theta$ but not on $x$ and $\aux_{x,\theta}$ can depend on both $x$ and $\theta$ as the notation suggests.

\begin{description}
\item[Indistinguishability from the normal challenge:] Let 
\ifnum\llncs=0
$\key:=(x,\theta,(\sk_{i,x[i]})_{i\in[n]:\theta[i]=0},(\sk_{i,0},\sk_{i,1})_{i\in[n]:\theta[i]=1})$.
\else 
$$\key:=(x,\theta,(\sk_{i,x[i]})_{i\in[n]:\theta[i]=0},(\sk_{i,0},\sk_{i,1})_{i\in[n]:\theta[i]=1}).$$ \fi
It requires that the normal challenge tuple $(\chall,\aux)$ and the special challenge tuple $(\chall_\theta,\aux_{x,\theta})$ are computationally indistinguishable given $\key$.

Suppose we simulate the VRA security experiment for a QPT adversary $\qA$ until it outputs deletion certificate $\cert$, gives special challenge $\chall_\theta$ to $\A$, and obtains $\qA$'s final output $\ans^\prime$.
Note that we can generate $\qsk$ and $\dvk$ required for the simulation from $\key$.
Thus, indistinguishability from the normal challenge property implies that $\Pr[\DelVrfy(\dvk,\cert)=\top \land P(\aux_{x,\theta},\ans^\prime)=1]$ is negligibly close to $\qA$'s winning probability in the original experiment, where $\aux_{x,\theta}$ is the special auxiliary information corresponding to $\chall_\theta$. 
\item[Extractability for the computational basis positions:] It requires that there exists an extractor $\ExtComp_\theta$ with the following property. 
Suppose we simulate the VRA security experiment for a QPT adversary $\qA$ until it outputs deletion certificate $\cert$, gives $\chall_\theta$ to $\A$, and obtains $\qA$'s final output $\ans^\prime$. Then, if $P(\aux_{x,\theta},\ans^\prime)=1$ holds for the special auxiliary information $\aux_{x,\theta}$ corresponding to $\chall_\theta$, $\ExtComp_\theta$ can extract $(x[i])_{i\in[n]:\theta[i]=0}$ almost certainly from $\ans^\prime$.
As the notation suggests, $\ExtComp_\theta$ can depend on $\theta$ but not on $x$.
\end{description}
The reason for the constraints for the dependence on $x$ and $\theta$ posed to $\chall_\theta$, $\aux_{x,\theta}$ and $\ExtComp_\theta$ will be clear in the next paragraph.

\paragraph{Reduction to the certified deletion property of BB84 states.}
We now sketch the reduction algorithm $\qB$ attacking the certified deletion property of BB84 states using an adversary $\qA$ for our secure key leasing scheme.
\begin{enumerate}
\item Given a BB84 state $\ket{x^\theta}$, $\qB$ samples $2n$ secret keys $(\sk_{i,b})_{i\in[n],b\in\bit}$, and generates $\qsk$ from $\ket{x^\theta}$ and those $2n$ secret keys as the key generation algorithm designates. $\qB$ then sends $\qsk$ to $\qA$.
\item When $\qA$ outputs $\cert^\prime$, $\qB$ converts it into a deletion certificate $y$ for the BB84 state using $(\sk_{i,b})_{i\in[n],b\in\bit}$ as described above and outputs $y$.
\item Given $(\theta,(x[i])_{i\in[n]:\theta[i]=1})$, $\qB$ sets $\dvk=(\theta,(x[i])_{i\in[n]:\theta[i]=1},(\sk_{i,0},\sk_{i,1})_{i\in[n]:\theta[i]=1})$. $\qB$ also generates a special challenge $\chall_\theta$ and sends $\dvk$ and $\chall_\theta$ to $\qA$. When $\qA$ outputs $\ans^\prime$, $\qB$ outputs $(z[i])_{i\in[n]:\theta[i]=0}\gets\ExtComp_\theta(\ans^\prime)$. $\qB$ can generate $\chall_\theta$ and executes $\ExtComp_\theta$ since they do not depend on $x$ though they might depend on $\theta$.
\end{enumerate}

Suppose $\qA$'s winning probability in the original security experiment is $\epsilon$.
From the indistinguishability from the normal challenge of the special challenge, we have $\Pr[\DelVrfy(\dvk,\cert)=\top \land P(\aux_{x,\theta},\ans^\prime)=1]=\epsilon-\negl(\secp)$, where $\aux_{x,\theta}$ is the special auxiliary information corresponding to $\chall_\theta$.
$\qB$ can output a valid deletion certificate for the given BB84 state $\ket{x^\theta}$ if the event $\DelVrfy(\dvk,\cert)=\top$ occurs as we checked above. Moreover, from the extractability for the computational basis positions, $\qB$ can also outputs $(x[i])_{i\in[n]:\theta[i]=0}$ with overwhelming probability if the event $P(\aux_{x,\theta},\ans^\prime)=1$ occurs.
This means $\qB$ wins the certified deletion experiment for BB84 states with probability at least $\epsilon-\negl(\secp)$, which means $\epsilon$ is negligible and the proof is done.
Note that $\qB$ do not need to check the condition $P(\aux_{x,\theta},\ans^\prime)=1$ and thus it is not an issue if $\aux_{x,\theta}$ depends on $x$ that is not given to $\qB$ throughout the experiment.

Therefore, all we have to do is to show how to construct the algorithms other than the key generation, deletion, and deletion verification algorithms so that the resulting construction meets the required property.
We below show it for each of our PKE-SKL, PRF-SKL, and DS-SKL.

\paragraph{Idea behind VRA security.}
Here, we briefly explain why our framework enables us to achieve VRA security, unlike existing constructions. 

Our construction shares some similarities with that of PKE-SKL in \cite{EC:AKNYY23}. In their scheme, the quantum secret key is defined as  
$
\qsk_i=\frac{1}{\sqrt{2}}\left(\ket{0}\ket{\sk_{i,0}}+\ket{1}\ket{\sk_{i,1}}\right)
$ 
for all \( i \in [n] \). The deletion algorithm follows the same approach as ours: measuring \( \qsk_i \) in the Hadamard basis yields an outcome \( (e_i, d_i) \) for each \( i \in [n] \), and the deletion certificate is set as 
$
\cert \seteq (e_i, d_i)_{i \in [n]}.
$ 
This ensures that 
$
e_i = d_i \cdot (\sk_{i,0} \oplus \sk_{i,1})
$ 
for all \( i \in [n] \). Consequently, their deletion verification algorithm checks whether this condition holds for all \( i \in [n] \). However, to perform this verification, the deletion verification key \( \dvk \) must include \( (\sk_{i,0}, \sk_{i,1}) \) for every \( i \in [n] \), effectively revealing the entire secret of the scheme. As a result, once \( \dvk \) is leaked, no security remains.

In contrast, in our construction, the deletion verification algorithm checks validity of $(e_i,d_i)$ only for $i$ such that $\theta[i]=1$. Our key idea is that this partial check is sufficient for a reduction to the certified deletion property of BB84 states. 
Consequently,  \( \dvk \) contains \( (\sk_{i,0}, \sk_{i,1}) \) only for $i$ such that $\theta[i]=1$, leaving the remaining secret keys undisclosed even if \( \dvk \) is leaked. This is the key reason why our scheme achieves VRA security.

\subsection{PKE-SKL}
A PKE-SKL scheme has the encryption algorithm $\Enc$ and the decryption algorithm $\Dec$ in addition to $\qKG$, $\qDel$, and $\DelVrfy$.
$\Enc$ and $\qDec$ have the same input-output behavior with those of standard PKE schemes except that $\qDec$ takes as input a quantum secret key instead of a classical one.
A PKE-SKL scheme also satisfies decryption correctness similarly to standard PKE.
The decryption correctness also guarantees the reusability of the quantum secret key (see \Cref{rem:reusability}).

\paragraph{Unpredictability style VRA security definition.}
As stated above, it is obtained by setting $\chall$ as a ciphertext $\ct^*$ of a random message $m^*$, $\aux$ as $m^*$, $\ans$ as the guess $m^\prime$ for $m^*$ by the adversary $\qA$, and $P$ as the predicate such that given $\aux=m^*$ and $\ans=m^\prime$, it outputs $1$ if and only if $m^\prime=m^*$, in the security experiment provided in \cref{sec-overview-def-skl}.

\paragraph{Construction.}
We construct $\PKESKL=\PKESKL.(\qKG,\Enc,\qDec,\qDel,\DelVrfy)$ for the message space $\bit^n$ using a classical PKE scheme $\PKE=\PKE.(\KG,\Enc,\Dec)$ for the message space $\bit$.
\begin{description}
\item[$\PKESKL.\KG$:] It generates $(\pk_{i,b},\sk_{i,b})\gets\PKE.\KG(1^\secp)$ for every $i\in[n]$ and $b\in\bit$. It also generates $(x,\theta)\gets\bit^n\times\bit^n$. It then generates $\qsk$ and $\dvk$ using $2n$ secret keys $(\PKE.\sk_{i,b})_{i\in[n],b\in\bit}$ and $(x,\theta)$ as described in \cref{sec-overview-framework}. It also sets $\pk:=(\PKE.\pk_{i,b})_{i\in[n],b\in\bit}$.
\item[\rm $\PKESKL.\Enc$ and $\PKESKL.\Dec$:] Given a public key $\pk=(\PKE.\pk_{i,b})_{i\in[n],b\in\bit}$, and a message $m\in\bit^n$, $\PKESKL.\Enc$ generates $\PKE.\ct_{i,b}\gets\PKE.\Enc(\PKE.\pk_{i,b},m[i])$ for every $i\in[n]$ and $b\in\bit$. The resulting ciphertext is $\ct:=(\PKE.\ct_{i,b})_{i\in[n],b\in\bit}$.

$\PKESKL.\Dec$ takes $\qsk=(\qsk_1,\ldots,\qsk_n)$ and $\ct:=(\PKE.\ct_{i,b})_{i\in[n],b\in\bit}$ and acts as follows.
Let $D_i$ be a unitary that has $\PKE.\ct_{i,0}$ and $\PKE.\ct_{i,1}$ hardwired, and maps $\ket{b}\ket{\PKE.\sk}\ket{c}$ to $\ket{b}\ket{\PKE.\sk}\ket{c\oplus\PKE.\Dec(\sk,\PKE.\ct_{i,b})}$.
$\PKESKL.\Dec$ applies $D_i$ to $\qsk_i\tensor\ket{0}$ and obtains $m_i$ by measuring the last register for every $i\in[n]$, and output $m_1\|\cdots\|m_n$.
The construction satisfies the decryption correctness.
For the proof for decryption correctness, see \cref{sec:const-PKE-SKL-classical-certificate}.
\end{description}
$\PKESKL.\qDel$ and $\PKESKL.\DelVrfy$ are constructed as described in \cref{sec-overview-framework}.

\paragraph{Special challenge and auxiliary information.}
In the construction of $\PKESKL$, the special challenge $\chall_\theta=\hatct=(\PKE.\hatct_{i,b})_{i\in[n],b\in\bit}$ is generated as 
\begin{align}
\PKE.\hatct_{i,b}\gets
\begin{cases}
\PKE.\Enc(\PKE.\ek_{i,b}, m^*[i]\oplus b) & (\textrm{if}~~\theta[i]=0)\\
\PKE.\Enc(\PKE.\ek_{i,b}, m^*[i]) & (\textrm{if}~~\theta[i]=1),
\end{cases}
\end{align}
where $m^*\gets\bit^n$.
Let $x_c$ be a string such that $x_c[i]=x[i]$ if $\theta[i]=0$ and $x_c=0$ otherwise and $\hatm=m^*\oplus x_c$.
The special auxiliary information $\aux_{x,\theta}$ corresponding to $\chall_\theta$ is $\hatm$.

Assuming $\PKE$ is IND-CPA secure, we prove the indistinguishability from the normal challenge, that is, we prove the above $(\chall_\theta=\hatct,\aux_{x,\theta}=\hatm)$ is computationally indistinguishable from the normal challenge distribution $(\chall=\PKESKL.\Enc(\pk,m^*),\aux=m^*)$, where $m^*\gets\bit^n$, given $\key=(x,\theta,(\PKE.\sk_{i,x[i]})_{i\in[n]:\theta[i]=0},\allowbreak(\PKE.\sk_{i,b})_{i\in[n],b\in\bit:\theta[i]=1})$.
The normal challenge distribution is the same as the distribution $(\PKESKL.\Enc(\pk,\hatm),\hatm)$ since $\hatm=m^*\oplus x_c$ is distributed uniformly at random. 
Moreover, the only difference between the latter distribution and the special challenge distribution is that in the former, $\PKE.\ct_{i,1-x[i]}$ for every $i\in[n]$ such that $\theta[i]=0$ is generated as $\PKE.\Enc(\PKE.\pk_{i,1-x[i]},m^*[i]\oplus x[i])$ but in the latter, it is generated as $\PKE.\Enc(\PKE.\pk_{i,1-x[i]},m^*[i]\oplus (1-x[i]))$.
Then, the proof is completed from the IND-CPA security of $\PKE$ since $\PKE.\sk_{i,1-x[i]}$ for every $i\in[n]$ such that $\theta[i]=0$ is not given to the distinguisher.

We finally check the extractability for the computational basis position.
Suppose $\ans^\prime=m^\prime$ satisfies $P(\aux_{x,\theta},\ans^\prime)=1$. Then, we have $m^\prime[i]=\hatm[i]=m^*[i]\oplus x[i]$ holds for every $i\in[n]$ such that $\theta[i]=0$.
We can easily extract $x[i]$ by computing $m^\prime[i]\oplus m^*[i]$ for every $i\in[n]$ such that $\theta[i]=0$.

\subsection{PRF-SKL}

We move on to our PRF-SKL construction.
As stated before, we first construct one with unpredictability style security, that is, unpredictable function with secure key leasing (UPF-SKL), and then upgrade it into PRF-SKL via the Goldreich-Levin technique.
We below review the syntax and security notion for UPF-SKL, and then provide the construction of our UPF-SKL and explain its special challenge property.

A UPF-SKL scheme has two evaluation algorithms $\Eval$ and $\qLEval$ in addition to $\qKG$, $\qDel$, and $\DelVrfy$.
$\Eval$ (resp. $\qLEval$) takes master secret key $\msk$ (resp. quantum secret key $\qsk$) and an input $\prfinp\in\Domprf$ and outputs $\prfout\in\Ranprf$, where $\Domprf$ and $\Ranprf$ are the domain and range, respectively.
For evaluation correctness, we require that $\Eval(\msk,\prfinp)=\qLEval(\qsk,\prfinp)$ holds with overwhelming probability for any input $\prfinp\in\Domprf$, where $(\msk,\qsk,\dvk)\gets\qKG(1^\secp)$.
The evaluation correctness also guarantees the reusability of the quantum secret key.

\paragraph{Unpredictability style VRA security definition.}
It is obtained by setting $\chall$ as a random input $\prfinp^*$, $\aux$ as $\prfout^*=\Eval(\msk,\prfinp^*)$, $\ans$ as the guess $\prfout^\prime$ for $\prfout^*$ by the adversary $\qA$, and $P$ as the predicate that takes $\aux=\prfout^*$ and $\ans=\prfout^\prime$ and outputs $1$ if and only if $\prfout^\prime=\prfout^*$, in the security experiment provided in \cref{sec-overview-def-skl}.

\paragraph{Building block: two-key equivocal PRF.}
We use two-key equivocal PRF (TEPRF) that can be constructed from one-way functions~\cite{C:HJOSW16} as the building block of our UPF-SKL.
A TEPRF for the domain $\Domprf$ consists of two algorithms $\TEPRF.\KG$ and $\TEPRF.\Eval$.
$\TEPRF.\KG$ takes as input a security parameter $1^\secp$ and an input $\tlprfinp\in\Domprf$ and outputs two keys $\TEPRF.\sk_0$ and $\TEPRF.\sk_1$.
$\TEPRF.\Eval$ takes as input a secret key $\TEPRF.\sk$ and an input $\prfinp\in\Domprf$, and outputs $\prfout$.
As correctness, it satisfies the following:  
for any pair of keys $(\TEPRF.\sk_0,\TEPRF.\sk_1)$ generated by $\TEPRF.\KG(1^\secp,\tlprfinp)$, it hold that 
$\TEPRF.\Eval(\TEPRF.\sk_0,\tlprfinp)\ne \TEPRF.\Eval(\TEPRF.\sk_1,\tlprfinp)$
and 
$\TEPRF.\Eval(\TEPRF.\sk_0,\prfinp)=\TEPRF.\Eval(\TEPRF.\sk_1,\prfinp)$ for any  $\prfinp\ne\tlprfinp$.
It also satisfies the security notion called differing point hiding that guarantees that for any QPT adversary and $\tlprfinp\in\Domprf$, if the adversary is given only either one of $(\TEPRF.\sk_0,\TEPRF.\sk_1)\gets\TEPRF.\KG(1^\secp,\tlprfinp)$, the adversary cannot obtain any information of $\tlprfinp$.

\paragraph{Construction.}
We construct $\UPFSKL=\UPFSKL.(\qKG,\Eval,\qLEval,\qDel,\DelVrfy)$ for the domain $\bit^{n\cdot\secp}$ using a TEPRF $\TEPRF=\TEPRF.(\KG,\Eval)$ for the domain $\bit^\secp$.
\begin{description}
\item[$\UPFSKL.\qKG$:] It generates $\tlprfinp_i\gets\bit^\ell$, and $(\TEPRF.\sk_{i,0},\TEPRF.\sk_{i,1})\gets\TEPRF.\KG(1^\secp,\tlprfinp_i)$ for every $i\in[n]$. It also generates $(x,\theta)\gets\bit^n\times\bit^n$. It then generates $\qsk$ and $\dvk$ using $2n$ secret keys $(\TEPRF.\sk_{i,b})_{i\in[n],b\in\bit}$ and $(x,\theta)$ as described in \cref{sec-overview-framework}.
It also sets $\msk=(\TEPRF.\sk_{i,0})_{i\in[n]}$. 
\item[\rm $\UPFSKL.\Eval$ and $\UPFSKL.\qLEval$:] Given a master secret key $\msk=(\TEPRF.\sk_{i,0})_{i\in[n]}$ and an input $\prfinp=\prfinp_1\|\cdots\|\prfinp_n$, where $\prfinp_i\in\bit^\secp$ for every $i\in[n]$, it computes $\prfout_{i}\gets\TEPRF.\Eval(\TEPRF.\sk_{i,0},\prfinp_i)$ for every $i\in[n]$, and outputs $\prfout:=\prfout_1\|\cdots\|\prfout_n$.

$\UPFSKL.\qLEval$ takes $\qsk=(\qsk_1,\ldots,\qsk_n)$ and $\prfinp=\prfinp_1\|\cdots\|\prfinp_n$ and acts as follows.
Let $E_i$ be a unitary that has $\prfinp_i$ hardwired, and maps $\ket{b}\ket{\TEPRF.\sk}\ket{c}$ to $\ket{b}\ket{\TEPRF.\sk}\ket{c\oplus\TEPRF.\Eval(\TEPRF.\sk,\prfinp_i)}$.
$\UPFSKL.\qLEval$ applies $E_i$ to $\qsk_i\tensor\ket{0\cdots0}$ and obtains $\prfout_i$ by measuring the last register for every $i\in[n]$, and outputs $\prfout:=\prfout_1\|\cdots\|\prfout_n$.
We can check that $\UPFSKL.\qLEval(\qsk,\prfinp)$ ouptuts $\UPFSKL.\Eval(\msk,\prfinp)$, unless $\prfinp_i=\tlprfinp_i$ holds for some $i\in[n]$, where $\tlprfinp_i$ is the input used when generating $(\TEPRF.\sk_{i,0},\TEPRF.\sk_{i,1})$.
The construction satisfies evaluation correctness.
\ifnum\llncs=0 For the formal proof of evaluation correctness, see \cref{sec:const-PRF-SKL-classical-certificate}. \fi
\end{description}
$\UPFSKL.\qDel$ and $\UPFSKL.\DelVrfy$ are constructed as described in \cref{sec-overview-framework}.

\paragraph{Special challenge and auxiliary information.}
In the construction of $\UPFSKL$, the special challenge $\chall_\theta=\hatprfinp=\hatprfinp_1\|\cdots\|\hatprfinp_n$ is generated as 
\begin{align}
\hatprfinp=
\begin{cases}
\tlprfinp_i & (\textrm{if}~~\theta[i]=0)\\
\prfinp_i^* & (\textrm{if}~~\theta[i]=1),
\end{cases}
\end{align}
where $\tlprfinp_i$ is the input used when generating $(\TEPRF.\sk_{i,0},\TEPRF.\sk_{i,1})$ and $\prfinp_i^*\gets\bit^\secp$ for every $i\in[n]$.
Let $\hatprfout_i\gets\TEPRF.\Eval(\TEPRF.\sk_{i,x[i]},\hatprfinp_i)$ for every $i\in[n]$ and $\hatprfout=\hatprfout_1\|\cdots\|\hatprfout_n$.
The special auxiliary information $\aux_{x,\theta}$ corresponding to $\chall_\theta$ is $\hatprfout$.

Based on the differing point hiding property of $\TEPRF$, we prove the indistinguishability from the normal challenge, that is, we prove that the above $(\chall_\theta=\hatprfinp,\aux_{x,\theta}=\hatprfout)$ is computationally indistinguishable from the normal challenge $(\chall=\prfinp^*,\aux=\prfout^*=\UPFSKL.\Eval(\msk,\prfinp^*))$, where $\prfinp^*\gets\bit^{\secp\cdot n}$, given $\key=(x,\theta,(\TEPRF.\sk_{i,x[i]})_{i\in[n]:\theta[i]=0},\allowbreak(\TEPRF.\sk_{i,b})_{i\in[n],b\in\bit:\theta[i]=1})$.
Given $\key$, we can generate $\prfout^*$ and $\hatprfout$ in exactly the same way respectively from $\prfinp^*$ and $\hatprfinp$.
Thus, what we need to prove is that $\prfinp^*$ and $\hatprfinp$ are computationally indistinguishable given $\key$.
This follows from the differing point hiding property of $\TEPRF$ and the fact that for every $i\in[n]$ such that $\theta[i]=0$, only $\TEPRF.\sk_{i,x[i]}$ is given and $\TEPRF.\sk_{i,1-x[i]}$ is completely hidden.

We finally check the extractability for the computational basis position.
Suppose $\ans^\prime=\prfout^\prime=\prfout^\prime_1\|\cdots\|\prfout^\prime_n$ satisfies $P(\aux_{x,\theta},\ans^\prime)=1$. Then, we have $\prfout^\prime_i=\TEPRF.\Eval(\TEPRF.\sk_{i,x[i]},\hatprfinp_i)$ holds for every $i\in[n]$ such that $\theta[i]=0$.
From the correctness of $\TEPRF$, we have $\TEPRF.\Eval(\TEPRF.\sk_{i,0},\hatprfinp_i)\ne\TEPRF.\Eval(\TEPRF.\sk_{i,1},\hatprfinp_i)$ for every $i\in[n]$ such that $\theta[i]=0$.
Thus, we can extract $x[i]$ by checking which one of $\TEPRF.\Eval(\TEPRF.\sk_{i,0},\hatprfinp_i)$ and $\TEPRF.\Eval(\TEPRF.\sk_{i,1},\hatprfinp_i)$ matches with $\prfout^\prime_i$ for every $i\in[n]$ such that $\theta[i]=0$.

\subsection{DS-SKL}
A DS-SKL scheme has the signing algorithm $\qSign$ and the signature verification algorithm $\SigVrfy$ in addition to $\qKG$, $\qDel$, and $\DelVrfy$.
$\qSign$ and $\SigVrfy$ have the same input-output behavior with those of standard digital signature schemes except that $\qSign$ takes as input a quantum secret key instead of a classical one.
We quire standard signature verification correctness for a DS-SKL scheme that guarantees that $\SigVrfy(\svk,m,\qSign(\sk,m))=\top$ with overwhelming probability for any message $m$, where $(\svk,\qsk,\dvk)\gets\qKG(1^\secp)$. We denote the signature verification key as $\svk$.
Differently from PKE-SKL and PRF-SKL (and UPF-SKL), the correctness does not imply reusability of the quantum secret key.
Thus, we separately require reusability that guarantees that we can generate a signature almost without affecting the quantum secret key.

\paragraph{Unpredictability style VRA security definition.}
It is obtained by setting $\chall$ as a random message $m^*$, $\aux$ as $m^*$, $\ans$ as a signature $\sig^\prime$, and $P$ as the predicate that takes as input $\aux=m^*$ and $\ans=\sig^\prime$ and outputs $1$ if and only if $\SigVrfy(\svk,m^*,\sig^\prime)=\top$, in the security experiment provided in \cref{sec-overview-def-skl}. 
\paragraph{Building blocks: TEPRF and coherently-signable constrained signature.}
We use TEPRF also in our construction of DS-SKL.
Moreover, we use constrained signature satisfying coherently-signability that we introduce in this work.

A constrained signature scheme consists of four algorithms $\KG$, $\Constrain$, $\Sign$, $\SigVrfy$.
By executing $\KG$, we obtain a signature verification key $\svk$ and a master secret key $\msk$.
By executing the constrain algorithm $\Constrain$ with $\msk$ and a function $f$, we can generate a constrained signing key $\sk_f$.
$\Sign$ and $\SigVrfy$ acts as those of a standard digital signature scheme, where $\Sign$ takes a constrained signing key $\sk_f$ as input together with a message $m$.
As the correctness, the constrained signing key $\sk_f$ satisfies that it can be used to generate a valid signature for a message $m$ such that $f(m)=1$.
As the security, it satisfies that any QPT adversary who is given $\sk_f$ cannot generate a valid signature for a message $m$ such that $f(m)=0$.
\takashi{I changed the roles of $f(m)=1$ and $f(m)=0$ (since this is more convenient for the description of the CS scheme.)}

Coherent-signability of constrained signature roughly guarantees that there exists a QPT algorithm $\QSign$ such that for any $(f_0,f_1,\ket{\psi_0},\ket{\psi_1},m,\alpha,\beta)$ satisfying $f_0(m)=f_1(m)=1$ and $|\alpha|^2+|\beta|^2=1$, $\QSign$ takes as input $\alpha\ket{0}\ket{\sk_{f_0}}\ket{\psi_0}+ \beta\ket{1}\ket{\sk_{f_1}}\ket{\psi_1}$\fuyuki{I assume $-$ can be absorbed by $\beta$, but does it feel weird?}\takashi{I see, you are right.} and $m$, and can output a valid classical signature $\sigma$ almost without affecting $\alpha\ket{0}\ket{\sk_{f_0}}\ket{\psi_0}+\beta\ket{1}\ket{\sk_{f_1}}\ket{\psi_1}$, where $\sk_{f_0}$ and $\sk_{f_1}$ are constrained keys for $f_0$ and $f_1$, respectively.
We show how to construct coherently-signable constrained signature based on the SIS assumption. 
We give an overview of the construction of our coherently-signable constrained signature in the final paragraph of this subsection.  

\paragraph{Construction.}
We construct $\DSSKL=\DSSKL.(\qKG,\qSign,\SigVrfy,\qDel,\DelVrfy)$ for the message space $\bit^{n\cdot\secp}$ using a TEPRF $\TEPRF=\TEPRF.(\KG,\Eval)$ for the domain $\bit^\secp$ and a coherently-signable constrained signature scheme $\CS=\CS.(\KG,\allowbreak\Constrain,\Sign,\SigVrfy,\QSign)$.
\begin{description}
\item[$\DSSKL.\KG$:] It first generates $\tlprfinp_i\gets\bit^\ell$, $(\TEPRF.\sk_{i,0},\TEPRF.\sk_{i,1})\gets\TEPRF.\KG(1^\secp,\tlprfinp_i)$, and $(\CS.\svk_i,\CS.\msk_i)\gets\CS.\KG(1^\secp)$ for every $i\in[n]$.
It then generates $\CS.\sk_{i,b}\gets \CS.\Constrain(\CS.\msk_{i},F[\TEPRF.\sk_{i.b}])$ for $i\in[n]$ and $b\in \bit$, where $F[\TEPRF.\sk_{i.b}]$ is a function that takes $(u,v)$ as input and outputs $1$ if $v=\TEPRF.\Eval(\TEPRF.\sk_{i,b},u)$ and $0$ otherwise.
It sets $\DS.\sk_{i,b}=(\TEPRF.\sk_{i,b},\CS.\sk_{i,b})$ for every $i\in[n]$ and $b\in\bit$.
It also generates $(x,\theta)\gets\bit^n\times\bit^n$.
It then generates $\qsk$ and $\dvk$ using $2n$ secret keys $(\DS.\sk_{i,b})_{i\in[n],b\in\bit}$ and $(x,\theta)$ as described in \cref{sec-overview-framework}.
It also sets $\svk=(\CS.\svk_i)_{i\in[n]}$.
\item[\rm $\DSSKL.\qSign$ and $\DSSKL.\SigVrfy$:] $\DSSKL.\qSign$ takes $\qsk=(\qsk_1,\ldots,\qsk_n)$ and $m=\prfinp_1\|\cdots\|\prfinp_n$ as inputs.
Let $E_i$ be a unitary that has $\prfinp_i$ hardwired, and maps $\ket{b}\ket{\TEPRF.\sk}\ket{\CS.\sk}\ket{c}$ to $\ket{b}\ket{\TEPRF.\sk}\ket{\CS.\sk}\ket{c\oplus\TEPRF.\Eval(\TEPRF.\sk,\prfinp_i)}$.
$\DSSKL.\qSign$ first applies $E_i$ to $\qsk_i\tensor\ket{0\cdots0}$ and obtains $\prfout_i$ by measuring the last register for every $i\in[n]$. Then, it generates $\CS.\sigma_i$ by executing $\QSign(\qsk_i,\prfinp_i\|\prfout_i)$.
The resulting signature is $\sigma=(\prfout_i,\CS.\sigma_i)_{i\in[n]}$.

$\DSSKL.\DelVrfy$ takes $\svk=(\CS.\svk_i)_{i\in[n]}$, $m=m_1\|\cdots\|m_n$, and $\sigma=(\prfout_i,\CS.\sigma_i)_{i\in[n]}$ as input, and checks if $\CS.\SigVrfy(\CS.\svk_i,m_i,\CS.\sigma_i)=\top$ or not for every $i\in[n]$.
$\DSSKL.\qSign$ outputs a valid signature almost without affecting $\qsk$, unless $\prfinp_i=\tlprfinp_i$ holds for some $i\in[n]$, from the correctness of $\TEPRF$ and the coherent-signability of $\CS$.
As a result, the construction satisfies signature verification correctness and reusability of the quantum secret key.
\ifnum\llncs=0 For the proofs for them, see \cref{sec:const-DS-SKL-classical-certificate}. \fi

\end{description}
$\DSSKL.\qDel$ and $\DSSKL.\DelVrfy$ are constructed as described in \cref{sec-overview-framework}.

\paragraph{Special challenge and auxiliary information.}
In the construction of $\DSSKL$, the special challenge $\chall_\theta=\hatm=\hatprfinp_1\|\cdots\|\hatprfinp_n$ is generated as 
\begin{align}
\hatprfinp=
\begin{cases}
\tlprfinp_i & (\textrm{if}~~\theta[i]=0)\\
\prfinp^* & (\textrm{if}~~\theta[i]=1),
\end{cases}
\end{align}
where $\tlprfinp_i$ is the input used when generating $(\TEPRF.\sk_{i,0},\TEPRF.\sk_{i,1})$ and $\prfinp^*\gets\bit^\ell$ for every $i\in[n]$.
The corresponding special auxiliary information $\aux_{x,\theta}$ is also set to $\hatm$.

We can prove indistinguishability from the normal challenge of this special challenge based on the security of $\TEPRF$, similarly to our UPF-SKL scheme.
Thus, we omit it and go to the extractability for the computational basis positions.
Suppose an QPT adversary outputs $\ans^\prime=\sigma^\prime=(\prfout^\prime_i,\CS.\sigma_i^\prime)_{i\in[n]}$ that satisfies $P(\aux_{x,\theta},\ans^\prime)=1$. Then, we have $\CS.\SigVrfy(\CS.\svk_i,\hatprfinp_i\|\prfout_i^\prime,\CS.\sigma_i)=\top$ holds for every $i\in[n]$ such that $\theta[i]=0$.
For $i\in[n]$ such that $\theta[i]=0$, only $\CS.\sk_{i,x[i]}$ is given to the adversary and $\CS.\sk_{i,1-x[i]}$ is completely hidden.
Then, from the unforgeability of $\CS$, it should hold that $\prfout^\prime_i=\TEPRF.\Eval(\TEPRF.\sk_{i,x[i]},\hatprfinp_i)$ for every $i\in[n]$ such that $\theta[i]=0$.
From the correctness of $\TEPRF$, we have $\TEPRF.\Eval(\TEPRF.\sk_{i,0},\hatprfinp_i)\ne\TEPRF.\Eval(\TEPRF.\sk_{i,1},\hatprfinp_i)$ for every $i\in[n]$ such that $\theta[i]=0$.
Then, we can extract $x[i]$ by checking which one of $\TEPRF.\Eval(\TEPRF.\sk_{i,0},\hatprfinp_i)$ and $\TEPRF.\Eval(\TEPRF.\sk_{i,1},\hatprfinp_i)$ matches with $\prfout^\prime_i$ for every $i\in[n]$ such that $\theta[i]=0$.

\paragraph{Construction of coherently-signable constrained signatures.}  \takashi{I added this paragraph.}
Tsabary \cite{TCC:Tsabary17} gave a simple generic conversion from homomorphic signatures to constrained signatures. 
We consider a constrained signature scheme obtained by applying the above conversion to the homomorphic signature scheme of \cite{STOC:GorVaiWic15} based on the SIS assumption, and show that the scheme satisfies coherent-signability.\footnote{Strictly speaking, our scheme is slightly different from the one obtained by applying the conversion of \cite{TCC:Tsabary17} to the scheme of \cite{STOC:GorVaiWic15}. \ifnum\llncs=0 See \Cref{rem:CS_difference}.\fi} Very roughly speaking, a signature of their scheme on a message $\msg$ is a solution of the SIS problem $\vx$ with respect to a matrix $\mA_\msg$ associated with $\msg$, i.e., $\vx$ is a short vector such that $\mA_\msg\vx=\vzero \mod q$. A constrained signing key for a function $f$ enables us to generate a trapdoor of $\mA_\msg$, which enables us to solve the SIS problem with respect to $\mA_\msg$, for any $\msg$ such that $f(\msg)=1$. To show that the scheme satisfies coherent-signability, we use the quantum Gaussian sampling technique developed in \cite{ITCS:Poremba23,TCC:AnaPorVai23}. Roughly, the quantum Gaussian sampling enables us to generate a Gaussian superposition over solutions to the SIS problem with respect to a matrix $\mA$ given a trapdoor of $\mA$. Notably, the Gaussian superposition is independent of a specific trapdoor used in the procedure, and thus we can coherently run it over a superposition of different trapdoors without collapsing the state.\footnote{Precisely speaking, we have to also make sure that the garbage (i.e., the state of ancilla qubits) is also independent of the trapdoor.} This property of the quantum Gaussian sampling is directly used to show the coherent-signability of the constrained signature scheme.


\section{Preliminaries}\label{sec:prelim}
\paragraph{Notations and conventions.}
In this paper, standard math or sans serif font stands for classical algorithms (e.g., $C$ or $\algo{Gen}$) and classical variables (e.g., $x$ or $\keys{pk}$).
Calligraphic font stands for quantum algorithms (e.g., $\qalgo{Gen}$) and calligraphic font and/or the bracket notation for (mixed) quantum states (e.g., $\qstate{q}$ or $\ket{\psi}$).

Let $[\ell]$ denote the set of integers $\{1, \cdots, \ell \}$, $\secp$ denote a security parameter, and $y \seteq z$ denote that $y$ is set, defined, or substituted by $z$.
For a finite set $X$ and a distribution $D$, $x \chosen X$ denotes selecting an element from $X$ uniformly at random, $x \chosen D$ denotes sampling an element $x$ according to $D$. Let $y \gets \algo{A}(x)$ and $y \gets \qalgo{A}(\qstate{x})$ denote assigning to $y$ the output of a probabilistic or deterministic algorithm $\algo{A}$ and a quantum algorithm $\qalgo{A}$ on an input $x$ and $\qstate{x}$, respectively. When we explicitly show that $\algo{A}$ uses randomness $r$, we write $y \gets \algo{A}(x;r)$.
For a classical algorithm $\algo{A}$, we write $y\in \algo{A}(x)$ to mean that $y$ is in the range of $\algo{A}(x)$, i.e., $\Pr[\algo{A}(x)=y]>0$. \takashi{I added the definition.}
PPT and QPT algorithms stand for probabilistic polynomial-time algorithms and polynomial-time quantum algorithms, respectively.
Let $\negl$ denote a negligible function and $\poly$ denote a positive polynomial. 
For a string $x\in\bit^n$, $x[i]$ is its $i$-th bit. 
For strings $x,y\in \bit^n$, $x\cdot y$ denotes $\bigoplus_{i\in[n]} x[i]\cdot y[i]$. For quantum states $\rho$ and $\sigma$, $\|\rho-\sigma\|_{tr}$ means their trace distance. 

\newcommand{\CDBB}{\mathsf{CDBB84}}

\paragraph{Certified deletion property of BB84 states.}
For strings $x\in\bit^n$ and $\theta\in\bit^n$, let $\ket{x^\theta} \seteq H^{\theta[1]}\ket{x[1]}\tensor\cdots \tensor H^{\theta[n]}\ket{x[n]}$ where  
$H$ is the Hadamard operator. A state of the form $\ket{x^\theta}$ is called a BB84 state~\cite{BB84}. We review the certified deletion property of BB84 states.

\begin{theorem}[Certified Deletion Property of BB84 States~{\cite[Proposition 4]{EPRINT:BehSatShi21}}]\label{thm:CD_BB84}
Consider the following game between an adversary $\qA=(\qA_0,\qA_1)$ the challenger.
\begin{enumerate}
\item The challenger picks a uniformly random strings $x\in\bit^n$ and $\theta\in\bit^n$ and sends $\ket{x^\theta}$ to $\qA_0$. 
\item $\qA_0$ outputs a classical string $y\in\bit^n$ and a state $\qstate{st}$.
\item $\theta$, $(x[i])_{i\in[n]:\theta[i]=1}$,  \takashi{I added it for consistency to the technical overview.}
and $\qstate{st}$ are sent to $\qA_1$.
$\qA_1$ outputs a classical string $z\in\bit^n$.
\end{enumerate}
Let $\CDBB(\qA,n)$ be a random variable which takes the value $1$ if $y[i]=x[i]$ for every $i\in[n]$ such that $\theta[i]=1$ and $z[i]=x[i]$ for every $i\in[n]$ such that $\theta[i]=0$, and takes the value $0$ otherwise. Then 
for any adversary $\qA=(\qA_0,\qA_1)$, it holds that
\begin{align}
\Pr[\CDBB(\qA,n)=1]\leq 2^{-\Omega(n)}.
\end{align}
\end{theorem}
\begin{remark}
In \cite[Proposition 4]{EPRINT:BehSatShi21}, they do not include $(x[i])_{i\in[n]:\theta[i]=1}$ as input of $\qA_1$.  However, we observe that giving $(x[i])_{i\in[n]:\theta[i]=1}$ to $\qA_1$ does not strengthen the statement  
 since we have $(x[i])_{i\in[n]:\theta[i]=1}=(y[i])_{i\in[n]:\theta[i]=1}$ whenever $\CDBB(\qA,n)=1$, in which case the adversary can recover it from $y$ and $\theta$ by itself. We include $(x[i])_{i\in[n]:\theta[i]=1}$ as input of $\qA_1$ since this is more convenient for security proofs of our schemes. \takashi{I added this remark.}
\end{remark}

\subsection{Classical Cryptographic Primitives}\label{sec:standard_crypto}

\paragraph{Public-key encryption.}
\begin{definition}[PKE]
A PKE scheme $\PKE$ is a tuple of three algorithms $(\KG, \Enc, \Dec)$. 
Below, let $\Ms$  be the message space of $\PKE$.\mor{The results of this paper hold even if some of $\KG,\Enc,\Dec$ are QPT?}\takashi{Yes (If we allow $\Enc$ of PKE-SKL to be quantum as well.)}
\begin{description}
\item[$\KG(1^\secp)\ra(\ek,\dk)$:] The key generation algorithm is a PPT algorithm that takes a security parameter $1^\lambda$, and outputs an encryption key $\ek$ and a decryption key $\dk$. 

\item[$\Enc(\ek,\msg)\ra\ct$:] The encryption algorithm is a PPT algorithm that takes an encryption key $\ek$ and a message $\msg \in \Ms$, and outputs a ciphertext $\ct$.

\item[$\Dec(\dk,\ct)\ra\tilde{\msg}$:] The decryption algorithm is a deterministic classical 
polynomial-time algorithm that takes a decryption key $\dk$ and a ciphertext $\ct$, and outputs a value $\tilde{\msg}$.

\item[Correctness:]For every $\msg \in \Ms$, we have
\begin{align}
\Pr\left[
\Dec(\dk, \ct) \allowbreak = \msg
\ \middle |
\begin{array}{ll}
(\ek,\dk)\gets\KG(1^\secp)\\
\ct\gets\Enc(\ek,\msg)
\end{array}
\right] 
=1-\negl(\secp).
\end{align}
\end{description}
\end{definition}

\begin{definition}[IND-CPA Security]\label{def:IND-CPA_PKE}
We say that a PKE scheme $\PKE$ with the message space $\Ms$ is IND-CPA secure if it satisfies the following requirement, formalized from the experiment $\expb{\PKE,\qA}{ind}{cpa}(1^\secp,\coin)$ between an adversary $\qA$ and the challenger:
        \begin{enumerate}
            \item  The challenger runs $(\ek,\dk)\gets\KG(1^\secp)$ and sends $\ek$ to $\qA$. 
            \item $\qA$ sends $(\msg_0^*,\msg_1^*)\in \Ms^2$ 
            to the challenger. 
            \item The challenger generates $\ct^*\la\Enc(\ek,\msg_\coin^*)$ and sends $\ct^*$ to $\qA$.
            \item $\qA$ outputs a guess $\coin^\prime$ for $\coin$. The challenger outputs $\coin'$ as the final output of the experiment.
        \end{enumerate}
        For any QPT $\qA$, it holds that
\begin{align}
\advb{\PKE,\qA}{ind}{cpa}(\secp) \seteq \abs{\Pr[\expb{\PKE,\qA}{ind}{cpa} (1^\secp,0) = 1] - \Pr[\expb{\PKE,\qA}{ind}{cpa} (1^\secp,1) = 1] }\leq \negl(\secp).
\end{align}
\end{definition}

\ifnum\llncs=0
\paragraph{Two-key equivocal PRF.}
We rely on a primitive called two-key equivocal PRF (TEPRF) introduced in  \cite{C:HJOSW16}.
\begin{definition}[Two-Key Equivocal PRF]\label{def:TEPRF} \takashi{I changed the definition according to the technical overview.}
A two-key equivocal PRF (TEPRF) with input length $\ell$ (and output length $1$)\footnote{Though we can also define it for larger output length, we assume the output length to be $1$ by default similarly to \cite{C:HJOSW16}.} is a tuple of two algorithms $(\KG,\Eval)$.
\begin{description}
\item[$\KG(1^\secp,s^*) \ra (\key_0,\key_1)$:] The key generation algorithm is a PPT algorithm that takes as input the security parameter $1^\secp$ and a string $s^*\in \bit^\ell$, and outputs two keys $\key_0$ and $\key_1$. 
\item[$\Eval(\key,s) \ra b$:] The evaluation algorithm is a deterministic classical polynomial-time algorithm that takes as input a key $\key$ and an input $s\in \bit^\ell$, and outputs a bit $b\in \bit$. 
\end{description}
\begin{description}
    \item[Equality:]
    For all $\secp\in \mathbb{N}$, $s^*\in \bit^\ell$, $(\key_0,\key_1)\gets \KG(1^\secpar,s^*)$, $s\in \bit^\ell\setminus \{s^*\}$, 
    \[
    \Eval(\key_0,s)=\Eval(\key_1,s).
    \]
    \item[Different values on target:] 
     For all $\secp\in \mathbb{N}$, $s^*\in \bit^n$, $(\key_0,\key_1)\gets \KG(1^\secpar,s^*)$, 
    \[
    \Eval(\key_0,s^*)\ne \Eval(\key_1,s^*).
    \]
  \item[Differing point hiding.] 
For any (stateful) QPT adversary $\A$, 
    \begin{align}
    \abs{
    \Pr\left[
    \A(\key_b)=1:
    \begin{array}{l}
    (s^*_0,s^*_1,b) \gets \A(1^\secpar)\\
    (\key_0,\key_1)\gets \KG(1^\secpar,s^*_0)\\
    \end{array}
    \right]
    -
   \Pr\left[
    \A(\key_b)=1:
    \begin{array}{l}
    (s^*_0,s^*_1,b) \gets \A(1^\secpar)\\
    (\key_0,\key_1)\gets \KG(1^\secpar,s^*_1)\\
    \end{array}
    \right]
    }\le \negl(\secp).
    \end{align}
\end{description}
\end{definition}
Our definition of TEPRFs is slightly different from the original one in \cite{C:HJOSW16}. In \Cref{sec:TEPRF}, we show that their definition implies ours. Since they show that OWFs imply TEPRFs under their definition, a similar implication holds for our definition of TEPRFs. 
\begin{theorem}\label{thm:OWF_to_TEPRF}
Assuming the existence of OWFs, there exist TEPRFs with input length $\ell$  for any polynomial $\ell=\ell(\secpar)$.
\end{theorem}
\fi

\ifnum\llncs=0
\section{Cryptographic Primitives with Secure Key Leasing}\label{sec:cryptoprimitive-skl-defs} 
\subsection{Public Key Encryption with Secure Key Leasing} \label{sec:pke-skl-defs}
In this subsection, we define public key encryption with secure key leasing (PKE-SKL) with classical revocation. 
\else  
\section{Definition of PKE-SKL} \label{sec:pke-skl-defs}
In this section, we define public key encryption with secure key leasing (PKE-SKL) with classical revocation.  
\fi
The syntax is similar to that in \cite{EC:AKNYY23} except that we introduce an algorithm that generates a classical certificate of deletion of a decryption key. A similar definition is introduced in \cite{TCC:AnaPorVai23} as key-revocable PKE with classical revocation.   
\begin{definition}[PKE-SKL with classical revocation]
A PKE-SKL scheme $\PKESKL$ with classical revocation is a tuple of five algorithms $(\qKG, \Enc, \qDec,\qDel,\DelVrfy)$. 
Below, let $\cM$  be the message space of $\PKESKL$. 
\begin{description}
\item[$\qKG(1^\secp)\ra(\ek,\qdk,\dvk)$:] The key generation algorithm is a QPT algorithm that takes a security parameter $1^\lambda$, and outputs an encryption key $\ek$, a decryption key $\qdk$, and a deletion verification key $\dvk$.

\item[$\Enc(\ek,\msg)\ra\ct$:] The encryption algorithm is a PPT algorithm that takes an encryption key $\ek$ and a message $\msg \in \cM$, and outputs a ciphertext $\ct$.

\item[$\qDec(\qdk,\ct)\ra\tilde{\msg}$:] The decryption algorithm is a QPT algorithm that takes a decryption key $\qdk$ and a ciphertext $\ct$, and outputs a value $\tilde{\msg}$.

\item[$\qDel(\qdk)\ra\cert$:] The deletion algorithm is a QPT algorithm that takes a decryption key $\qdk$, and outputs a classical string $\cert$.

\item[$\DelVrfy(\dvk,\cert)\ra\top/\bot$:] The deletion verification algorithm is a deterministic classical polynomial-time algorithm that takes a deletion verification key $\dvk$ and a deletion certificate $\cert$, and outputs $\top$ or $\bot$.

\item[Decryption correctness:]For every $\msg \in \cM$, we have
\begin{align}
\Pr\left[
\qDec(\qdk, \ct) \allowbreak = \msg
\ \middle |
\begin{array}{ll}
(\ek,\qdk,\dvk)\gets\qKG(1^\secp)\\
\ct\gets\Enc(\ek,\msg)
\end{array}
\right] 
=1-\negl(\secp).
\end{align}

\item[Deletion verification correctness:] We have 
\begin{align}
\Pr\left[
\DelVrfy(\dvk,\cert)=\top
\ \middle |
\begin{array}{ll}
(\ek,\qdk,\dvk)\gets\qKG(1^\secp)\\
\cert\gets\qDel(\qdk)
\end{array}
\right] 
=1-\negl(\secp).
\end{align}
\end{description}
\end{definition}
\begin{remark}[Reusability]\label{rem:reusability}
We can assume without loss of generality that a decryption key of a PKE-SKL scheme is reusable, i.e., it can be reused to decrypt (polynomially) many ciphertexts. In particular, we can assume that 
for honestly generated $\ct$ and $\qdk$, if we decrypt $\ct$ by using $\qdk$, the state of the decryption key after the decryption is negligibly close to that before the decryption in terms of trace distance. 
This is because the output of the decryption is almost deterministic by decryption correctness, and thus such an operation can be done without almost disturbing the input state by the gentle measurement lemma~\cite{Winter99}.    
\end{remark}

We next introduce the security notions for PKE-SKL with classical revocation.
Unlike existing security definitions of PKE-SKL~\cite{EC:AKNYY23} (or key-revocable PKE~\cite{TCC:AnaPorVai23}), we allow the adversary to obtain the verification key after submitting a valid certificate that passes the verification. We give more detailed comparisons after presenting our definitions.    

\begin{definition}[IND-VRA security]\label{def:IND-VRA_PKESKL} 
We say that a PKE-SKL scheme $\PKESKL$  with classical revocation for  the message space $\cM$ is IND-VRA secure,\footnote{"VRA" stands for "\textbf{V}erifiation key \textbf{R}evealing \textbf{A}ttack"} if it satisfies the following requirement, formalized by the experiment $\expb{\PKESKL,\qA}{ind}{vra}(1^\secp,\coin)$ between an adversary $\qA$ and the challenger:
        \begin{enumerate}
            \item  The challenger runs $(\ek,\qdk,\dvk)\gets\qKG(1^\secp)$ and sends $\ek$ and $\qdk$ to $\qA$. 
            \item $\qA$ sends $\cert$ and $(\msg_0^*,\msg_1^*)\in \cM^2$ to the challenger. If $\DelVrfy(\dvk,\cert)=\bot$, the challenger outputs $0$ as the final output of this experiment. Otherwise, the challenger generates $\ct^*\la\Enc(\ek,\msg_\coin^*)$, and sends $\dvk$ and $\ct^*$ to $\qA$.
            \item $\qA$ outputs a guess $\coin^\prime$ for $\coin$. The challenger outputs $\coin'$ as the final output of the experiment.
        \end{enumerate}
        For any QPT $\qA$, it holds that
\begin{align}
\advb{\PKESKL,\qA}{ind}{vra}(\secp) \seteq \abs{\Pr[\expb{\PKESKL,\qA}{ind}{vra} (1^\secp,0) = 1] - \Pr[\expb{\PKESKL,\qA}{ind}{vra} (1^\secp,1) = 1] }\leq \negl(\secp).
\end{align} 
\end{definition}

We also define the one-way variant of the above security. 
\begin{definition}[OW-VRA security]\label{def:OW-VRA_PKESKL}
We say that a PKE-SKL scheme  $\PKESKL$ with classical revocation  for the message space $\cM$ is OW-VRA secure, if it satisfies the following requirement, formalized by the experiment $\expb{\PKESKL,\qA}{ow}{vra}(1^\secp)$ between an adversary $\qA$ and the challenger:
        \begin{enumerate}
            \item  The challenger runs $(\ek,\qdk,\dvk)\gets\qKG(1^\secp)$ and sends $\ek$ and $\qdk$ to $\qA$.              
            \item $\qA$ sends $\cert$ to the challenger. If $\DelVrfy(\dvk,\cert)=\bot$, the challenger outputs $0$ as the final output of this experiment.
            Otherwise, the challenger chooses $\msg^*\gets \cM$, generates $\ct^*\la\Enc(\ek,\msg^*)$, and sends $\dvk$ and $\ct^*$ to $\qA$.
            \item $\qA$ outputs $\msg^\prime$. The challenger outputs $1$ if $\msg^\prime=\msg^*$ and otherwise outputs $0$ as the final output of the experiment.
        \end{enumerate}
        For any QPT $\qA$, it holds that
\begin{align}
\advb{\PKESKL,\qA}{ow}{vra}(\secp) \seteq \Pr[\expb{\PKESKL,\qA}{ow}{vra} (1^\secp) = 1]\leq \negl(\secp).
\end{align}
\end{definition}
\begin{remark}[Security as plain PKE]
It is straightforward to see that IND-VRA (resp. OW-VRA) security implies IND-CPA (resp. OW-CPA) security as a plain PKE scheme with quantum decryption keys.
\end{remark}

By the quantum Goldreich-Levin lemma~\cite{AC02,C:CLLZ21} we have the following lemma. 
\begin{lemma}\label{lem:ow-ind}
If there exists a OW-VRA secure PKE-SKL scheme with classical revocation, then there exists an IND-VRA secure PKE-SKL scheme with classical revocation.
\end{lemma}
The proof is almost identical to that of \cite[Lemma 3.12]{EC:AKNYY23}, and thus we omit it. \takashi{I removed it since that is too similar.}

\paragraph{\bf Comparisons with existing definitions.} 
Here, we give brief comparisons among our and existing security definitions. \ifnum\llncs=0 See \Cref{sec:comparison_def} for more details. \fi

Agrawal et al. \cite{EC:AKNYY23}   introduced security notions called IND-KLA and OW-KLA security where the adversary is allowed to make arbitrarily many queries to the deletion verification oracle.\footnote{Strictly speaking, these security notions  are defined for PKE-SKL with quantum revocation, but the definition can be extended to the classical revocation setting in a straightforward manner. \ifnum\llncs=0 See \Cref{def:IND-CPA_PKESKL,def:OW-CPA_PKESKL} for the definitions.\fi} 
Though we do not explicitly allow the adversary to access the deletion verification oracle, we can show that 
\begin{itemize}
\item OW-VRA security implies 
\ifnum\llncs=1
OW-KLA security. 
\else
OW-KLA security  
(\Cref{lem:IND-VRA_to_IND-KLA}). 
\fi
\item IND-VRA security implies IND-KLA security if the message space size is  
\ifnum\llncs=1
polynomial. 
\else
polynomial  (\Cref{lem:OW-VRA_to_OW-KLA}).
\fi
\end{itemize}
The above suggests that IND-VRA security and IND-KLA security may be incomparable when the message space size is super-polynomial. However, this is not a problem in terms of feasibility since we can always convert IND-VRA secure scheme into IND-KLA secure scheme by first restricting the message space to be single bits and then expanding the message space by parallel repetition which preserves IND-KLA security. In particular, combined with \Cref{lem:ow-ind}, if there exists a OW-VRA secure PKE-SKL scheme with classical revocation, then there exists an IND-SKL secure PKE-SKL scheme with classical revocation. 

\takashi{updated}
The security definition of Ananth et al. \cite{TCC:AnaPorVai23}, which we call APV-security, is very close to a weaker variant of IND-VRA security where the verification key is not given to the adversary. However, APV-security is not implied by IND-VRA security due to a subtle reason that it requires ciphertext pseudorandomness rather than a more standard indistinguishability-based definition. If we consider a slightly weakened variant of APV-security with an indistinguishability-based definition, then it would be implied by IND-VRA security. 
We remark that achieving APV-security based solely on the existence of IND-CPA PKE may be challenging, as the gap between indistinguishability-based security and ciphertext pseudorandomness, even for standard PKE, remains unresolved.  

\ifnum\llncs=0
\subsection{Pseudorandom and Unpredictable Functions with Secure Key Leasing}\label{def:PRF-SKL}
In this subsection, we define pseudorandom functions with secure key leasing (PRF-SKL) with classical revocation. 
The syntax is similar to that of key-revocable PRFs in \cite{TCC:AnaPorVai23}   
except that we introduce a deletion algorithm that generates a classical certificate of deletion of a key.\footnote{\cite{TCC:AnaPorVai23} also mentions a classical revocation variant, but they do not provide its formal definition.}

\begin{definition}[PRF-SKL with classical revocation]
A PRF-SKL scheme  $\PRFSKL$ with classical revocation for the domain $\Domprf$ and the range $\Ranprf$ is a tuple of five algorithms $(\qKG, \Eval, \qLEval,\qDel,\DelVrfy)$. 
\begin{description}
\item[$\qKG(1^\secp)\ra(\msk,\qsk,\dvk)$:] The key generation algorithm is a QPT algorithm that takes a security parameter $1^\lambda$, and outputs a master secret key $\msk$, a secret key $\qsk$, and a deletion verification key $\dvk$.

\item[$\Eval(\msk,\prfinp)\ra \prfout$:] The evaluation algorithm is a deterministic classical polynomial-time algorithm that takes a master secret key $\msk$ and an input $\prfinp \in \Domprf$, and outputs a value $\prfout$.

\item[$\qLEval(\qsk,\prfinp)\ra \prfout$:] The leased evaluation algorithm is a QPT algorithm that takes a secret key $\qsk$ and an input $\prfinp \in \Domprf$, and outputs a value $\prfout$.

\item[$\qDel(\qsk)\ra\cert$:] The deletion algorithm is a QPT algorithm that takes a secret key $\qsk$, and outputs a classical string $\cert$.

\item[$\DelVrfy(\dvk,\cert)\ra\top/\bot$:] The deletion verification algorithm is a deterministic classical polynomial-time algorithm that takes a deletion verification key $\dvk$ and a deletion certificate $\cert$, and outputs $\top$ or $\bot$.

\item[Evaluation correctness:]For every $\prfinp \in \Domprf$, we have
\begin{align}
\Pr\left[
\qLEval(\qsk,\prfinp) \allowbreak = \Eval(\msk,\prfinp)
\ \middle |
\begin{array}{ll}
(\msk,\qsk,\dvk)\gets\qKG(1^\secp)\\
\end{array}
\right] 
=1-\negl(\secp).
\end{align}

\item[Deletion verification correctness:] We have 
\begin{align}
\Pr\left[
\DelVrfy(\dvk,\cert)=\top
\ \middle |
\begin{array}{ll}
(\msk,\qsk,\dvk)\gets\qKG(1^\secp)\\
\cert\gets\qDel(\qsk)
\end{array}
\right] 
=1-\negl(\secp).
\end{align}
\end{description}
\end{definition}
\begin{remark}[Reusability]\label{rem:reusability＿PRF}
We can assume without loss of generality that a key of a PRF-SKL scheme is reusable, i.e., it can be reused to evaluate on  (polynomially) many inputs. This is because the output of the leased evaluation algorithm is almost unique by evaluation correctness, and thus such an operation can be done without almost disturbing the input state by the gentle measurement lemma~\cite{Winter99}.    
\end{remark}

Our security definition is similar to that in \cite{TCC:AnaPorVai23} except that we allow the adversary to receive the verification key after submitting a valid certificate.  
\begin{definition}[PR-VRA security]\label{def:PR-VRA}
We say that a PRF-SKL scheme  $\PRFSKL$ with classical revocation for the domain $\Domprf$ and the range $\Ranprf$ is PR-VRA secure, if it satisfies the following requirement, formalized by the experiment $\expb{\PRFSKL,\qA}{pr}{vra}(1^\secp,\coin)$ between an adversary $\qA$ and the challenger:
        \begin{enumerate}
            \item  The challenger runs $(\msk,\qsk,\dvk)\gets\qKG(1^\secp)$ and sends $\qsk$ to $\qA$. 
            \item $\qA$ sends $\cert$ to the challenger. If $\DelVrfy(\dvk,\cert)=\bot$, the challenger outputs $0$ as the final output of this experiment. Otherwise, the challenger generates $\prfinp^*\la\Domprf$, $\prfout^*_0\la\Eval(\msk,\prfinp^*)$, and $\prfout^*_1\la\Ranprf$, and sends $(\dvk,\prfinp^*,\prfout^*_\coin)$ to $\qA$.
            \item $\qA$ outputs a guess $\coin^\prime$ for $\coin$. The challenger outputs $\coin'$ as the final output of the experiment.
        \end{enumerate}
        For any QPT $\qA$, it holds that
\begin{align}
\advb{\PRFSKL,\qA}{pr}{vra}(\secp) \seteq \abs{\Pr[\expb{\PRFSKL,\qA}{pr}{vra} (1^\secp,0) = 1] - \Pr[\expb{\PRFSKL,\qA}{pr}{vra} (1^\secp,1) = 1] }\leq \negl(\secp).
\end{align} 
\end{definition}
\begin{remark}[Relationship with the definition of \cite{TCC:AnaPorVai23}]
It is easy to see that PR-VRA security implies the security of \cite{TCC:AnaPorVai23} since theirs is a weaker variant of PR-VRA security where the verification key is not given to the adversary.\footnote{Strictly speaking, \cite{TCC:AnaPorVai23} only gives a security definition in the quantum revocation setting, but it can be extended to the classical revocation setting in a straightforward manner.}
\end{remark}
\begin{remark}[Variant with the deletion verification oracle]
We can consider a seemingly stronger definition where we allow the adversary to make arbitrarily many queries to the deletion verification oracle in the first phase. However, we can show that the PR-VRA security implies security in such a setting as well by a similar argument to the proof of \Cref{lem:IND-VRA_to_IND-KLA}. Thus, we omit the deletion verification oracle from the definition for simplicity. 
\end{remark}
\begin{remark}[Security as a plain PRF]\label{rem:security_plain_PRF}
PR-VRA security does not imply security as a plain PRF though it \emph{does} imply security as a plain \emph{weak} PRF. However, there is a very simple transformation to add security as a plain PRF as observed in \cite[Claim 9.8]{TCC:AnaPorVai23}.  
We can simply take an XOR of PR-VRA secure PRF-SKL and a plain PRF to achieve both security simultaneously. 
Thus we focus on achieving PR-VRA security in this paper. 
\end{remark}

As an intermediate goal towards constructing PRF-SKL, we introduce a primitive which we call unpredictable functions with secure key leasing (UPF-SKL).\footnote{Though UPF-SKL and PRF-SKL are syntactically identical, we treat them as different primitives for clarity.}

\begin{definition}[UPF-SKL]\label{def:UPF-SKL}
A UPF-SKL scheme $\mathsf{UPFSKL}$ with classical revocation has the same syntax as PRF-SKL with classical revocation and satisfies the following security, which we call UP-VRA security, 
formalized by the experiment $\expb{\UPFSKL,\qA}{up}{vra}(1^\secp)$ between an adversary $\qA$ and the challenger:
        \begin{enumerate}
            \item  The challenger runs $(\msk,\qsk,\dvk)\gets\qKG(1^\secp)$ and sends $\qsk$ to $\qA$. 
            \item $\qA$ sends $\cert$ to the challenger. If $\DelVrfy(\dvk,\cert)=\bot$, the challenger outputs $0$ as the final output of this experiment. Otherwise, the challenger generates $\prfinp^*\la\Domprf$, and sends $(\dvk,\prfinp^*)$ to $\qA$.
            \item $\qA$ outputs a guess $t^\prime$ for the output on $s^*$. The challenger outputs $1$ if $t^\prime=\Eval(\msk,s^*)$ and otherwise outputs $0$. 
        \end{enumerate}
        For any QPT $\qA$, it holds that
\begin{align}
\advb{\UPFSKL,\qA}{up}{vra}(\secp) \seteq \Pr[\expb{\UPFSKL,\qA}{up}{vra} (1^\secp) = 1] \leq \negl(\secp).
\end{align} 
\end{definition}

By the quantum Goldreich-Levin lemma~\cite{AC02,C:CLLZ21} we have the following theorem. 
\begin{lemma}\label{lem:upf-prf}
If there exists a UP-VRA secure UPF-SKL scheme with classical revocation, then there exists a PR-VRA secure PRF-SKL scheme with classical revocation.
\end{lemma}
See \Cref{sec:UPF-to-PRF} for the proof. 

\subsection{Digital Signatures with Secure Key Leasing}
In this subsection, we define digital signatures with secure key leasing (DS-SKL) with classical revocation. 
A recent work by Morimae et al.~\cite{TQC:MorPorYam24} introduced 
a similar primitive called digital signatures with revocable signing keys. A crucial difference from their definition is that we require the quantum signing key to be static, i.e., the state of the signing key almost does not change before and after signing. 
\mor{static is the same as reusability defined below? In that case, it is better to use a single terminology.
reusability is formaly defined, but static is not. In remark 3.18, you use static.}
\takashi{
Yes, that is the same as reusability defined below. 
However, reusability is a confusing term when comparing with \cite{TQC:MorPorYam24} since theirs is also "reusable" in some sense. 
To clarify the difference, I renamed it to "reusability with static signing keys".
}

\begin{definition}[DS-SKL with classical revocation]
A DS-SKL scheme $\DSSKL$  with classical revocation  
is a tuple of five algorithms $(\qKG, \qSign, \SigVrfy,\qDel,\DelVrfy)$. 
Below, let $\cM$  be the message space of $\DSSKL$. 
\begin{description}
\item[$\qKG(1^\secp)\ra(\qsigk, \sigvk,\dvk)$:] The key generation algorithm is a QPT algorithm that takes a security parameter $1^\lambda$, and outputs a signing key $\qsigk$, a signature verification key $\sigvk$, and a deletion verification key $\dvk$.

\item[$\qSign(\qsigk,\msg)\ra(\qsigk',\sigma)$:] The signing algorithm is a QPT algorithm that takes a signing key $\qsigk$ and a message $\msg \in \cM$, and outputs a subsequent signing key $\qsigk'$ and a signature $\sigma$.  \takashi{I modified the syntax to output $\qsigk'$.}

\item[$\SigVrfy(\sigvk,\msg,\sigma)\ra \top/\bot$:] The signature verification algorithm is a deterministic classical polynomial-time algorithm that takes a signature verification key $\sigvk$, a message $\msg \in \cM$, and a signature $\sigma$, and outputs $\top$ or $\bot$. 

\item[$\qDel(\qsigk)\ra\cert$:] The deletion algorithm is a QPT algorithm that takes a signing key $\qsigk$, and outputs a deletion certificate $\cert$.

\item[$\DelVrfy(\dvk,\cert)\ra\top/\bot$:] The deletion verification algorithm is a deterministic classical polynomial-time algorithm that takes a deletion  verification key $\dvk$ and a deletion certificate $\cert$, and outputs $\top$ or $\bot$.

\item[Signature verification correctness:]For every $\msg \in \cM$, we have
\begin{align}
\Pr\left[
\SigVrfy(\sigvk, \msg,\sigma) \allowbreak = \top
\ \middle |
\begin{array}{ll}
(\qsigk,\sigvk,\dvk)\gets\qKG(1^\secp)\\
(\qsigk',\sigma)\gets\qSign(\qsigk,\msg)
\end{array}
\right] 
=1-\negl(\secp).
\end{align}

\item[Deletion verification correctness:] We have 
\begin{align}
\Pr\left[
\DelVrfy(\dvk,\cert)=\top
\ \middle |
\begin{array}{ll}
(\qsigk,\sigvk,\dvk)\gets\qKG(1^\secp)\\
\cert\gets\qDel(\qsigk)
\end{array}
\right] 
=1-\negl(\secp).
\end{align}
\item[Reusability with static signing keys:]
Let $\qsigk$ be an honestly generated signing key and $\msg\in \mathcal{M}$ be any message. 
Suppose that we run $(\qsigk',\sigma)\gets\qSign(\qsigk,\msg)$. 
Then we have 
\begin{align}
    \|\qsigk-\qsigk'\|_{tr}=\negl(\secp).
\end{align}
\end{description}
\end{definition}

\begin{remark}[Reusability]\label{rem:reusability_DS}
The previous work \cite{TQC:MorPorYam24} considered a weaker definition of reusability, where $\sigk'$ is not required to be close to $\sigk$ as long as it can be still used to generate signatures on further messages.
To emphasize the difference from their definition, we call the above property reusability with static signing keys. 
Unlike the cases of PKE-SKL and PRF-SKL, we cannot assume (even the weaker version of) reusability without loss of generality since signatures generated by the signing algorithm may not be unique. Thus, we explicitly state it as a requirement.
\end{remark}

We next introduce the security notions for DS-SKL with classical revocation.

\begin{definition}[RUF-VRA security]\label{def:RUF-VRA} 
We say that a DS-SKL scheme  $\DSSKL$ with classical revocation  for the message space $\cM$ is RUF-VRA secure,\footnote{"RUF" stands for "\textbf{R}andom message \textbf{U}n\textbf{F}orgeability".} if it satisfies the following requirement, formalized by the experiment $\expb{\DSSKL,\qA}{ruf}{vra}(1^\secp)$  between an adversary $\qA$ and the\mor{a dattari the dattari surunode touitsu shitahougayoi}\takashi{let's use "the".} challenger:
        \begin{enumerate}
            \item  The challenger runs $(\qsigk,\sigvk,\dvk)\gets\qKG(1^\secp)$ and sends $\qsigk$ and $\sigvk$ to $\qA$. 
            \item $\qA$ sends $\cert$  to the challenger. If $\DelVrfy(\dvk,\cert)=\bot$, the challenger outputs $0$ as the final output of this experiment. Otherwise, the challenger chooses $\msg^*\la\mathcal{M}$, and sends $\dvk$ and $\msg^*$ to $\qA$.
            \item $\qA$ outputs a signature $\sigma^\prime$. The challenger outputs $1$ if $\SigVrfy(\sigvk,\msg^*,\sigma^\prime)=\top$ and otherwise outputs $0$.  
        \end{enumerate}
        For any QPT $\qA$, it holds that
\begin{align}
\advb{\DSSKL,\qA}{ruf}{vra}(\secp) \seteq \Pr[\expb{\DSSKL,\qA}{ruf}{vra} (1^\secp) = 1] \leq \negl(\secp).
\end{align} 
\end{definition}

\begin{remark}[Comparison with \cite{TQC:MorPorYam24} ]\label{rem:comparison_DS}
We cannot directly compare our security definition with that in \cite{TQC:MorPorYam24} since the syntax is significantly different. In \cite{TQC:MorPorYam24}, the deletion verification algorithm additionally takes a set $S$ of messages as input. Intuitively, the deletion verification algorithm verifies that the signer generated signatures only on messages that belong to $S$. Thus, their security definition requires that the adversary cannot generate a valid signature on any message outside $S$ after submitting a certificate that passes verification w.r.t. $S$. While their definition looks conceptually stronger, such security is impossible to achieve when the signing key is static as in our definition since the adversary may generate signatures on any messages without being known by the verifier. Thus, we adopt the above random message style definition analogously to a security definition of copy-protection of signing keys~\cite{TCC:LLQZ22}. 
\end{remark}
\begin{remark}[Variant with the deletion  verification oracle]
We can consider a seemingly stronger definition where we allow the adversary to make arbitrarily many queries to the deletion verification oracle in the first phase similarly. However, we can show that the RUF-VRA security implies security in such a setting as well by a similar argument to the proof of \Cref{lem:IND-VRA_to_IND-KLA}.  Thus, we omit the deletion verification oracle from the definition for simplicity.
\end{remark}
\begin{remark}[Security as a plain DS]
 RUF-VRA security does not imply EUF-CMA security as a plain DS. However, similarly to the case of PRF, there is a very simple transformation to add EUF-CMA security as a plain DS. 
We can simply parallelly run an RUF-VRA secure DS-SKL scheme and EUF-CMA secure plain DS scheme where the signature verification algorithm accepts if the signature consists of valid signatures of both schemes. It is straightforward to see that the resulting scheme satisfies both RUF-VRA security as an DS-SKL scheme and EUF-CMA security as a plain DS scheme. Thus, we focus on constructing an RUF-VRA secure DS-SKL scheme.  
\end{remark}
\fi


\newcommand{\qstateA}{\qstate{st}_{\qA}}
\ifnum\llncs=0
\section{Public Key Encryption with Secure Key Leasing}\label{sec:const-PKE-SKL-classical-certificate} 
\else
\section{Construction of PKE-SKL}\label{sec:const-PKE-SKL-classical-certificate} 
\fi
In this section, we prove the following theorem. 
\begin{theorem}\label{thm:PKE-SKL}
If there exists an IND-CPA secure PKE scheme, then there exists an IND-VRA secure PKE-SKL scheme with classical revocation.  
\end{theorem}
By \Cref{lem:ow-ind}, it suffices to construct a OW-VRA secure PKE-SKL scheme with classical revocation, which we construct in the next subsection.
\subsection{Construction of OW-VRA Secure PKE-SKL}\label{sec:construction_PKE-SKL}
Let $\PKE=(\PKE.\KG,\PKE.\Enc,\PKE.\Dec)$ be an IND-CPA secure PKE scheme for single-bit messages. 
Let $\ell=\ell(\secpar)$ be the length of a decryption key generated by $\PKE.\KG(1^\secpar)$.\footnote{We can assume that all decryption keys are of the same length w.l.o.g.} 
Then, we construct a PKE-SKL scheme  $(\PKESKL.\qKG, \PKESKL.\Enc, \allowbreak\PKESKL.\qDec,\allowbreak\PKESKL.\qDel,\PKESKL.\DelVrfy)$ with classical revocation for $n$-bit messages where $n=\omega(\log \secpar)$ 
as follows.

\begin{description}
\item[$\PKESKL.\qKG(1^\secp)$:] $ $
\begin{itemize}
    \item Generate $x,\theta\gets\bit^n$.
    \item Generate $(\PKE.\ek_{i,b},\PKE.\dk_{i,b})\gets \PKE.\KG(1^\secp)$ for $i\in [n]$ and $b\in \bit$. 
    \item Generate
    \begin{align}
    \qdk_i\seteq
    \begin{cases}
    \ket{x[i]}\ket{\PKE.\dk_{i,x[i]}} & (\textrm{if~}\theta[i]=0)\\
    \frac{1}{\sqrt{2}}\left(\ket{0}\ket{\PKE.\dk_{i,0}}+(-1)^{x[i]}\ket{1}\ket{\PKE.\dk_{i,1}}\right) & (\textrm{if~}\theta[i]=1).
    \end{cases}
    \end{align}
    \item Output an encryption key 
    \[\ek\seteq (\PKE.\ek_{i,b})_{i\in[n],b\in\bit},\] 
    a decryption key 
    \[\qdk\seteq (\qdk_1,\ldots,\qdk_n)\]
    and a deletion verification key 
    \[\dvk\seteq ((x[i])_{i\in[n]:\theta[i]=1},\theta,(\PKE.\dk_{i,0},\PKE.\dk_{i,1})_{i\in[n]:\theta[i]=1}).\] 
\end{itemize}

\item[$\PKESKL.\Enc(\ek,\msg)$:] $ $
\begin{itemize}
    \item Parse $\ek=(\PKE.\ek_{i,b})_{i\in[n],b\in\bit}$. 
    \item Generate $\PKE.\ct_{i,b}\gets \PKE.\Enc(\PKE.\ek_{i,b},\msg[i])$ for $i\in [n]$ and $b\in \bit$.
\item Output $\ct\seteq (\PKE.\ct_{i,b})_{i\in [n],b\in \bit}$. 
\end{itemize}
 
\item[$\PKESKL.\qDec(\qdk,\ct)$:] $ $
\begin{itemize}
   \item Parse  $\qdk= (\qdk_1,\ldots,\qdk_n)$ and $\ct= (\PKE.\ct_{i,b})_{i\in [n],b\in \bit}$.
   \item Let $D_i$ be a unitary that works as follows:
   \[\ket{u}\ket{\PKE.\dk}\ket{v} \mapsto \ket{u}\ket{\PKE.\dk}\ket{v\oplus\PKE.\Dec(\PKE.\dk,\PKE.\ct_{i,u})}\] 
   for any $u,v\in \bit$ and $\PKE.\dk\in \bit^\ell$ 
   for every $i\in[n]$. 
   \mor{I first misunderstood that $\PKE.\dk=\{\PKE.\dk_{i,b}\}_{i,b}$, because $x$ is $\{x[i]\}_i$. Maybe just some other notation such as $a$ instead of $\PKE.\dk$?}\takashi{I see what you mean, but I still prefer the current notation since using another character may make it difficult to understand its role.} 
   \item Apply $D_i$ to $\qdk_i\tensor\ket{0}$ and measure the rightmost register to obtain an outcome $\msg_i^\prime$ for every $i\in[n]$.
\item Output $\msg^\prime\seteq \msg_i^\prime\concat\ldots\concat \msg_n^\prime$. 
\end{itemize}

\item[$\PKESKL.\qDel(\qdk)$:] $ $
\begin{itemize}
\item Parse $\qdk= (\qdk_1,\ldots,\qdk_n)$.
\item Measure $\qdk_i$ in the Hadamard basis to obtain an outcome $(e_i,d_i)\in \bit\times \bit^{\ell}$ for every $i\in[n]$.
\item Output $\cert\seteq (e_i,d_i)_{i\in[n]}$.
\end{itemize}

\item[$\PKESKL.\DelVrfy(\dvk,\cert)$:] $ $
\begin{itemize}
    \item Parse $\dvk=((x[i])_{i\in[n]:\theta[i]=1},\theta,(\PKE.\dk_{i,0},\PKE.\dk_{i,1})_{i\in[n]:\theta[i]=1})$ and  $\cert= (e_i,d_i)_{i\in [n]}$. 
    \item Output $\top$ if it holds that  
    \[e_i=x[i] \oplus d_i\cdot(\PKE.\dk_{i,0}\oplus\PKE.\dk_{i,1})\]
    for every $i\in[n]$ such that $\theta[i]=1$. Otherwise, output $\bot$.
\end{itemize}
\end{description}

\paragraph{\bf Decryption correctness.}
\mor{Some reviewers might think that decryption correctness and verification correnctness are trivial that do not need any proof.}\takashi{I think it's fine to keep it since this is not so long and one can just skip it if s/he thinks it's trivial.}
We show that $\msg_i^\prime=\msg[i]$ holds with overwhelming probability for $\msg_i^\prime$ computed by $\PKESKL.\qDec(\qdk,\ct)$ for every $i\in[n]$, where $(\ek,\qdk,\dvk)\gets\PKESKL.\qKG(1^\secp)$ and $\ct\gets\PKESKL.\Enc(\ek,\msg)$.

We first consider $i\in[n]$ such that $\theta[i]=0$. In this case, $\qdk_i=\ket{x[i]}\ket{\PKE.\dk_{i,x[i]}}$.
Thus, the state obtained by applying the map $D_i$ to $\qdk_i\tensor\ket{0}$ is $\ket{x[i]}\ket{\PKE.\dk_{i,x[i]}}\ket{m[i]}$ with overwhelming probability from the correctness of $\PKE$ and the fact that $\PKE.\ct_{i,x[i]}$ is a ciphertext of $m[i]$.
Therefore, we have $m_i^\prime=m[i]$ with overwhelming probability.

We next consider $i\in[n]$ such that $\theta[i]=1$. In this case, $\qdk_i=\frac{1}{\sqrt{2}}(\ket{0}\ket{\PKE.\dk_{i,0}}+(-1)^{x[i]}\ket{1}\ket{\PKE.\dk_{i,1}})$.
Thus, the state obtained by applying the map $D_i$ to $\qdk_i\tensor\ket{0}$ is $\frac{1}{\sqrt{2}}(\ket{0}\ket{\PKE.\dk_{i,0}}+(-1)^{x[i]}\ket{1}\ket{\PKE.\dk_{i,1}})\tensor\ket{m[i]}$ with overwhelming probability from the correctness of $\PKE$ and the fact that both $\PKE.\ct_{i,0}$ and $\PKE.\ct_{i,1}$ are ciphertexts of $m[i]$.
Therefore, we have $m_i^\prime=m[i]$ with overwhelming probability also in this case.

\paragraph{\bf Deletion verification correctness.}
We show that $e_i=x[i]\oplus d_i \cdot (\PKE.\dk_{i,0}\oplus\PKE.\dk_{i,1})$ holds for every $(e_i,d_i)$ for $i\in[n]$ such that $\theta[i]=1$ where $(\ek,\qdk,\dvk)\gets\PKESKL.\qKG(1^\secp)$ and 
$(e_i,d_i)_{i\in [n]}\gets \PKESKL.\qDel(\qdk)$. 
For every $i\in[n]$ such that $\theta[i]=1$, we have $\qdk_i=\frac{1}{\sqrt{2}}(\ket{0}\ket{\PKE.\dk_{i,0}}+(-1)^{x[i]}\ket{1}\ket{\PKE.\dk_{i,1}})$.
If we apply the Hadamard gate to every qubit of $\qdk_i$, we obtain
\begin{align}
&\frac{1}{\sqrt{2}}\left(
\frac{1}{\sqrt{2^{\ell+1}}}\sum_{e_i,d_i} (-1)^{d_i\cdot\PKE.\dk_{i,0}}\ket{e_i}\ket{d_i}
+
\frac{1}{\sqrt{2^{\ell+1}}}\sum_{e_i,d_i} (-1)^{x[i]\oplus e_i\oplus d_i\cdot\PKE.\dk_{i,1}}\ket{e_i}\ket{d_i}
\right)\\
=&
\frac{1}{\sqrt{2^{\ell}}}\sum_{
\substack{
(e_i,d_i):\\
d_i\cdot\PKE.\dk_{i,0}=x[i]\oplus e_i\oplus d_i\cdot\PKE.\dk_{i,1}
}}(-1)^{d_i\cdot\PKE.\dk_{i,0}}\ket{e_i}\ket{d_i}.
\end{align}
Thus, we have $e_i=x[i]\oplus d_i \cdot (\PKE.\dk_{i,0}\oplus\PKE.\dk_{i,1})$ for $(e_i,d_i)$ for every $i\in[n]$ such that $\theta[i]=1$.

\paragraph{\bf Security.} 

\begin{theorem}\label{thm:PKE_OW-VRA}
If $\PKE$ is $\INDCPA$ secure, then $\PKESKL$ is OW-VRA secure. 
\end{theorem}
\begin{proof}  
Let $\qA$ be a QPT adversary against OW-VRA security of $\PKESKL$. 
We consider the following sequence of hybrids. 
\begin{description}
\item[$\hybi{0}$:] This is the original security experiment $\expb{\PKESKL,\qA}{ow}{vra}(1^\secp)$. More specifically, it works as follows.
\begin{enumerate}
            \item  The challenger generates $x,\theta\gets\bit^n$ and $(\PKE.\ek_{i,b},\PKE.\dk_{i,b})\gets \PKE.\KG(1^\secp)$ for $i\in [n]$ and $b\in \bit$. The challenger also generates
            \begin{align}
    \qdk_i\seteq
    \begin{cases}
    \ket{x[i]}\ket{\PKE.\dk_{i,x[i]}} & (\textrm{if}~~\theta[i]=0)\\
    \frac{1}{\sqrt{2}}\left(\ket{0}\ket{\PKE.\dk_{i,0}}+(-1)^{x[i]}\ket{1}\ket{\PKE.\dk_{i,1}}\right) & (\textrm{if}~~\theta[i]=1).
    \end{cases}
    \end{align}
     The challenger sets
     \begin{align}
    &\ek\seteq (\PKE.\ek_{i,b})_{i\in[n],b\in\bit},~~~~~
    \qdk\seteq (\qdk_1,\ldots,\qdk_n),\\
    &\dvk\seteq ((x[i])_{i\in[n]:\theta[i]=1},\theta,(\PKE.\dk_{i,0},\PKE.\dk_{i,1})_{i\in[n]:\theta[i]=1}).
    \end{align}
     The challenger sends $\ek$ and $\qdk$ to $\qA$.  
     \item $\qA$ sends $\cert=(e_i,d_i)_{i\in[n]}$ to the challenger.
     The challenger outputs $0$ and the experiment ends if $e_i\ne x[i] \oplus d_i\cdot(\PKE.\dk_{i,0}\oplus\PKE.\dk_{i,1})$  for some $i\in[n]$ such that $\theta[i]=1$. Otherwise, the challenger generates $\ct^*$ as follows. The challenger chooses $\msg^*\gets \bit^n$, generates $\PKE.\ct^*_{i,b}\gets \PKE.\Enc(\PKE.\ek_{i,b},\msg^*[i])$ for every $i\in [n]$ and $b\in \bit$, and sets $\ct^*\seteq(\PKE.\ct^*_{i,b})_{i\in[n],b\in\bit}$. The challenger sends $\dvk$ and $\ct^*$ to $\qA$.
         \item $\qA$ outputs $\msg'$. The challenger outputs $1$ if $\msg'=\msg^*$ and otherwise $0$ as the final outcome of this experiment.  
        \end{enumerate}
By the definition, we clearly have 
\[
\Pr[\hybi{0}=1]=\advb{\PKESKL,\qA}{ow}{vra}(1^\secp).
\]



\takashi{I removed one hybrid here. As a result, the numbering of the remaining hybrids should be decremented by one. I did so carefully, but if you find this is not done somewhere, please just correct it.}

\item[$\hybi{1}$:] This is identical to $\hybi{0}$ except that we replace $\msg^*[i]$ with $\msg^*[i]\oplus x[i]$ for every $i\in[n]$ such that $\theta[i]=0$ throughout the experiment. In particular, this causes the following changes.
\begin{itemize}
\item For every $i\in[n]$ such that $\theta[i]=0$, we generate $\PKE.\ct^*_{i,b}$ as $\PKE.\ct^*_{i,b}\gets\PKE.\Enc(\PKE.\ek_{i,b},\msg^*[i]\oplus x[i])$ for every $b\in\bit$. 
\item $\qA$ wins if $m^\prime[i]=\msg^*[i]\oplus x[i]$ holds for every $i\in[n]$ such that $\theta[i]=0$.
\end{itemize}

Since $\msg^*[i]$ is chosen uniformly at random independently of all the other random variables that are generated by that point,
$\msg^*[i]\oplus x[i]$ is also uniformly random and independent  
for every $i\in[n]$. Thus, this change does not affect the view of $\qA$  
and
thus we have 
\[
\Pr[\hybi{0}=1]=\Pr[\hybi{1}=1].
\]

\item[$\hybi{2}$:] This is identical to $\hybi{1}$ except that the challenger generates $\PKE.\ct^*_{i,1-x[i]}$ as
\begin{align}
\PKE.\ct^*_{i,1-x[i]}\gets\PKE.\Enc(\PKE.\ek_{i,1-x[i]},\msg^*[i]\oplus (1-x[i]))
\end{align}
for every $i\in[n]$ such that $\theta[i]=0$.

For $i\in[n]$ such that $\theta[i]=0$, we have $\qdk_i=\ket{x[i]}\ket{\PKE.\dk_{i,x[i]}}$ and $\PKE.\dk_{i,1-x[i]}$ is not given to $\qA$.
Also, all the decryption keys of $\PKE$ contained in 
the deletion verification key $\dvk$ are $(\PKE.\dk_{i,0},\PKE.\dk_{i,1})_{i\in[n]:\theta[i]=1}$, and in particular the verification key contains no information about $\PKE.\dk_{i,1-x[i]}$ for $i\in[n]$ such that $\theta[i]=0$. 
Then, by a straightforward hybrid argument based on the $\INDCPA$ security of $\PKE$, we have 
\[
\abs{\Pr[\hybi{1}=1]-\Pr[\hybi{2}=1]}\le\negl(\secp).
\]
\end{description}

Below, by a reduction to the certified deletion property of BB84 states (\Cref{thm:CD_BB84}), 
we show 
\[
\Pr[\hybi{2}=1]\le\negl(\secp). 
\]

In $\hybi{2}$, $\PKE.\ct^*_{i,b}$ is generated as 
\begin{align}
\PKE.\ct^*_{i,b}\gets
\begin{cases}
\PKE.\Enc(\PKE.\ek_{i,b}, \msg^*[i]\oplus b) & (\textrm{if}~~\theta[i]=0)\\
\PKE.\Enc(\PKE.\ek_{i,b}, \msg^*[i]) & (\textrm{if}~~\theta[i]=1).
\end{cases}
\end{align}

We construct an adversary $\qB=(\qB_0,\qB_1)$ against the certified deletion property of BB84 states (\Cref{thm:CD_BB84}) as follows. 
\begin{description}
\item[$\qB_0$:] Given an $n$-qubit BB84 state
$\ket{\psi}=\ket{x}_{\theta}$, 
$\qB_0$ parses 
$\ket{\psi}=\ket{\psi_1}\tensor\cdots\tensor\ket{\psi_n}$ and generates $(\PKE.\ek_{i,b},\PKE.\dk_{i,b})\gets\PKE.\KG(1^\secp)$ for every $i\in[n]$ and $b\in\bit$.
Let $U_i$ be a unitary that maps $\ket{u}\ket{v}$ to $\ket{u}\ket{v\oplus\PKE.\dk_{i,u}}$ where $u\in \bit$ and $v\in \bit^\ell$. 
$\qB_0$ applies $U_i$ to $\ket{\psi_i}\ket{0^{\ell}}$ to obtain $\qdk_i$ for every $i\in[n]$.
$\qB_0$ then sets
     \begin{align}
    \ek\seteq (\PKE.\ek_{i,b})_{i\in[n],b\in\bit}~~~~\textrm{and}~~~~
    \qdk\seteq (\qdk_1,\ldots,\qdk_n)
    \end{align}
     and sends $\ek$ and $\qdk$ to $\qA$.  When $\qA$ outputs $\cert=(e_i,d_i)_{i\in[n]}$, let $\qstateA$ be the state of $\qA$ at this point. 
     $\qB_0$ computes $y_i=e_i\oplus d_i\cdot(\PKE.\dk_{i,0}\oplus\PKE.\dk_{i,1})$ for every $i\in[n]$. 
     $\qB_0$ outputs $y\seteq y_1\concat\ldots\concat y_n$ and  
    $\qstate{st}_{\qB}\seteq (\ek,(\PKE.\dk_{i,b})_{i\in[n],b\in\bit},\qstateA)$.
\item[$\qB_1$:] Given $\theta$, $(x[i])_{i\in[n]:\theta[i]=1}$, and $\qstate{st}_{\qB}=(\ek,(\PKE.\dk_{i,b})_{i\in[n],b\in\bit},\qstateA)$, $\qB_1$ sets $$\dvk^\prime\seteq((x[i])_{i\in[n]:\theta[i]=1},\theta,(\PKE.\dk_{i,0},\PKE.\dk_{i,1})_{i\in[n]:\theta[i]=1}).$$
$\qB_1$ then generates
\begin{align}
\PKE.\ct^*_{i,b}\gets
\begin{cases}
\PKE.\Enc(\PKE.\ek_{i,b}, \msg^*[i]\oplus b) & (\textrm{if}~~\theta[i]=0)\\
\PKE.\Enc(\PKE.\ek_{i,b}, \msg^*[i]) & (\textrm{if}~~\theta[i]=1),
\end{cases}
\end{align}
where $\msg^*\gets\bit^n$.
$\qB_1$ executes $\qA$ with the input $\dvk^\prime$ and $\ct^*\seteq(\PKE.\ct^*_{i,b})_{i\in[n],b\in\bit}$ and the state $\qstateA$, and obtains an output $\msg^\prime$.
$\qB_1$ sets $z_i\seteq \msg^\prime[i]\oplus \msg^*[i]$ for every $i\in[n]$ and outputs $z=z_1\concat\ldots\concat z_n$.
\end{description}

$\qB$ perfectly simulates $\hybi{2}$ for $\qA$. 
If $\qA$ wins the simulated experiment (i.e., $\hybi{2}$ outputs $1$), we have
\begin{align}
e_i=x[i]\oplus d_i \cdot (\PKE.\dk_{i,0}\oplus\PKE.\dk_{i,1})
\end{align}
 for every $i\in[n]$ such that $\theta[i]=1$
and 
\begin{align}
\msg^\prime[i]=\msg^*[i]\oplus x[i]
\end{align}
for every $i\in[n]$ such that $\theta[i]=0$. 
In this case, $y_i=x[i]$ for every $i\in [n]$ such that $\theta[i]=1$ and $z_i=x[i]$ for every $i\in [n]$ such that $\theta[i]=0$. 
Thus, we have 
\[
\Pr[\hybi{2}=1]\le\Pr[\CDBB(\qB,n)=1]\le \negl(\secpar)
\] 
where the latter inequality follows from the certified deletion property of BB84 states (\Cref{thm:CD_BB84}) 
and $n=\omega(\log\secpar)$.

Combining the above, we have 
\[
\advb{\PKESKL,\qA}{ow}{vra}(1^\secp)\le \negl(\secpar).
\] 
This completes the proof of \cref{thm:PKE_OW-VRA}.
\end{proof}
By \Cref{lem:ow-ind,thm:PKE_OW-VRA}, we obtain \Cref{thm:PKE-SKL}. 

\subsection{Variants}
\paragraph{\bf PKE variants.} 
Thanks to the simplicity of our construction and its security proof, we can easily extend it to other forms of PKE. For example, it is immediate to see that our construction works even if the base PKE scheme has quantum ciphertexts and/or quantum public keys as in \cite{EPRINT:MorYam22c,cryptoeprint:2023/282,TCC:BGHMSVW23} 
if we allow the resulting PKE-SKL scheme to inherit the similar property. 
Moreover, if the base PKE scheme has tamper-resilient quantum public keys \cite{C:KMNY24}, i.e., the security holds against adversaries that tamper with quantum public keys, then our construction satisfies a similar level of security. 
Contrary, it is unclear if the construction of Agrawal et al. \cite{EC:AKNYY23} extends to those settings.  

\paragraph{\bf IBE and ABE.} 
Agrawal et al. \cite{EC:AKNYY23} defined IBE-SKL and ABE-SKL and showed that PKE-SKL can be generically upgraded into IBE-SKL (resp. ABE-SKL for a function class $\Fs$) assuming the existence of plain IBE (resp. ABE for $\Fs$). Though their conversion is given in the quantum revocation setting, it is easy to see that the conversion works in the classical revocation setting and preserves IND-VRA security.\footnote{
We omit a formal definition of IND-VRA security for IBE-SKL and ABE-SKL with classical revocation, but it can be defined similarly to \cite{EC:AKNYY23}. 
Similarly to \cite{EC:AKNYY23}, the security has bounded decrypting key restriction which means that the adversary can obtain only bounded number of copies of a decryption key that can decrypt the challenge ciphertext (but there is no limitation about the number of decryption keys that cannot decrypt the challenge ciphertext).}  
Thus, we can show that plain IBE and ABE can be generically upgraded into IND-VRA secure IBE-SKL and ABE-SKL schemes. A caveat of their conversion is that it relies on garbled circuits where we garble encryption algorithm of the underlying IBE or ABE, which makes it very inefficient. \fuyuki{??}\takashi{ok now?}
We observe that our construction of OW-VRA secure PKE-SKL can be used to  \emph{directly} construct IBE-SKL and ABE-SKL from plain IBE and ABE. That is, if we instantiate the scheme based on plain IBE (resp. ABE) instead of plain PKE, then we obtain OW-VRA secure IBE-SKL (resp. ABE-SKL). This construction only makes black-box use of the underlying scheme and thus the actual efficiency would be much better. 
\ifnum\llncs=0 However there is a subtlety that we do not know how to upgrade OW-VRA security into IND-VRA security for IBE-SKL or ABE-SKL.  A natural idea to do so is to use quantum Goldreich-Levin lemma similarly to the case of PKE-SKL, but the problem is that the reduction does not seem to work due to the decryption key queries in the second stage.
\fi

\ifnum\llncs=0

\section{Pseudorandom Functions with Secure Key Leasing}\label{sec:const-PRF-SKL-classical-certificate}
In this section, we prove the following theorem. 
\begin{theorem}\label{thm:PRF-SKL}
If there exist OWFs, then there exist PR-VRA secure PRF-SKL with classical revocation.  
\end{theorem}
By \Cref{lem:upf-prf}, it suffices to construct a UP-VRA secure UPF-SKL scheme with classical revocation, which we construct in the next subsection. 

\subsection{Construction of UP-VRA Secure UPF-SKL} 
Let $\TEPRF=(\TEPRF.\ObvKG,\TEPRF.\KG,\TEPRF.\Eval)$ be a TEPRF with input length $\ell=\omega(\log \secpar)$. 
Let $\ell_\key$ be the length of a key generated by $\TEPRF.\KG(1^\secpar,s)$ for $s\in \bit^\ell$.\footnote{We can assume that all keys are of the same length w.l.o.g.} 
Then, we construct a UPF-SKL scheme   $\UPFSKL=(\UPFSKL.\qKG,\allowbreak \UPFSKL.\Eval, \UPFSKL.\qLEval,\allowbreak\UPFSKL.\qDel,\UPFSKL.\DelVrfy)$ with classical revocation with $n\ell$-bit inputs and $n$-bit outputs where $n=\omega(\log \secpar)$ as follows.

\begin{description}
\item[$\UPFSKL.\qKG(1^\secp)$:] $ $
\begin{itemize}
    \item Generate $x,\theta\gets\bit^n$.
    \item $\tilde{s}_i\gets \bit^\ell$ for $i\in [n]$.
    \item Generate $(\TEPRF.\key_{i,0},\TEPRF.\key_{i,1})\gets\TEPRF.\KG(1^\secp,\tilde{s}_i)$ for $i\in [n]$.
    \item Generate
    \begin{align}
    \qsk_i\seteq
    \begin{cases}
    \ket{x[i]}\ket{\TEPRF.\key_{i,x[i]}} & (\textrm{if~}\theta[i]=0)\\
    \frac{1}{\sqrt{2}}\left(\ket{0}\ket{\TEPRF.\key_{i,0}}+(-1)^{x[i]}\ket{1}\ket{\TEPRF.\key_{i,1}}\right) & (\textrm{if~}\theta[i]=1).
    \end{cases}
    \end{align}
    \item Output a master secret key 
    \[
    \msk\seteq(\TEPRF.\key_{i,0})_{i\in [n]},
    \] 
    a quantum secret key 
    \[
    \qsk\seteq(\qsk_1,\ldots, \qsk_n),
    \]
    and a verification key
    \[\dvk\seteq((x[i])_{i\in[n]:\theta[i]=1},\theta,(\TEPRF.\key_{i,0},\TEPRF.\key_{i,1})_{i\in[n]:\theta[i]=1}).
    \]
\end{itemize}

\item[$\UPFSKL.\Eval(\msk,\prfinp)$:] $ $
\begin{itemize}
    \item Parse $\msk=(\TEPRF.\key_{i,0})_{i\in [n]}$
    and $s=s_1\concat\ldots \concat s_n$ 
    where $s_i\in \bit^\ell$ for each $i\in[n]$.
    \item Compute 
    $t_i\gets \TEPRF.\Eval(\TEPRF.\key_{i,0},s_i)$ for $i\in [n]$. 
    \item Output $t\seteq t_1\concat\ldots \concat t_n$. 
\end{itemize}
 
\item[$\UPFSKL.\qLEval(\qsk,\prfinp)$:] $ $
\begin{itemize}
   \item Parse  $\qsk= (\qsk_1,\ldots,\qsk_n)$ and $s=s_1\concat\ldots \concat s_n$ where $s_i\in \bit^\ell$ for each $i\in[n]$.
   \item Let $E_i$ be a unitary that works as follows: 
   \mor{$\TEPRF.\key$ could be misunderstood as $\{\TEPRF.\key_{i,b}\}_{i,b}$.}\takashi{I see what you mean, but I still prefer the current notation since using another character may make it difficult to understand its role.} 
   \[\ket{u}\ket{\TEPRF.\key}\ket{v} \mapsto \ket{u}\ket{\TEPRF.\key}\ket{v\oplus\TEPRF.\Eval(\TEPRF.\key,s_i)}\]
   where $u,v\in \bit$ and $\TEPRF.\key\in \bit^{\ell_\key}$.  
   \item Apply $E_i$ to $\qsk_i\tensor\ket{0}$ and measure the rightmost register to obtain $t_i^\prime$ for every $i\in[n]$.
  \item Output $t^\prime\seteq t^\prime_1\concat\ldots \concat t^\prime_n$. 
\end{itemize}

\item[$\UPFSKL.\qDel(\qsk)$:] $ $
\begin{itemize}
\item Parse $\qsk= (\qsk_1,...,\qsk_n)$.
\item Measure $\qsk_i$ in the Hadamard basis and obtain $(e_i,d_i)$ for every $i\in[n]$.
\item Output $\cert\seteq (e_i,d_i)_{i\in [n]}$. 
\end{itemize}

\item[$\UPFSKL.\DelVrfy(\dvk,\cert)$:] $ $
\begin{itemize}
    \item Parse $\dvk=((x[i])_{i\in[n]:\theta[i]=1},\theta,(\TEPRF.\key_{i,0},\TEPRF.\key_{i,1})_{i\in[n]:\theta[i]=1})$ and $\cert= (e_i,d_i)_{i\in [n]}$. 
    \item Output $\top$ if it holds that  
\[e_i=x[i]\oplus d_i\cdot (\TEPRF.\key_{i,0}\oplus \TEPRF.\key_{i,1})\]
for every $i\in [n]$ such that $\theta[i]=1$. Otherwise, output $\bot$. 
\end{itemize}
\end{description}

\paragraph{\bf Evaluation correctness.}
We show that 
for any $s=s_1\concat\ldots\concat s_n\in \bit^{n\ell}$,  
$t_i^\prime=t_i=\TEPRF.\Eval(\TEPRF.\sk_{i,0},s_i)$ holds for every $i\in[n]$ with overwhelming probability 
where 
$(\msk,\qsk,\dvk)\gets\UPFSKL.\qKG(1^\secp)$,  $t_i^\prime$ is computed as in $\UPFSKL.\qLEval(\qsk,s)$
and $t_i$ computed as in $\UPFSKL.\Eval(\msk,s)$. 

We first consider $i\in[n]$ such that $\theta[i]=0$. In this case, $\qsk_i=\ket{x[i]}\ket{\TEPRF.\sk_{i,x[i]}}$.
Note that 
by the equality property of $\TEPRF$, we have $\Eval(\TEPRF.\sk_{i,0},s_i)=\Eval(\TEPRF.\sk_{i,1},s_i)$ unless $s_i=\tilde{s}_i$ where $\tilde{s}_i$ is the value chosen in $\UPFSKL.\qKG$. 
Thus, 
the state obtained by applying the map $E_i$ to $\qsk_i\tensor\ket{0}$ is \ifnum\llncs=0  $\ket{x[i]}\ket{\TEPRF.\sk_{i,x[i]}}\ket{\TEPRF.\Eval(\TEPRF.\sk_{i,0},s_i)}$ 
\else
\[\ket{x[i]}\ket{\TEPRF.\sk_{i,x[i]}}\ket{\TEPRF.\Eval(\TEPRF.\sk_{i,0},s_i)}\] \fi
unless $s_i=\tilde{s}_i$. 
Since the probability that $s_i=\tilde{s}_i$  is $2^{-\ell}\le \negl(\secpar)$, 
we have $t_i^\prime=t_i$ with overwhelming probability.

We next consider $i\in[n]$ such that $\theta[i]=1$. In this case, $\qsk_i=\frac{1}{\sqrt{2}}(\ket{0}\ket{\TEPRF.\sk_{i,0}}+(-1)^{x[i]}\ket{1}\ket{\TEPRF.\sk_{i,1}})$.
Thus, the state obtained by applying the map $E_i$ to $\qsk_i\tensor\ket{0}$ is $\frac{1}{\sqrt{2}}(\ket{0}\ket{\TEPRF.\sk_{i,0}}+(-1)^{x[i]}\ket{1}\ket{\TEPRF.\sk_{i,1}})\tensor\ket{\TEPRF.\Eval(\TEPRF.\sk_{i,0},s_i)}$ unless $s_i=\tilde{s}_i$ where $\tilde{s}_i$ is the value chosen in $\UPFSKL.\qKG$. 
Since the probability that $s_i=\tilde{s}_i$  is $2^{-\ell}\le \negl(\secpar)$, 
we have $t_i^\prime=t_i$ with overwhelming probability also in this case.

\paragraph{\bf Deletion verification correctness.}
We show that $e_i=x[i]\oplus d_i \cdot (\TEPRF.\sk_{i,0}\oplus\TEPRF.\sk_{i,1})$ holds for every $(e_i,d_i)$ for $i\in[n]$ such that $\theta[i]=1$ where $(\msk,\qsk,\dvk)\gets\UPFSKL.\qKG(1^\secp)$ and 
$(e_i,d_i)_{i\in [n]}\gets \UPFSKL.\qDel(\qsk)$. 
For every $i\in[n]$ such that $\theta[i]=1$, we have $\qsk_i=\frac{1}{\sqrt{2}}(\ket{0}\ket{\TEPRF.\sk_{i,0}}+(-1)^{x[i]}\ket{1}\ket{\TEPRF.\sk_{i,1}})$.
If we apply the Hadamard gate to every qubit of $\qsk_i$, we obtain
\begin{align}
&\frac{1}{\sqrt{2}}\left(
\frac{1}{\sqrt{2^{\ell_\key+1}}}\sum_{e_i,d_i} (-1)^{d_i\cdot\TEPRF.\sk_{i,0}}\ket{e_i}\ket{d_i}
+
\frac{1}{\sqrt{2^{\ell_\key+1}}}\sum_{e_i,d_i} (-1)^{x[i]\oplus e_i\oplus d_i\cdot\TEPRF.\sk_{i,1}}\ket{e_i}\ket{d_i}
\right)\\
=&
\frac{1}{\sqrt{2^{\ell_\key}}}\sum_{
\substack{
(e_i,d_i):\\
d_i\cdot\TEPRF.\sk_{i,0}=x[i]\oplus e_i\oplus d_i\cdot\TEPRF.\sk_{i,1}
}}(-1)^{d_i\cdot\TEPRF.\sk_{i,0}}\ket{e_i}\ket{d_i}.
\end{align}
Thus, we have $e_i=x[i]\oplus d_i \cdot (\TEPRF.\sk_{i,0}\oplus\TEPRF.\sk_{i,1})$ for every $i\in[n]$ such that $\theta[i]=1$.

\paragraph{\bf Security.} 

\begin{theorem}\label{thm:UPF_UP-VRA}
If $\TEPRF$ satisfies the  properties of \Cref{def:TEPRF}, then $\UPFSKL$ is UP-VRA secure. 
\end{theorem}
\begin{proof}  
Let $\qA$ be a QPT adversary against UP-VRA security of $\UPFSKL$. 
We consider the following sequence of hybrids. 
\begin{description}
\item[$\hybi{0}$:] This is the original security experiment $\expb{\UPFSKL,\qA}{up}{vra}(1^\secp)$. More specifically, it works as follows.
\begin{enumerate}
            \item  The challenger generates $x,\theta\gets\bit^n$ and 
     $\tilde{s}_i\gets \bit^\ell$  and $(\TEPRF.\key_{i,0},\TEPRF.\key_{i,1})\gets\TEPRF.\KG(1^\secp,\tilde{s}_i)$ for $i\in [n]$.
     The challenger also generates
            \begin{align}
    \qsk_i\seteq
    \begin{cases}
    \ket{x[i]}\ket{\TEPRF.\sk_{i,x[i]}} & (\textrm{if}~~\theta[i]=0)\\
    \frac{1}{\sqrt{2}}\left(\ket{0}\ket{\TEPRF.\sk_{i,0}}+(-1)^{x[i]}\ket{1}\ket{\TEPRF.\sk_{i,1}}\right) & (\textrm{if}~~\theta[i]=1).
    \end{cases}
    \end{align}
     The challenger sets
     \begin{align}
    &\qsk\seteq (\qsk_1,\ldots,\qsk_n),\\
    &\dvk\seteq ((x[i])_{i\in[n]:\theta[i]=1},\theta,(\TEPRF.\sk_{i,0},\TEPRF.\sk_{i,1})_{i\in[n]:\theta[i]=1}).
    \end{align}
     The challenger sends $\qsk$ to $\qA$.  
     \item $\qA$ sends $\cert=(e_i,d_i)_{i\in[n]}$ to the challenger.
     The challenger outputs $0$ and the experiment ends if $e_i\ne x[i] \oplus d_i\cdot(\TEPRF.\sk_{i,0}\oplus\TEPRF.\sk_{i,1})$  for some $i\in[n]$ such that $\theta[i]=1$. Otherwise, the challenger chooses $s^*=s^*_1\concat\ldots \concat s^*_n\gets \bit^{n\ell}$ and sends $s^*$ and $\dvk$ to $\qA$. 
     \item $\qA$ outputs $t^\prime=t_1^\prime\concat\ldots\concat t_n^\prime$. 
     The challenger outputs $1$ if $t_i^\prime=t_i^*$ where $t_i^*=\TEPRF.\Eval(\TEPRF.\sk_{i,0},s^*_i)$ for every $i\in [n]$ and otherwise $0$ as the final outcome of this experiment. 
        \end{enumerate}
By the definition, we clearly have 
\[
\Pr[\hybi{0}=1]=\advb{\UPFSKL,\qA}{up}{vra}(1^\secp).
\]



\takashi{I removed one hybrid here. As a result, the numbering of the remaining hybrids should be decremented by one. I did so carefully, but if you find this is not done somewhere, please just correct it.}

\item[$\hybi{1}$:] This is identical to $\hybi{0}$ except that the challenger generates $t^*_i$ using $\TEPRF.\sk_{i,x[i]}$ instead of $\TEPRF.\sk_{i,0}$ 
for $i\in [n]$ such that $\theta[i]=0$. 
This modification causes the change that $t^*_i$ is generated as  
\[
t^*_i\seteq \TEPRF.\Eval(\TEPRF.\key_{i,x[i]},s_i^*)
\]
for $i\in [n]$ such that $\theta[i]=0$.

By the equality property of TEPRF, we have 
\[
\TEPRF.\Eval(\TEPRF.\key_{i,x[i]},s_i^*)=
\TEPRF.\Eval(\TEPRF.\key_{i,0},s_i^*)
\]
unless $s_i^*=\tilde{s}_i$. 
Since $s_i^*$ is uniformly random, the probability that $s_i^*=\tilde{s}_i$ for some $i\in[n]$ is at most $n\cdot 2^{-\ell}\le \negl(\secp)$. Thus, we have 
\[
\abs{\Pr[\hybi{0}=1]-\Pr[\hybi{1}=1]}\le \negl(\secpar). 
\]

\item[$\hybi{2}$:] This is identical to $\hybi{1}$ except that the challenger samples $s^*=s_1^*\concat\ldots \concat s^*_n$ at the beginning of the experiment and sets $\tilde{s}_i\seteq s^*_i$ for $i\in [n]$ such that $\theta[i]=0$. 

Note that the only place where $\tilde{s}_i$ is used is in the generation of 
\[
(\TEPRF.\key_{i,0},\TEPRF.\key_{i,1})\gets \TEPRF.\qKG(1^\secpar,\tilde{s}_i).
\]
In $\hybi{2}$, uniformly random  $\tilde{s}_i$ that is independent of $s^*_i$ is used, but in $\hybi{3}$, we set $\tilde{s}_i=s^*_i$ for $i\in[n]$ such that $\theta[i]=0$. 
They are indistinguishable by the differing point hiding property of $\TEPRF$ since only one of  $\TEPRF.\key_{i,0}$ or $\TEPRF.\key_{i,1}$ is used in the hybrids for $i\in [n]$ such that $\theta[i]=0$. 
In fact, for $i\in [n]$ such that $\theta[i]=0$, 
$\qsk_i$ is a computational basis state which is either $\ket{0}\ket{\TEPRF.\key_{i,0}}$ or $\ket{1}\ket{\TEPRF.\key_{i,1}}$. 
Also, all the keys of $\TEPRF$ contained in the deletion verification key are $(\TEPRF.\key_{i,0},\TEPRF.\key_{i,1})_{i\in[n]:\theta[i]=1}$, and in particular the verification key contains no information about $\TEPRF.\key_{i,b}$ for $i\in[n]$ and $b\in \bit$ such that $\theta[i]=0$. 
Thus,  by a straightforward reduction to the differing point hiding property of $\TEPRF$, 
we have 
\[
\abs{\Pr[\hybi{1}=1]-\Pr[\hybi{2}=1]}\le \negl(\secpar). 
\] 
\item[$\hybi{3}$:]
This is identical to $\hybi{2}$ except that we change the way of sampling $\tilde{s}_i$ and $s^*_i$ for $i\in[n]$. 
Specifically, the challenger first samples $\tilde{s}_i\gets \bit^\ell$ at the beginning of the experiment without sampling $s^*_i$ for every $i\in[n]$. 
After receiving $\cert$ from $\qA$, the challenger sets 
$s^*_i\seteq \tilde{s}_i$ for every $i\in[n]$ such that $\theta[i]=0$ and 
samples $s^*_i\gets \bit^\ell$ for every $i\in[n]$ such that $\theta[i]=1$. 

Since $(\tilde{s}_i,s^*_i)$ is uniformly random conditioned  that $\tilde{s}_i=s^*_i$ for every $i\in[n]$ such that $\theta[i]=0$ in both hybrids, the change does not affect the view of $\qA$ and thus we have 
\[
\Pr[\hybi{2}=1]=\Pr[\hybi{3}=1]. 
\]
\end{description}

Below, by a reduction to the certified deletion property of BB84 states (\Cref{thm:CD_BB84}), 
we show 
\[
\Pr[\hybi{3}=1]\le\negl(\secp). 
\]

We construct an adversary $\qB=(\qB_0,\qB_1)$ against the certified deletion property of BB84 states (\Cref{thm:CD_BB84}) as follows. 
\begin{description}
\item[$\qB_0$:]Given an $n$-qubit BB84 state
$\ket{\psi}=\ket{x}_{\theta}$, 
$\qB_0$ parses 
$\ket{\psi}=\ket{\psi_1}\tensor\cdots\tensor\ket{\psi_n}$ and generates  $\tilde{s}_i\gets \bit^\ell$ for $i\in[n]$ and
$(\TEPRF.\key_{i,0},\TEPRF.\key_{i,1})\gets\TEPRF.\qKG(1^\secp,\tilde{s}_i)$ for every $i\in[n]$.
Let $U_i$ be a unitary that maps $\ket{u}\ket{v}$ to $\ket{u}\ket{v\oplus\TEPRF.\sk_{i,u}}$ where $u\in \bit$ and $v\in \bit^{\ell_\key}$. 
$\qB_0$ applies $U_i$ to $\ket{\psi_i}\ket{0^{\ell_\key}}$ to obtain $\qsk_i$ for every $i\in[n]$.
$\qB_0$ then sets $\qsk\seteq (\qsk_1,\ldots,\qsk_n)$ 
     and sends $\qsk$ to $\qA$.  When $\qA$ outputs $\cert=(e_i,d_i)_{i\in[n]}$, let $\qstateA$ be the state of $\qA$ at this point. 
     $\qB_0$ computes $y_i=e_i\oplus d_i\cdot(\TEPRF.\sk_{i,0}\oplus\TEPRF.\sk_{i,1})$ for every $i\in[n]$. 
     $\qB_0$ outputs $y\seteq y_1\concat\ldots\concat y_n$ and  
    $\qstate{st}_{\qB}\seteq ((\tilde{s}_i)_{i\in[n]},(\TEPRF.\sk_{i,b})_{i\in[n],b\in\bit},\qstateA)$.
\item[$\qB_1$:] Given $\theta$, $(x[i])_{i\in[n]:\theta[i]=1}$,  and $\qstate{st}_{\qB}=( (\tilde{s}_i)_{i\in[n]},(\TEPRF.\sk_{i,b})_{i\in[n],b\in\bit},\qstateA)$, $\qB_1$ sets $$\dvk^\prime\seteq((y_i)_{i\in[n]:\theta[i]=1},\theta,\allowbreak (\TEPRF.\sk_{i,0},\TEPRF.\sk_{i,1})_{i\in[n]:\theta[i]=1}).$$
$\qB_1$ sets $s^*_i\seteq \tilde{s}_i$ for $i\in [n]$ such that $\theta[i]=0$ and generates 
$s^*_i\gets \bit^\ell$ for $i\in [n]$ such that $\theta[i]=1$.   
$\qB_1$ then executes $\qA$ with the input $\dvk^\prime$ and $s^*=s^*_1\concat\ldots\concat s^*_n$ and the state $\qstateA$, and obtains an output $t^\prime=t^\prime_1\concat\ldots\concat t^\prime_n$.  
By the ``different values on target'' property of $\TEPRF$, there is unique $b_i\in \bit$ such that \[
t^\prime_i=\TEPRF.\Eval(\TEPRF.\sk_{i,b_i},s^*_i)
\]  
for every $i\in[n]$ such that $\theta[i]=0$. 
(Remember that each $t^\prime_i$ is a single bit.)
$\qB_1$ sets $z_i\seteq b_i$ for every $i\in[n]$ such that $\theta[i]=0$ and $z_i\seteq 0$ for every $i\in[n]$ such that $\theta[i]=1$,\footnote{
For $i\in[n]$ such that $\theta[i]=1$. 
$z_i$ can be set arbitrarily since it is irrelevant to the winning condition of the experiment for the certified deletion property of BB84 states.} and outputs $z=z_1\concat\ldots\concat z_n$.
\end{description}

$\qB$ perfectly simulates $\hybi{3}$ for $\qA$. 
If $\qA$ wins the simulated experiment (i.e., $\hybi{3}$ outputs $1$), we have
\begin{align}
e_i=x[i]\oplus d_i \cdot (\TEPRF.\sk_{i,0}\oplus\TEPRF.\sk_{i,1})
\end{align}
 for every $i\in[n]$ such that $\theta[i]=1$ 
and 
\[
t^\prime_i=\TEPRF.\Eval(\TEPRF.\sk_{i,x[i]},s^*_i)
\] 
for every $i\in[n]$ such that $\theta[i]=0$.
In this case, $y_i=x[i]$ for every $i\in [n]$ such that $\theta[i]=1$ and $z_i=x[i]$ for every $i\in [n]$ such that $\theta[i]=0$. 
Thus, we have 
\[
\Pr[\hybi{3}=1]\le\Pr[\CDBB(\qB,n)=1]\le \negl(\secpar)
\] 
where the latter inequality follows from the certified deletion property of BB84 states (\Cref{thm:CD_BB84}) 
and $n=\omega(\log\secpar)$. 

Combining the above, we have 
\[
\advb{\UPFSKL,\qA}{up}{vra}(1^\secp)\le \negl(\secpar).
\] 
This completes the proof of \cref{thm:UPF_UP-VRA}.
\end{proof}
By \Cref{lem:upf-prf,thm:UPF_UP-VRA}, we obtain \Cref{thm:PRF-SKL}. 


\section{Digital Signatures with Secure Key Leasing}\label{sec:const-DS-SKL-classical-certificate}
In this section, we prove the following theorem. 
\begin{theorem}\label{thm:DS-SKL}
If the SIS assumption holds,  then there exist RUF-VRA secure DS-SKL scheme with classical revocation.\footnote{See \Cref{sec:lattice_pre} for the precise meaning of the SIS assumption.} 
\end{theorem}
\mor{You use TEPRF as well. TEPRF is constructed from qSIS?}\takashi{I added an explanation below.}
Our construction of DS-SKL is based on a primitive called \emph{coherently-signable constrained signatures} which we introduce in this paper in addition to TEPRFs. 
In \Cref{sec:CS_def}, we define coherently-signable constrained signatures. In \Cref{sec:const_DS-SKL}, we construct RUF-VRA secure DS-SKL scheme with classical revocation using coherently-signable constrained signatures. 
We defer the construction of coherently-signable constrained signatures based on the SIS assumption 
to  \Cref{sec:CS}. 
 We remark that TEPRFs exist if the SIS assumption holds since SIS assumption implies OWFs and OWFs imply TEPRFs by \Cref{thm:OWF_to_TEPRF}.

\subsection{Coherently-Signable Constrained Signatures}\label{sec:CS_def}
\begin{definition}[Constrained signatures]\label{def:cs}
A constrained signatures (CS) with the message space $\Ms$ and constraint class $\Fs= \setbk{f \colon \Ms \ra \zo{}}$ is a tuple of four algorithms $(\Setup,\Constrain,\Sign,\Vrfy)$.
\begin{description}
\item[$\Setup(1^\secp) \ra (\vk,\msk)$:] The setup algorithm is a PPT algorithm that takes as input the security parameter $1^\secp$, and outputs a master signing key $\msk$ and a verification key $\vk$. 
\item[$\Constrain(\msk,f)\ra \sigk_f$:] The constrain algorithm is a PPT algorithm that takes as input the master signing key $\msk$ and a constraint $f\in \Fs$, and outputs a constrained signing key $\sigk_f$.
\item[$\Sign(\sigk_f,\msg) \ra \sigma$:] The signing algorithm is a PPT algorithm that takes as input a constrained key $\sigk_f$ and an message $\msg \in \Ms$, and outputs a signature $\sigma$.
\item[$\Vrfy(\vk,\msg,\sigma) \ra \top/\bot$:] The verification algorithm is a deterministic classical polynomial-time algorithm that takes as input a verification key $\vk$,  message $\msg\in\Ms$, and signature $\sigma$, and outputs $\top$ or $\bot$. 
\end{description}
\begin{description}
\item[Correctness:] 
For any $\msg\in \Ms$ and $f\in \Fs$ such that $f(\msg)=1$, \takashi{I swapped the roles of $f(\msg)=0$ and $f(\msg)=1$. I fixed the whole paper based on this modification, but if you find it's not done somewhere, please correct it.} we have 
\begin{align}
    \Pr\left[\Vrfy(\vk,\msg,\sigma)=\top:
    \begin{array}{l}
    (\vk,\msk)\gets \Setup(1^\secp)\\
    \sigk_f \gets \Constrain(\msk,f)\\
    \sigma \gets \Sign(\sigk_f,\msg)
    \end{array}
    \right]\ge 1-\negl(\secp).
\end{align}
\end{description}
\end{definition}

\begin{definition}[Selective single-key security]
We say that a CS scheme satisfies selective single-key security if for any stateful QPT $\qA$, we have 
\begin{align}
    \Pr\left[\Vrfy(\vk,\msg,\sigma)=\top
    \land f(\msg)=0
    :
    \begin{array}{l}
    f \gets \qA(1^\secpar)\\
    (\vk,\msk)\gets \Setup(1^\secp)\\
    \sigk_f \gets \Constrain(\msk,f)\\
    (\msg,\sigma) \gets \qA(\vk,\sigk_f)
    \end{array}
    \right]\le \negl(\secp).
\end{align} 
\end{definition}

We introduce an additional property which we call coherent signability. Roughly, this ensures that a superposition of two constrained signing keys can be used to generate a signature almost without collapsing the state. Looking ahead, this property is needed to achieve reusability of DS-SKL constructed in \Cref{sec:const-DS-SKL-classical-certificate}. 
\begin{definition}[Coherent-signability]\label{def:coherent_signing}
\takashi{I changed the definition to allow for entanglement with other registers.}
We say that a CS scheme is coherently-signable if for any polynomial $L=L(\secp)$,  there is a QPT algorithm $\QSign$ that takes a quantum state $\ket{\psi}$ and a classical message $\msg\in \Ms$ and outputs a quantum state $\ket{\psi'}$ and a signature $\sigma$, satisfying the following: 
\begin{enumerate}
\item \label{item:coherent_one}
Let $f\in \Fs$, 
$(\vk,\msk)\in \Setup(1^\secp)$, and   
\mor{$(\vk,\msk)\gets \Setup(1^\secp)$? Ika douyou} \takashi{$(\vk,\msk)\in \Setup(1^\secp)$ means that $(\vk,\msk)$ is in the range of $\Setup(1^\secp)$. I believe this notation is common in cryptography. I added the definition in Preliminaries.}  
$\sigk_{f}\in \Constrain(\msk,f)$, 
the output distribution of $\QSign\left(\ket{z
}\ket{\sigk_f},\msg\right)$ is identical to that of $\Sign(\sigk_f,\msg)$ for any $z\in \bit^L$. 
\item \label{item:coherent_two}
For any family $f_z\in \Fs$ for $z\in \bit^L$,  $(\vk,\msk)\in \Setup(1^\secp)$,  $\sigk_{f_z}\in \Constrain(\msk,f_z)$ for  $z\in \bit^L$,   and $\msg\in \Ms$ such that $f_z(\msg)=1$  for all $z\in \bit^L$, let $\ket{\psi}$ be a state of the form $\ket{\psi}=\sum_{z\in \bit^L}\alpha_z\ket{z}\ket{\sigk_{f_z}}$ for $\alpha_z\in \CC$ such that $\sum_{z\in \bit^L}|\alpha_z|^2=1$. 
Suppose that we run $(\ket{\psi'},\sigma)\gets\QSign(\ket{\psi},\msg)$. 
Then we have 
\begin{align}
    \|\ket{\psi}\bra{\psi}-\ket{\psi'}\bra{\psi'}\|_{tr}=\negl(\secp).
\end{align}
\end{enumerate}
\end{definition}

We prove the following theorem in \Cref{sec:CS}. 
\begin{theorem}\label{thm:CS}
Assuming the SIS assumption, for any polynomials $\ell=\ell(\secp)$ and $\iota=\iota(\secp)$, 
there is a selectively single-key secure and coherently-signable CS scheme for 
the function class $\mathcal{F}$ that consists of functions expressed as circuits of description size $\ell$ and input length $\iota$
\mor{This is
talking about $\cF$?} \takashi{I clarified it.}
\end{theorem}

\subsection{Construction of RUF-VRA Secure DS-SKL}\label{sec:const_DS-SKL} 
Let $\TEPRF=(\TEPRF.\ObvKG,\TEPRF.\KG,\TEPRF.\Eval)$ be a TEPRF with input length $\ell=\omega(\log \secpar)$. Let $\ell_\key$ be the length of a key generated by $\TEPRF.\KG(1^\secpar,s)$ for $s\in \bit^\ell$.\footnote{We can assume that all keys are of the same length w.l.o.g.} 
Let $\CS=(\CS.\Setup,\CS.\Constrain,\CS.\Sign,\CS.\Vrfy)$ be a constrained signature scheme for a function class $\Fs$ 
that satisfies selective single-key security and coherent signability with the coherent signing algorithm $\CS.\QSign$. 
We require that $\Fs$ contains the function $F[\TEPRF.\key]$ that appears in the description of the following scheme for any $\TEPRF.\key$ generated by $\TEPRF.\KG$.  
Let $\ell_\sigk$ be the length of a constrained signing key generated by $\CS.\Constrain(\msk,f)$ for $(\vk,\msk)\in \Setup(1^\secpar)$ and $f\in \Fs$.\footnote{We can assume that all constrained signing keys are of the same length w.l.o.g.} 
Then, we construct a DS-SKL scheme   $\DSSKL=(\DSSKL.\qKG,\DSSKL.\qSign, \DSSKL.\SigVrfy,\allowbreak\DSSKL.\qDel,\DSSKL.\DelVrfy)$ with classical revocation for $n\ell$-bit messages where $n=\omega(\log \secpar)$ as follows.

\begin{description}
\item[$\DSSKL.\qKG(1^\secp)$:] $ $
\begin{itemize}
    \item Generate $x,\theta\gets\bit^n$.
    \item $\tilde{s}_i\gets \bit^\ell$ for $i\in [n]$.
    \item Generate $(\TEPRF.\key_{i,0},\TEPRF.\key_{i,1})\gets\TEPRF.\KG(1^\secp,\tilde{s}_i)$ for $i\in [n]$.
    \item Generate $(\CS.\msk_i,\CS.\vk_i)\gets \CS.\Setup(1^\secpar)$ for $i\in[n]$. 
    \item Generate $\CS.\sigk_{i,b}\gets \CS.\Constrain(\CS.\msk_{i},F[\TEPRF.\key_{i.b}])$ for $i\in[n]$ and $b\in \bit$
    where $F[\TEPRF.\key_{i.b}]$ is a function that takes $(x,y)\in \bit^{\ell}\times \bit$ as input and outputs $1$ if $y=\TEPRF.\Eval(\TEPRF.\key_{i,b},x)$ 
    and otherwise outputs $0$. 
    \item Generate
    {\ifnum\llncs=1\footnotesize\fi
    \begin{align}
    \qsigk_i\seteq
    \begin{cases}
    \ket{x[i]}\ket{\TEPRF.\key_{i,x[i]}}\ket{\CS.\sigk_{i,x[i]}} & (\textrm{if~}\theta[i]=0)\\
    \frac{1}{\sqrt{2}}\left(\ket{0}\ket{\TEPRF.\key_{i,0}}\ket{\CS.\sigk_{i,0}}+(-1)^{x[i]}\ket{1}\ket{\TEPRF.\key_{i,1}}\ket{\CS.\sigk_{i,1}}\right) & (\textrm{if~}\theta[i]=1).
    \end{cases}
    \end{align}
    }
    \item Output 
    a quantum signing key 
    \[
    \qsigk\seteq(\qsigk_1,\ldots, \qsigk_n),
    \]
    a signature verification key
    \[
    \sigvk\seteq (\CS.\vk_1,\ldots,\CS.\vk_n).
    \]
    and a deletion verification key
    \[\dvk\seteq((x[i])_{i\in[n]:\theta[i]=1},\theta,(\TEPRF.\key_{i,0},\TEPRF.\key_{i,1},\CS.\sigk_{i,0},\CS.\sigk_{i,1})_{i\in[n]:\theta[i]=1}).
    \]
\end{itemize}

\item[$\DSSKL.\qSign(\qsigk,\msg)$:] $ $
\begin{itemize}
    \item Parse $\qsigk=(\qsigk_1,\ldots, \qsigk_n)$ 
    and $\msg=s_1\concat\ldots \concat s_n$ where $s_i\in \bit^\ell$ for each $i\in[n]$.
   \item Let $E_i$ be a unitary that works as follows: \[\ket{u}\ket{\TEPRF.\key}\ket{\CS.\sigk}\ket{v} \mapsto \ket{u}\ket{\TEPRF.\key}\ket{\CS.\sigk}\ket{v\oplus\TEPRF.\Eval(\TEPRF.\key,s_i)}\]
   where $u,v\in \bit$, $\TEPRF.\key\in \bit^{\ell_\key}$, and $\CS.\sigk\in \bit^{\ell_\sigk}$.  
   \item Apply $E_i$ to $\qsigk_i\tensor\ket{0}$ and measure the rightmost register to obtain $t_i$ for every $i\in[n]$. 
   Let $\widehat{\qsigk}_i$ be the state of the signing key after the measurement. 
   \item Generate $(\qsigk^{\prime}_i,\sigma_i)\gets \CS.\QSign(\widehat{\qsigk}_i,s_i\concat t_i)$ for $i\in [n]$ 
   where the first two registers of $\widehat{\qsigk}_i$ are regarded as the register for $z$ in \Cref{def:coherent_signing}. 
   \mor{In the definition of $\QSign$, it gets two registers as input, but not it gets three registers. Do you mean
   $\QSign$ acts on the first and third register?}\takashi{I clarified it.}
  \item Output 
  $\qsigk^\prime=(\qsigk^{\prime}_1,\ldots,\qsigk^{\prime}_n)$ and
  $\sigma\seteq (t_i,\sigma_i)_{i\in[n]}$. 
\end{itemize}

\item[$\DSSKL.\SigVrfy(\sigvk,\msg,\sigma)$:] $ $
\begin{itemize} 
    \item Parse $\sigvk=(\CS.\vk_1,\ldots,\CS.\vk_n)$, $\msg=s_1\concat\ldots \concat s_n$ where $s_i\in \bit^\ell$ for each $i\in[n]$, and $\sigma=(t_i,\sigma_i)_{i\in[n]}$. 
    \item  Output $\top$ if it holds that $\CS.\Vrfy(\vk_i,s_i\concat t_i,\sigma_i)=\top$ for every $i\in [n]$. Otherwise output $\bot$. 
\end{itemize}

\item[$\DSSKL.\qDel(\qsigk)$:] $ $
\begin{itemize}
\item Parse $\qsigk= (\qsigk_1,...,\qsigk_n)$.
\item Measure $\qsigk_i$ in the Hadamard basis and obtain $(e_i,d_i,c_i)$ for every $i\in[n]$.
\item Output $\cert\seteq (e_i,d_i,c_i)_{i\in [n]}$. 
\end{itemize}

\item[$\DSSKL.\DelVrfy(\dvk,\cert)$:] $ $
\begin{itemize}
    \item Parse $\dvk=((x[i])_{i\in[n]:\theta[i]=1},\theta,(\TEPRF.\key_{i,0},\TEPRF.\key_{i,1},\CS.\sigk_{i,0},\CS.\sigk_{i,1})_{i\in[n]:\theta[i]=1})$
    and $\cert=(e_i,d_i,c_i)_{i\in [n]}$. 
    \item Output $\top$ if it holds that 
    
\[e_i=x[i]\oplus d_i\cdot (\TEPRF.\key_{i,0}\oplus \TEPRF.\key_{i,1})
\oplus 
c_i\cdot (\CS.\sigk_{i,0}\oplus \CS.\sigk_{i,1})
\]
for every $i\in [n]$ such that $\theta[i]=1$. Otherwise, output $\bot$. 
\end{itemize}
\end{description}

\paragraph{\bf Signature verification correctness.}
We show that 
for any $\msg=s_1\concat\ldots\concat s_n\in \bit^{n\ell}$,  
$\CS.\Vrfy(\vk_i,s_i\concat t_i,\sigma_i)=\top$ holds for every $i\in[n]$ with overwhelming probability 
where 
$(\qsigk,\sigvk,\dvk)\gets\DSSKL.\qKG(1^\secp)$,   $t_i$ and $\sigma_i$ are computed as in $\DSSKL.\qSign(\qsigk,\msg)$. 

We first consider $i\in[n]$ such that $\theta[i]=0$. In this case, $\qsigk_i=\ket{x[i]}\ket{\TEPRF.\sk_{i,x[i]}}\ket{\CS.\sigk_{i,x[i]}}$. 
Then we have $t_i=\TEPRF.\Eval(\TEPRF.\sk_{i,x[i]},s_i)$ 
and thus we have $F[\TEPRF.\sk_{i,x[i]}](s_i\concat t_i)=1$. 
Thus, by the correctness of $\CS$, we have $\CS.\Vrfy(\vk_i,s_i\concat t_i,\sigma_i)=\top$. 

We next consider $i\in[n]$ such that $\theta[i]=1$. In this case, $\qsigk_i=\frac{1}{\sqrt{2}}(\ket{0}\ket{\TEPRF.\sk_{i,0}}\ket{\CS.\sigk_{i,0}}+(-1)^{x[i]}\ket{1}\ket{\TEPRF.\sk_{i,1}}\ket{\CS.\sigk_{i,1}})$.
Note that by the equality property of $\TEPRF$, we have $\Eval(\TEPRF.\sk_{i,0},s_i)=\Eval(\TEPRF.\sk_{i,1},s_i)$ unless $s_i=\tilde{s}_i$ where $\tilde{s}_i$ is the value chosen in $\DSSKL.\qKG$.  
In this case, we have $t_i=\Eval(\TEPRF.\sk_{i,0},s_i)=\Eval(\TEPRF.\sk_{i,1},s_i)$. 
Then we have $F[\TEPRF.\sk_{i,b}](s_i\concat t_i)=1$ for $b\in \bit$.  
Thus, by the correctness of $\CS$, we have $\CS.\Vrfy(\vk,s_i\concat t_i,\sigma_i)=\top$. 
Since the probability that $s_i=\tilde{s}_i$  is $2^{-\ell}\le \negl(\secpar)$, 
the above holds with overwhelming probability also in this case.

\paragraph{Reusability with static signing keys.} 
By the above proof of signature verification correctness, 
the measurement of $t_i$ does not collapse $\qsigk_i$ except for a negligible probability. Moreover, by the coherent signability, $\QSign$ only negligibly affects the state of $\qsigk_i$. Thus, the reusability with static signing keys holds.  

\paragraph{\bf Deletion verification correctness.}
We show that $e_i=x[i]\oplus d_i \cdot (\TEPRF.\sk_{i,0}\oplus\TEPRF.\sk_{i,1})
\oplus c_i\cdot (\CS.\sigk_{i.0}\oplus \CS.\sigk_{i,1})$ holds for every $(e_i,d_i,c_i)$ for $i\in[n]$ such that $\theta[i]=1$ where $(\qsigk,\sigvk,\dvk)\gets\DSSKL.\qKG(1^\secp)$ and 
$(e_i,d_i,c_i)_{i\in [n]}\gets \DSSKL.\qDel(\qsigk)$. 
For every $i\in[n]$ such that $\theta[i]=1$, we have 
\ifnum\llncs=0 $\qsigk_i=\frac{1}{\sqrt{2}}(\ket{0}\ket{\TEPRF.\sk_{i,0}}\ket{\CS.\sigk_{i,0}}+(-1)^{x[i]}\ket{1}\ket{\TEPRF.\sk_{i,1}}\ket{\CS.\sigk_{i,1}})$.
\else
\[\qsigk_i=\frac{1}{\sqrt{2}}(\ket{0}\ket{\TEPRF.\sk_{i,0}}\ket{\CS.\sigk_{i,0}}+(-1)^{x[i]}\ket{1}\ket{\TEPRF.\sk_{i,1}}\ket{\CS.\sigk_{i,1}}).\]
\fi
If we apply the Hadamard gate to every qubit of $\qsigk_i$, we obtain 
\ifnum\llncs=0
\begin{align}
&\frac{1}{\sqrt{2}}\left(
\begin{array}{ll}
\frac{1}{\sqrt{2^{\ell_\key+\ell_\sigk+1}}}\sum_{e_i,d_i,c_i} (-1)^{d_i\cdot\TEPRF.\sk_{i,0}\oplus c_i\cdot \CS.\sigk_{i,0}}\ket{e_i}\ket{d_i}\ket{c_i}\\
+
\frac{1}{\sqrt{2^{\ell_\key+\ell_\sigk+1}}}\sum_{e_i,d_i,c_i} (-1)^{x[i]\oplus e_i\oplus d_i\cdot\TEPRF.\sk_{i,1}\oplus c_i\cdot \CS.\sigk_{i,1}}\ket{e_i}\ket{d_i}\ket{c_i}
\end{array}
\right)
\\
=&
\frac{1}{\sqrt{2^{\ell_\key+\ell_\sigk}}}\sum_{
\substack{
(e_i,d_i):\\
d_i\cdot\TEPRF.\sk_{i,0}\oplus c_i\cdot \CS.\sigk_{i,0}=x[i]\oplus e_i\oplus d_i\cdot\TEPRF.\sk_{i,1}\oplus \CS.\sigk_{i,1}
}}(-1)^{d_i\cdot\TEPRF.\sk_{i,0}\oplus c_i\cdot \CS.\sigk_{i,0}}.\ket{e_i}\ket{d_i}\ket{c_i}.
\end{align}
\else
\begin{align}
&\frac{1}{\sqrt{2}}\left(
\begin{array}{ll}
\frac{1}{\sqrt{2^{\ell_\key+\ell_\sigk+1}}}\sum_{e_i,d_i,c_i} (-1)^{d_i\cdot\TEPRF.\sk_{i,0}\oplus c_i\cdot \CS.\sigk_{i,0}}\ket{e_i}\ket{d_i}\ket{c_i}\\
+
\frac{1}{\sqrt{2^{\ell_\key+\ell_\sigk+1}}}\sum_{e_i,d_i,c_i} (-1)^{x[i]\oplus e_i\oplus d_i\cdot\TEPRF.\sk_{i,1}\oplus c_i\cdot \CS.\sigk_{i,1}}\ket{e_i}\ket{d_i}\ket{c_i}
\end{array}
\right)
\\
=&
\frac{1}{\sqrt{2^{\ell_\key+\ell_\sigk}}}\sum_{
(e_i,d_i)\in S}(-1)^{d_i\cdot\TEPRF.\sk_{i,0}\oplus c_i\cdot \CS.\sigk_{i,0}}.\ket{e_i}\ket{d_i}\ket{c_i}
\end{align}
where $S$ consists of all $(e_i,d_i)$ satisfying 
\[
d_i\cdot\TEPRF.\sk_{i,0}\oplus c_i\cdot \CS.\sigk_{i,0}=x[i]\oplus e_i\oplus d_i\cdot\TEPRF.\sk_{i,1}\oplus \CS.\sigk_{i,1}. 
\]\fi
Thus, we have 
\ifnum\llncs=0 $e_i=x[i]\oplus d_i \cdot (\TEPRF.\sk_{i,0}\oplus\TEPRF.\sk_{i,1})\oplus c_i\cdot (\CS.\sigk_{i,0}\oplus \CS.\sigk_{i,1})$  
\else
\[e_i=x[i]\oplus d_i \cdot (\TEPRF.\sk_{i,0}\oplus\TEPRF.\sk_{i,1})\oplus c_i\cdot (\CS.\sigk_{i,0}\oplus \CS.\sigk_{i,1})\]
\fi
for every $i\in[n]$ such that $\theta[i]=1$.

\paragraph{\bf Security.} 

\begin{theorem}\label{thm:DS_RUF-VRA}
If $\TEPRF$ satisfies the properties of \Cref{def:TEPRF}, then $\DSSKL$ is RUF-VRA secure. 
\end{theorem}
\begin{proof}  
Let $\qA$ be a QPT adversary against RUF-VRA security of $\DSSKL$. 
We consider the following sequence of hybrids. 
\begin{description}
\item[$\hybi{0}$:] This is the original security experiment $\expb{\DSSKL,\qA}{ruf}{vra}(1^\secp)$. More specifically, it works as follows.
\begin{enumerate}
            \item  The challenger generates $x,\theta\gets\bit^n$, 
     $\tilde{s}_i\gets \bit^\ell$, $(\TEPRF.\key_{i,0},\TEPRF.\key_{i,1})\gets\TEPRF.\KG(1^\secp,\tilde{s}_i)$ for $i\in [n]$,
     $(\CS.\msk_i,\CS.\vk_i)\gets \CS.\Setup(1^\secpar)$ for $i\in[n]$, and 
    $\CS.\sigk_{i,b}\gets \CS.\Constrain(\CS.\msk_{i},F[\TEPRF.\key_{i.b}])$ for $i\in[n]$ and $b\in \bit$.
    The challenger also generates 
    {\ifnum\llncs=1\footnotesize\fi 
    \begin{align}
    \qsigk_i\seteq
    \begin{cases}
    \ket{x[i]}\ket{\TEPRF.\key_{i,x[i]}}\ket{\CS.\sigk_{i,x[i]}} & (\textrm{if~}\theta[i]=0)\\
    \frac{1}{\sqrt{2}}\left(\ket{0}\ket{\TEPRF.\key_{i,0}}\ket{\CS.\sigk_{i,0}}+(-1)^{x[i]}\ket{1}\ket{\TEPRF.\key_{i,1}}\ket{\CS.\sigk_{i,1}}\right) & (\textrm{if~}\theta[i]=1).
    \end{cases}
    \end{align}  
    }
     The challenger sets
  \begin{align}
    &\qsigk\seteq(\qsigk_1,\ldots, \qsigk_n),\\
    &\sigvk\seteq (\CS.\vk_1,\ldots,\CS.\vk_n),\\
&\dvk\seteq((x[i])_{i\in[n]:\theta[i]=1},\theta,(\TEPRF.\key_{i,0},\TEPRF.\key_{i,1},\CS.\sigk_{i,0},\CS.\sigk_{i,1})_{i\in[n]:\theta[i]=1}).
\end{align}
     The challenger sends $\qsigk$ and $\sigvk$ to $\qA$.  
     \item $\qA$ sends $\cert=(e_i,d_i,c_i)_{i\in[n]}$ to the challenger.
     The challenger outputs $0$ and the experiment ends if $e_i\neq x[i]\oplus d_i \cdot (\TEPRF.\sk_{i,0}\oplus\TEPRF.\sk_{i,1})\oplus c_i\cdot (\CS.\sigk_{i,0}\oplus \CS.\sigk_{i,1})$  for some $i\in[n]$ such that $\theta[i]=1$. Otherwise, the challenger chooses $\msg^*=s^*_1\concat\ldots \concat s^*_n\gets \bit^{n\ell}$ and sends $\msg^*$ and $\dvk$ to $\qA$. 
     \item $\qA$ outputs $\sigma^\prime=(t^\prime_i,\sigma^\prime_i)_{i\in[n]}$.   
     The challenger outputs $1$ if $\CS.\Vrfy(\CS.\vk_i,s^*_i\concat t^\prime_i,\sigma^\prime_i)=\top$ for every $i\in [n]$  and otherwise $0$ as the final outcome of this experiment. 
        \end{enumerate}
By the definition, we clearly have 
\[
\Pr[\hybi{0}=1]=\advb{\DSSKL,\qA}{ruf}{vra}(1^\secp).
\]



\takashi{I removed one hybrid here. As a result, the numbering of the remaining hybrids should be decremented by one. I did so carefully, but if you find this is not done somewhere, please just correct it.}

\item[$\hybi{1}$:] This is identical to $\hybi{0}$ except that the challenger outputs 
$0$ if $t^\prime_i\neq \TEPRF.\Eval(\TEPRF.\key_{i,x[i]},s^*_i)$ for some $i\in [n]$ such that $\theta[i]=0$. 

The difference between $\hybi{0}$ and $\hybi{1}$ occurs only if there is $i\in[n]$ such  that $\theta[i]=0$, 
$\CS.\Vrfy(\CS.\vk_i,s^*_i\concat t^\prime_i,\sigma_i^*)=\top$, and  $t^\prime_i\neq \TEPRF.\Eval(\TEPRF.\key_{i,x[i]},s^*_i)$. 
Note that $t^\prime_i\neq \TEPRF.\Eval(\TEPRF.\key_{i,x[i]},s^*_i)$ is equivalent to 
$F[\TEPRF.\key_{i,x[i]}](s^*_i,t^\prime_i)=0$.  
For $i\in [n]$ such that $\theta[i]=0$, 
$\qsigk_i$ is a computational basis state $\ket{0}\ket{\TEPRF.\key_{i,x[i]}}\ket{\CS.\sigk_{i,x[i]}}$ where $\CS.\sigk_{i,x[i]}$ is a constrained signing key w.r.t. $F[\TEPRF.\key_{i,x[i]}]$. 
Also, all the signing keys of $\CS$ contained in the deletion verification key $\dvk$ are $(\CS.\sigk_{i,0},\CS.\sigk_{i,1})_{i\in[n]:\theta[i]=1}$,  and in particular the deletion verification key contains no information about $\CS.\msk_{i}$ 
for $i\in[n]$ such that $\theta[i]=0$. 
Thus,  by a straightforward reduction to the selective single-key security of $\CS$,  the above event occurs only with a negligible probability. 
Then we have 
\[
\abs{\Pr[\hybi{0}=1]-\Pr[\hybi{1}=1]}\le \negl(\secpar). 
\]

\item[$\hybi{2}$:] This is identical to $\hybi{1}$ except that the challenger samples $s^*=s_1^*\concat\ldots \concat s^*_n$ at the beginning of the experiment and sets $\tilde{s}_i\seteq s^*_i$ for $i\in [n]$ such that $\theta[i]=0$. 

Note that the only place where $\tilde{s}_i$ is used is in the generation of 
\[
(\TEPRF.\key_{i,0},\TEPRF.\key_{i,1})\gets \TEPRF.\KG(1^\secpar,\tilde{s}_i).
\]
In $\hybi{2}$, uniformly random  $\tilde{s}_i$ that is independent of $s^*_i$ is used, but in $\hybi{3}$, we set $\tilde{s}_i=s^*_i$ for $i\in[n]$ such that $\theta[i]=0$. 
They are indistinguishable by the differing point hiding property of $\TEPRF$ since only one of $\TEPRF.\key_{i,0}$ or $\TEPRF.\key_{i,1}$ is used in the hybrids for $i\in [n]$ such that $\theta[i]=0$. 
In fact, for $i\in [n]$ such that $\theta[i]=0$, 
$\qsigk_i$ is a computational basis state  $\ket{0}\ket{\TEPRF.\key_{i,x[i]}}\ket{\CS.\sigk_{i,x[i]}}$. 
Also, all the keys of $\TEPRF$  contained in 
the deletion verification key are   $(\TEPRF.\key_{i,0},\TEPRF.\key_{i,1})_{i\in[n]:\theta[i]=1}$, 
and in particular the deletion verification key contains no information about $\TEPRF.\key_{i,b}$ for $i\in[n]$ and $b\in \bit$ such that $\theta[i]=0$.  
Thus,  by a straightforward reduction to the differing point hiding property of $\TEPRF$, 
we have 
\[
\abs{\Pr[\hybi{1}=1]-\Pr[\hybi{2}=1]}\le \negl(\secpar). 
\] 
\item[$\hybi{3}$:]
This is identical to $\hybi{2}$ except that we change the way of sampling $\tilde{s}_i$ and $s^*_i$ for $i\in[n]$. 
Specifically, the challenger first samples $\tilde{s}_i\gets \bit^\ell$ at the beginning of the experiment without sampling $s^*_i$ for every $i\in[n]$. 
After receiving $\cert$ from $\qA$, the challenger sets 
$s^*_i\seteq \tilde{s}_i$ for every $i\in[n]$ such that $\theta[i]=0$ and 
samples $s^*_i\gets \bit^\ell$ for every $i\in[n]$ such that $\theta[i]=1$. 

Since $(\tilde{s}_i,s^*_i)$ is uniformly random conditioned  that $\tilde{s}_i=s^*_i$ for every $i\in[n]$ such that $\theta[i]=0$ in both hybrids, the change does not affect the view of $\qA$ and thus we have 
\[
\Pr[\hybi{2}=1]=\Pr[\hybi{3}=1]. 
\]
\end{description}

Below, by a reduction to the certified deletion property of BB84 states (\Cref{thm:CD_BB84}), 
we show 
\[
\Pr[\hybi{3}=1]\le\negl(\secp). 
\] 
We construct an adversary $\qB=(\qB_0,\qB_1)$ against the certified deletion property of BB84 states (\Cref{thm:CD_BB84}) as follows. 
\begin{description}
\item[$\qB_0$:] Given an $n$-qubit BB84 state
$\ket{\psi}=\ket{x}_{\theta}$, 
$\qB_0$ parses 
$\ket{\psi}=\ket{\psi_1}\tensor\cdots\tensor\ket{\psi_n}$ and generates 
$\tilde{s}_i \gets \bit^\ell$ for $i\in[n]$, $(\TEPRF.\key_{i,0},\TEPRF.\key_{i,1})\gets\TEPRF.\KG(1^\secp,\tilde{s}_i)$ for $i\in [n]$,
     $(\CS.\msk_i,\CS.\vk_i)\gets \CS.\Setup(1^\secpar)$ for $i\in[n]$, and 
   $\CS.\sigk_{i,b}\gets \CS.\Constrain(\CS.\msk_{i},F[\TEPRF.\key_{i.b}])$ for $i\in[n]$ and $b\in \bit$.
Let $U_i$ be a unitary that maps $\ket{u}\ket{v}\ket{w}$ to $\ket{u}\ket{v\oplus\TEPRF.\sk_{i,u}}\ket{w\oplus \CS.\sigk_{i,u}}$ where $u\in \bit$, $v\in\bit^{\ell_\key}$, and $w\in\bit^{\ell_\sigk}$.
$\qB_0$ applies $U_i$ to $\ket{\psi_i}\ket{0^{\ell_\key}}\ket{0^{\ell_\sigk}}$ to obtain $\qsigk_i$ for every $i\in[n]$.
$\qB_0$ then sets $\qsigk\seteq (\qsigk_1,\ldots,\qsigk_n)$ 
and $\sigvk\seteq (\CS.\vk_1,\ldots,\CS.\vk_n)$ 
     and sends $\qsigk$ and $\sigvk$ to $\qA$.  When $\qA$ outputs $\cert=(e_i,d_i,c_i)_{i\in[n]}$, let $\qstateA$ be the state of $\qA$ at this point. 
     $\qB_0$ computes $y_i=e_i\oplus d_i\cdot(\TEPRF.\sk_{i,0}\oplus\TEPRF.\sk_{i,1})\oplus c_i\cdot (\CS.\sigk_{i,0}\oplus \CS.\sigk_{i,1})$ for every $i\in[n]$. 
     $\qB_0$ outputs $y\seteq y_1\concat\ldots\concat y_n$ and  
    $\qstate{st}_{\qB}\seteq ((s^*_i)_{i\in[n]},(\TEPRF.\sk_{i,b},\CS.\sigk_{i,b})_{i\in[n],b\in\bit},\qstateA)$.
\item[$\qB_1$:] Given $\theta$ and $\qstate{st}_{\qB}= ((s^*_i)_{i\in[n]},(\TEPRF.\sk_{i,b},\CS.\sigk_{i,b})_{i\in[n],b\in\bit},\qstateA)$, $\qB_1$ sets  $$\dvk^\prime\seteq((y_i)_{i\in[n]:\theta[i]=1},\theta,\allowbreak (\TEPRF.\sk_{i,0},\TEPRF.\sk_{i,1},\CS.\sigk_{i,0},\CS.\sigk_{i,1})_{i\in[n]:\theta[i]=1}).$$ 
$\qB_1$ sets $s^*_i\seteq \tilde{s}_i$ for $i\in [n]$ such that $\theta[i]=0$ and generates 
$s^*_i\gets \bit^\ell$ for $i\in [n]$ such that $\theta[i]=1$.  
$\qB_1$ then executes $\qA$ with the input $\dvk^\prime$ and $\msg^*=s^*_1\concat\ldots\concat s^*_n$ and the state $\qstateA$, and obtains an output $\sigma^\prime=(t^\prime_i,\sigma^\prime_i)_{i\in [n]}$. 
By the ``different values on target'' property of TEPRF, there is unique $b_i\in \bit$ such that \[
t^\prime_i=\TEPRF.\Eval(\TEPRF.\sk_{i,b_i},s^*_i)
\] 
for every $i\in[n]$ such that $\theta[i]=0$. 
$\qB_1$ sets $z_i\seteq b_i$ for every $i\in[n]$ such that $\theta[i]=0$ and $z_i\seteq 0$ for every $i\in[n]$ such that $\theta[i]=1$,\footnote{
For $i\in[n]$ such that $\theta[i]=1$. 
$z_i$ can be set arbitrarily since it is irrelevant to the winning condition of the experiment for the certified deletion property of BB84 states.} and outputs $z=z_1\concat\ldots\concat z_n$. 
\end{description}

$\qB$ perfectly simulates $\hybi{3}$ for $\qA$. 
If $\qA$ wins the simulated experiment (i.e., $\hybi{3}$ outputs $1$), we have
\begin{align}
e_i=x[i]\oplus d_i \cdot (\TEPRF.\sk_{i,0}\oplus\TEPRF.\sk_{i,1})\oplus c_i\cdot (\CS.\sigk_{i,0}\oplus \CS.\sigk_{i,1})
\end{align}
 for every $i\in[n]$ such that $\theta[i]=1$ 
and 
\[
t^\prime_i=\TEPRF.\Eval(\TEPRF.\sk_{i,x[i]},s^*_i)
\]  
for every $i\in[n]$ such that $\theta[i]=0$.
In this case, $y_i=x[i]$ for every $i\in [n]$ such that $\theta[i]=1$ and $z_i=x[i]$ for every $i\in [n]$ such that $\theta[i]=0$. 
Thus, we have 
\[
\Pr[\hybi{3}=1]\le\Pr[\CDBB(\qB,n)=1]\le \negl(\secpar)
\] 
where the latter inequality follows from the certified deletion property of BB84 states (\Cref{thm:CD_BB84}) 
and $n=\omega(\log\secpar)$.

Combining the above, we have 
\[
\advb{\DSSKL,\qA}{ruf}{vra}(1^\secp)\le \negl(\secpar).
\] 
This completes the proof of \cref{thm:DS_RUF-VRA}.
\end{proof}
By \Cref{thm:CS,thm:DS_RUF-VRA}, we obtain \Cref{thm:DS-SKL}. 

\section{Construction of Coherently-Signable Constrained Signatures}\label{sec:CS}
In this section, we prove \Cref{thm:CS}, that is, we construct selectively single-key secure coherently- signable constrained signatures based on the SIS assumption.   
In \Cref{sec:lattice_pre}, we provide necessary background on lattices. In \Cref{sec:const_CS}, we give the construction of selectively single-key secure coherently- signable constrained signatures and its security proof.

\subsection{Lattice Preliminaries}\label{sec:lattice_pre} 
A lattice $\Lambda\subseteq \RR^m$ is a discrete additive subgroup of $\RR^m$. 
A shifted lattice is a subset that can be written as $\Lambda+\vt=\{\vx+\vt:\vx\in \Lambda\}$ by using a lattice $\Lambda$ and a vector $\vt\in \RR^n$. 
For $q\in \NN$, we write $\ZZ_q$ to mean the  additive group of integers modulo $q$. We represent elements of $\ZZ_q$ by elements in $\ZZ\cap (q/2,q/2]$ and often identify the elements of $\ZZ_q$ and their representatives. 
For $\mA\in \ZZ_q^{n\times m}$ and $\vy\in \ZZ_q^n$, we define 
$\Lambda_\mA^\bot:=\{\vx\in \ZZ^m:\mA\vx=\vzero \mod q\}$ and
$\Lambda_\mA^\vy:=\{\vx\in \ZZ^m:\mA\vx=\vy \mod q\}$. It is easy to see that 
$\Lambda_\mA^\bot$ is a lattice and
$\Lambda_\mA^\vy$ is a shifted lattice written as $\Lambda_\mA^\vy=\Lambda_\mA^\bot+ \vt$ for any $\vt\in \Lambda_\mA^\vy$.  

For a vector $\vx=(x_1,x_2,\ldots,x_n)\in \mathbb{R}^n$, we define 
$\|\vx\|_{\infty}\seteq \max_{i\in [n]}|x_i|$. 
Similarly, for a matrix $\mA=(a_{i,j})_{i\in[n],j\in[m]}\in \mathbb{R}^{n\times m}$, we define $\|\mA\|_\infty\seteq \max_{i\in[n],j\in[m]}a_{i,j}$. 
We often use the inequality $\|\mA\mB\|_\infty\le m\|\mA\|_\infty\|\mB\|_\infty$ for $\mA\in \RR^{n\times m}$ and $\mB\in\RR^{m\times \ell}$. 

For $m\in \NN$, the Gaussian function $\rho_\sigma:\RR^m \ra [0,1]$ with a scaling parameter $\sigma>0$ is defined as follows:
\[
\rho_{\sigma}(\vx)\seteq \exp(-\pi\|\vx\|^2/\sigma^2).
\]
For any countable set $S\subseteq \mathbb{R}^m$, let 
$\rho(S)\seteq \sum_{\vx \in S}\rho_{\sigma}(\vx)$ and  
if $0<\rho(S)<\infty$, 
we define the discrete Gaussian distribution over $S$ by the following probability density function: 
\[
D_{S,\sigma}(\vx)=\frac{\rho_\sigma(\vx)}{\rho_\sigma(S)}
\]
for $\vx\in S$. 

We rely on the following tail bound about discrete Gaussian. 

\begin{definition}
For $m\in \mathbb{N}$ and $R>0$, we define 
$$\mathcal{B}^m_\infty(R)=\{\vx\in \RR^m:\|\vx\|_\infty\le R\}.$$
\end{definition}

\begin{lemma}[{\cite[Lemma 2.10]{Banaszczyk95}}]\label{lem:tail_infty}
For any $m$-dimensional lattice $\Lambda$,  $\sigma>0$, and $r>0$,  
\[
\rho_\sigma(\Lambda\setminus  \mathcal{B}^m_\infty(r\sigma))\le \left(2m e^{-\pi r^2}\right)\rho_{\sigma}(\Lambda).
\]
\end{lemma}
\begin{corollary}\label{cor:distribution_close}
For $n,m\in \NN$, an integer $q\ge 2$, 
$\mA\in \ZZ_q^{n\times m}$, and $\sigma>0$, the statistical distance between $D_{\Lambda_\mA^\bot,\sigma}$ and $D_{\Lambda_\mA^\bot\cap \ZZ_q^m,\sigma}$ is at most $2m e^{-\pi \left(\frac{q/2-1}{\sigma}\right)^2}$.\footnote{Recall that an element of $\ZZ_q$ is represented as an element of $\ZZ \cap (-q/2,q/2]$.
In particular, $\Lambda_\mA^\bot\cap \ZZ_q^m$ means $\Lambda_\mA^\bot\cap (-q/2,q/2]^m$. 
}  
\end{corollary}
\begin{proof}
The statistical distance between $D_{\Lambda_\mA^\bot,\sigma}$ and $D_{\Lambda_\mA^\bot\cap \ZZ_q^m,\sigma}$ is:
\begin{align}
\frac{\rho_{\sigma}(\Lambda_\mA^\bot\setminus \ZZ_q^m)}{\rho_{\sigma}(\Lambda_\mA^\bot)}
\le 
\frac{\rho_{\sigma}(\Lambda_\mA^\bot\setminus B^m_\infty(q/2-1))}{\rho_{\sigma}(\Lambda_\mA^\bot)}
\le 
2m e^{-\pi \left(\frac{q/2-1}{\sigma}\right)^2}
\end{align}
 where the first inequality follows
$B^{m}_\infty(q/2-1)\subseteq \ZZ_q^m$, 
and the second inequality follows from \Cref{lem:tail_infty}. 
\end{proof}

By \cite[Lemma 2.11]{TCC:PeiRos06} and
\cite[Lemma 5.3]{STOC:GenPeiVai08}, we have the following lemma.  
\begin{lemma}[{\cite[Lemma 2.5]{TCC:AnaPorVai23} }]\label{lem:prob_zero}  
Let $n,m\in \NN$ and $q$ be a prime such that $m\ge 2n\log q$. 
For an at least $(1-2q^{-n})$-fraction of $\mA\in \ZZ_q^{n\times m}$, 
for any 
$\vy\in \ZZ_q^n$, 
$\vx^*\in \Lambda_\mA^\vy$,  
and any $\sigma=\omega(\sqrt{\log m})$, 
\begin{align}
    \Pr_{\vx\gets D_{\Lambda_\mA^\vy ,\sigma}}[\vx=\vx^*]\le 2^{-\Omega(m)}. 
\end{align}
\end{lemma}
\begin{remark}\label{rem:additional_assumption}
\cite[Lemma 2.5]{TCC:AnaPorVai23} claims that the above holds for all full-rank $\mA$ by citing 
\cite[Lemma 2.11]{TCC:PeiRos06} and
\cite[Lemma 5.3]{STOC:GenPeiVai08}. However, we found that we need to assume an additional property  of $\mA$ to apply \cite[Lemma 5.3]{STOC:GenPeiVai08}.\footnote{Concretely, we need to assume that 
for all $\vs\in \ZZ_q^n\setminus \{\vzero\}$, $\|\mA^\trans \vs \mod q\|_\infty\ge q/4$.
}
Fortunately, it is proven in \cite[Lemma 5.3]{STOC:GenPeiVai08} that an at least $(1-q^{-n})$-fraction of $\mA\in \ZZ_q^{n\times m}$ satisfies the required property.  Since an at least $(1-q^{-n})$-fraction of $\mA\in \ZZ_q^{n\times m}$ is full-rank (e.g., \cite[Lemma 5.1]{STOC:GenPeiVai08}), by taking the union bound,  an at least $(1-2q^{-n})$-fraction of $\mA\in \ZZ_q^{n\times m}$ satisfies both. 
\end{remark}

We rely on the quantum hardness of the short integer solution (SIS) problem defined below. 
\begin{definition} [Short Integer Solution Problem] \label{def:sis}	
 For parameters $n=n(\secp), m=m(\secp), q(\secp), \beta=\beta(\secp)$,  the quantum hardness of the \emph{short integer solution} (SIS) problem $\SIS_{n, m, q, \beta}$ means that for any QPT algorithm $\mathcal A$:
\begin{align}
\Pr[\mA\vu = \vzero \mod q ~\land~ \|\vu\|_\infty\le \beta(\secp) ~\land~ \vu \neq \vzero:
\mA\gets \ZZ_q^{m\times n}, \vu\gets \A(\mA)
]\le \negl(\secp).
\end{align}
\end{definition}

The parameter regime for the SIS problem considered in this work is similar to that in \cite{STOC:GorVaiWic15}.
We assume that for any $\beta=2^{\poly(\secp)}$, there are some $n=\poly(\secpar)$ and a prime $q=2^{\poly(\secpar)}$ such that for any $m=\poly(\secpar)$, the quantum hardness of $\SIS_{n, m, q, \beta}$ holds. (Note that we must have $q\ge \beta$ since otherwise $\SIS_{n, m, q, \beta}$ can be trivially solved.) We simply say that the SIS assumption holds to mean that the above is true.  
By Regev's reduction \cite{JACM:Regev09}, 
the assumption can be reduced to the quantum hardness of certain worst-case
lattice problems (e.g., ${\sf SIVP}$ or $\GapSVP$) with sub-exponential approximation factors, which is widely believe at the moment.

\paragraph{\bf Gadget Matrix.}
Let $n,k,q$ be positive integers such that $k\geq n\lceil \log q\rceil$.
A gadget matrix $\mG$ is defined as $\mI_n\otimes (1,2,...,2^{\lceil \log q\rceil-1})$ padded with $k - n\lceil \log q\rceil$ zero columns. 
The parameters $n,k,q$ are not explicitly written in the notation, but they will be clear from the context.  

\paragraph{\bf Trapdoors.}
There has been a long line of works \cite{Stoc:Ajtai96,STOC:GenPeiVai08,EC:AgrBonBoy10,JC:CHKP12,C:AgrBonBoy10,EC:MicPei12} about methods to sample a matrix along with a trapdoor that enables us to efficiently solve the SIS and LWE problems.
In particular, we rely on the following lemmas.

\begin{lemma}[{\cite[Lemma 2.2]{STOC:GorVaiWic15}}]\label{lem:trapdoor} 
There exist PPT algorithms $\TrapGen$, $\Sam$,  $\SamPre$  such that the following holds. Given integer $n\ge 1$, $q\ge 2$,  there exists some
$m^*=m^*(n,q)=O(n\log q)$, $\beta_{sam}=\beta_{sam}(n,q)=O(n\sqrt{\log q})$ such that for all $m\ge m^*$ and all $k=\poly(n)$, we have:
\begin{enumerate}
\item $\mR\gets \Sam(1^m,1^k,q)$ samples a matrix $\mR\in \ZZ_q^{m\times k}$ which satisfies $\|\mR\|_\infty\le \beta_{sam}$.
\item We have 
\[
\Delta(\mA,\mA')\le \negl(n)
\]
and
\[
\Delta((\mA,\td,\mR,\mV),(\mA,\td,\mR',\mV'))\le \negl(n)
\]
where $(\mA,\td)\gets \TrapGen(1^n,1^m,q)$, $\mA'\gets \ZZ_q^{n\times m}$, $\mR\gets \Sam(1^m,1^k,q)$, $\mV:=\mA\cdot \mR$, $\mV'\gets \ZZ_q^{n\times k}$, $\mR'\gets \SamPre(\mA,\mV',\td)$. 
 Moreover,
we guarantee that any
$\mR'\in \SamPre(\mA,\mV',\td)$ always satisfies 
$\mV'=\mA\cdot \mR'$ and $\|\mR'\|_\infty\le \beta_{sam}$.  
\end{enumerate}
\end{lemma}

\begin{definition}\label{def:LWE_trapdoor}
Let $\Invert$ be a classical deterministic polynomial-time algorithm. 
For  $n,m,q\in \NN$,    
$\beta_{err}>0$,
and $\mA\in \ZZ_q^{n\times m}$, 
we say that a classical string $\td$ is
a $\beta_{err}$-LWE-trapdoor of $\mA$ w.r.t. $\Invert$ if the following holds.  For any
$\mathbf{s}\in \mathbb{Z}_q^{n}$ and $\mathbf{e}\in \mathbb{Z}_q^{m}$  
such that $\|\mathbf{e}\|_{\infty}\le \beta_{err}$,  
we have 
\[
\Invert(\mathbf{A},\td,\mathbf{s}^\trans \mA+\ve^\trans \mod q)=(\vs,\ve).
\]
\end{definition}

\begin{lemma}[\cite{EC:MicPei12}]\label{lem:R_is_trapdoor}
There is a classical deterministic polynomial-time algorithm $\Invert$ such that the following holds. 
For $n,q,m,k\in \NN$ such that 
$q\ge 2$ and
$k\geq n\lceil \log q\rceil$,   
$\mathbf{A}\in \mathbb{Z}_q^{n\times m}$, 
and $\mathbf{R}\in \mathbb{Z}_q^{m\times k}$, 
$\mathbf{R}$ is a 
$\beta_{err}$-LWE-trapdoor of $\left[\mathbf{A}\mid \mathbf{A}\mathbf{R}+\mathbf{G}\right]$ w.r.t. $\Invert$ for  $\beta_{err}=q/(5k(m\|\mR\|_\infty+1))$ 
where $\mG\in \ZZ_q^{n\times k}$ is the gadget matrix. 
\end{lemma}

Though \Cref{lem:R_is_trapdoor} is implicitly shown in \cite{EC:MicPei12}, we give a proof in \Cref{sec:proof_MR} for completeness.

\paragraph{\bf Lattice Evaluation.}
The following lemma is an abstraction of the evaluation algorithms developed in a long sequence of works~\cite{EC:MicPei12,C:GenSahWat13,EC:BGGHNS14,STOC:GorVaiWic15}.
\begin{lemma}[{\cite[Theorem 2.5]{C:Tsabary19}}]\label{lem:lattice_eval}
There exist classical deterministic polynomial-time algorithms $\EvalF$ and $\EvalFX$ such that for any positive integers $n,q,\ell_{\mathsf{in}},\ell_{\mathsf{out}}$ and $m\geq n\lceil  \log q\rceil$, for any depth $d$ boolean circuit $f:\bit^{\ell_{\mathsf{in}}}\rightarrow \bit^{\ell_{\mathsf{out}}}$, input $x\in \bit^{\ell_{\mathsf{in}}}$ with $x[1]=1$,\footnote{This assumption is not explicitly stated in the original statement of \cite[Theorem 2.5]{C:Tsabary19}, but it is pointed out in \cite[Remark 3.2]{AC:KNYY20} that it would have been necessary.} and matrix $\mA\in \ZZ_q^{n\times m\ell }$, the outputs $\mH:= \EvalF(f,\mA)$ and $\widehat{\mH}:=\EvalFX(f,x,\mA)$ are both in $\ZZ^{m\ell_{\mathsf{in}}\times m\ell_{\mathsf{out}}}$ and it holds that $\|\mH\|_{\infty},\|\widehat{\mH}\|_{\infty}\leq (2m)^d$, and 
\[
[\mA-x\otimes \mG]\widehat{\mH}=\mA\mH-f(x)\otimes \mG \mod q
\]
where we regard $x$ and $f(x)$ as a row vector and $\mG\in \ZZ_q^{n\times m}$ is the gadget matrix. 
\end{lemma}

\paragraph{\bf Gaussian Superpositions.} 
The following lemma about efficient generation of Gaussian superpositions is shown in \cite{TCC:AnaPorVai23} based on the  Grover-Rudolph algorithm \cite{GR02}.
\begin{lemma}[{\cite[Sec. 3.1]{TCC:AnaPorVai23} }]\label{lem:Gaussian_superposition}
There is a quantum algorithm that given $m\in \NN$, $q\ge 2$, and $\sigma\in (\sqrt{8m},q/\sqrt{8m})$, and outputs a quantum state within trace distance $2^{-\Omega(m)}$ of the normalized variant of the following state:
\[
\sum_{\ve \in \ZZ_q^m}\rho_\sigma(\ve)\ket{\ve}.
\]
The algorithm runs in time $\poly(m,\log q)$. 
\end{lemma}

\paragraph{\bf Duality Lemma.} 
\begin{lemma}[{\cite[Lemma 16]{ITCS:Poremba23}}]\label{lem:duality}
Let $n,m\in \NN$, $q\ge 2$ be a prime modulus and $\sigma\in (\sqrt{8m},q/\sqrt{8m})$. 
For an at least $(1-2q^{-n})$-fraction of $\mA\in \ZZ_q^{n\times m}$ and any $\vy\in \ZZ_q^n$, 
if we apply the $q$-ary quantum Fourier transform to the normalized variant of the following state: 
\[
\ket{\psi_{\vy}}:=\sum_{\vs\in \ZZ_q^n} \sum_{\ve\in\ZZ_q^m}\rho_{q/\sigma}(\ve)\omega_q^{-\langle\vs,\vy\rangle}\ket{\vs^\trans \mA + \ve^\trans \mod q},
\]
then the resulting state is within trace distance 
 $\negl(m)$ of the normalized variant of the following state:
\[
\ket{\hat{\psi}_\vy}:=\sum_{\vx\in \ZZ_q^m\cap \Lambda_\mA^\vy
}
\rho_\sigma(\vx)\ket{\vx}.
\] 
\end{lemma}
\begin{remark}\label{rem:correction_2}
The original statement of \cite[Lemma 16]{ITCS:Poremba23} claims that the above holds for all full-rank $\mA$. However, since their proof relies on \cite[Lemma 5.3]{STOC:GenPeiVai08}, 
we believe that an additional property for $\mA$ is needed by the same reason as explained in \Cref{rem:additional_assumption}. As explained in  \Cref{rem:additional_assumption}, an at least $(1-2q^{-n})$-fraction of $\mA\in \ZZ_q^{n\times m}$ is full-rank and satisfies the required property for \cite[Lemma 5.3]{STOC:GenPeiVai08}. This is reflected in the above statement of \Cref{lem:duality}.
\end{remark}


\paragraph{\bf Quantum Gaussian Sampling.} \takashi{I changed the statement of the theorem to clarify the meaning of the state after the execution.}
Ananth et al. \cite{TCC:AnaPorVai23} showed that a LWE-trapdoor can be used to generate a Gaussian superposition over the corresponding lattice (i.e., a state of the form of $\ket{\hat{\psi}_\vy}$ in \Cref{lem:duality}). We observe that a similar procedure works to generate a Gaussian superposition given a \emph{superposition} of LWE-trapdoors almost without collapsing the superposition.  This is stated in the following theorem. 
\begin{theorem}[Generalization of {\cite[Theorem 3.3]{TCC:AnaPorVai23}}]\label{thm:QSampGauss}
Let $n,m$ be positive integers 
and $q$ be a prime. 
For any classical deterministic polynomial-time algorithm $\mathsf{Invert}$, there exist a QPT algorithm $\QSampGauss$ 
such that for at least $(1-2q^{-n})$-fraction of $\mA\in \ZZ_q^{n\times m}$, the following holds. 
Let $\sigma\in (\sqrt{8m},q/\sqrt{8m})$, 
$\vy\in \ZZ_q^n$, 
and $\ket{\psi}$ be a state of the form 
$\sum_{z\in \bit^L}\alpha_z\ket{z}\ket{\td_z}$
where $\alpha_z\in \mathbb{C}$ satisfies $\sum_{z\in \bit^L}|\alpha_z|^2$ and $\td_z$ is a $\beta_{err}$-LWE-trapdoor of $\mA$ w.r.t. $\mathsf{Invert}$ 
for $\beta_{err}\ge \frac{q}{\sigma}\left(\frac{n\log q+m}{\pi}\right)^{1/2}$. 
Then $\QSampGauss(\mA,\ket{\psi},\vy,\sigma)$ outputs states $\rho$ and $\tau$ where
$\rho$ is within trace distance $\negl(m)$ of $\ket{\psi}$ and $\tau$ is
within trace distance $\negl(m)$ of the normalized variant of the following state:  
\[
\sum_{\vx\in \ZZ_q^m\cap \Lambda_\mA^\vy}
\rho_\sigma(\vx)\ket{\vx}.
\]   
\end{theorem}

Though the proof of \Cref{thm:QSampGauss} is very similar to that of \cite[Theorem 3.3]{TCC:AnaPorVai23}, we give a full proof since there seems to be a (fixable) gap in their proof (see \Cref{rem:gap}).

\begin{proof}[Proof of \Cref{thm:QSampGauss}]
For ease of exposition, we describe how $\QSampGauss$ works when 
$\ket{\psi}$ is a computational basis state $\ket{z}\ket{\td_z}$ for some $z\in \bit^L$ and
$\beta_{err}$-LWE-trapdoor $\td_z$.  When $\ket{\psi}$ is a superposition over $\ket{z}\ket{\td_z}$, it just runs the same procedure coherently. 

\begin{enumerate}
\item[Algorithm $\QSampGauss$]
\item[Input:] $\mA\in \ZZ_q^{n\times m}$, $\ket{\psi}=\ket{z}\ket{\td_z}$, $\vy\in \ZZ_q^n$, $\sigma\in (\sqrt{8m},q/\sqrt{8m})$.
\item 
Use the algorithm of \Cref{lem:Gaussian_superposition}
to generate a quantum state within trace distance $2^{-\Omega(m)}$ of the normalized variant of the following state: 
\[
\ket{\psi_1}=\sum_{\vs\in \ZZ_q^n}\ket{\vs}\otimes \sum_{\ve\in \ZZ_q^m} \rho_{q/\sigma}(\ve)\ket{\ve}\otimes \ket{\vzero}.
\]
For the ease of exposition, we describe the rest of the algorithm as if it exactly generates the above state. 
\item Apply an appropriate phase operator to the first register to obtain the normalized variant of
\[
\ket{\psi_2}=\sum_{\vs\in \ZZ_q^n}\omega_q^{-\langle\vs,\vy\rangle}\ket{\vs}\otimes \sum_{\ve\in \ZZ_q^m} \rho_{q/\sigma}(\ve)\ket{\ve}\otimes \ket{\vzero}.
\]
\item Apply the unitary $U_\mA:\ket{\vs}\ket{\ve}\ket{0}\ra \ket{\vs}\ket{\ve}\ket{\vs^\trans \mA+\ve^\trans}$ to obtain the normalized variant of
\[
\ket{\psi_3}=\sum_{\vs\in \ZZ_q^n}\sum_{\ve\in \ZZ_q^m}\rho_{q/\sigma}(\ve)\omega_q^{-\langle\vs,\vy\rangle}\ket{\vs}\ket{\ve}\ket{\vs^\trans \mA+\ve^\trans \mod q}.
\]
\item \label{step:coherently_invert}
Apply the unitary $U_{inv}:\ket{\vs}\ket{\ve}\ket{\vy}\ra \ket{\vs-\Invert_1(\mA,\td_z,\vy)}\ket{\ve-\Invert_2(\mA,\td_z,\vy)}\ket{\vy}$
where $\Invert_1$ and $\Invert_2$ are the algorithms that outputs the first and second components of outputs of $\Invert$.
Below, we prove that the state at this point is within trace distance $\negl(m)$ of the normalized variant of the following state:
\[
\ket{\psi_4}=\sum_{\vs\in \ZZ_q^n}\sum_{\ve\in \ZZ_q^m}\rho_{q/\sigma}(\ve)\omega_q^{-\langle\vs,\vy\rangle}\ket{\vzero}\ket{\vzero}\ket{\vs^\trans \mA+\ve^\trans \mod q}.
\] 
\item  \label{step:output}
Apply the $q$-ary quantum Fourier transform to the rightmost register.

\item Output $\ket{z}\ket{\td_z}$ as $\rho$
and the state in the rightmost register as $\tau$
discarding the rest of registers. 
\end{enumerate}

Assuming that the state at the end of Step \ref{step:coherently_invert} is
within trace distance $\negl(n)$ of the normalized variant of $\ket{\psi_4}$, \Cref{lem:duality} immediately implies that for an at least $(1-2q^{-m})$-fraction of $\mA$,  the state after Step \ref{step:output}  is within trace distance $\negl(m)$ of 
\[
\ket{\vzero}\ket{\vzero}\otimes 
\displaystyle\sum_{\vx\in \ZZ_q^m\cap \Lambda_\mA^\vy}
\rho_\sigma(\vx)\ket{\vx}.
\]
This immediately implies 
\Cref{thm:QSampGauss}. Indeed, the state of the discarded registers is negligibly close to $\ket{\vzero}\ket{\vzero}$ independently of the LWE-trapdoor $\td_z$, and thus if we run the above procedure on a superposition $\ket{\psi}=\sum_{z\in \bit^L}\ket{z}\ket{\td_z}$,  
the above procedure only negligibly affects $\ket{\psi}$.   \takashi{Does this explanation make sense?}

In the rest of the proof, we prove that 
the state at the end of Step \ref{step:coherently_invert} is
within trace distance $\negl(n)$ of the normalized variant of $\ket{\psi_4}$.

Let $D_{\ZZ_q^m,q/\sigma}$ be the
probability density function on $\ZZ_q^m$ defined as follows: 
\[
D_{\ZZ_q^m,q/\sigma}(\vx)=\frac{\rho_{q/\sigma}(\vx)}{\rho_{q/\sigma}(\ZZ_q^m)}
\]
for $\vx\in \ZZ_q^m$.
Let 
\[
P_{bad}\seteq \Pr_{\ve\gets D_{\ZZ_q^m,q/\sigma}}[\|\ve\|_\infty> \beta_{err}].
\]
We define subnormalized states $\ket{\psi'_3}$ and $\ket{\psi'_4}$ as follows. 
\[
\ket{\psi'_3}=\sum_{\vs\in \ZZ_q^n}\sum_{\ve\in \ZZ_q^m:\|\ve\|_\infty\le \beta_{err}}\rho_{q/\sigma}(\ve)\omega_q^{-\langle\vs,\vy\rangle}\ket{\vs}\ket{\ve}\ket{\vs^\trans \mA+\ve^\trans},
\]
\[
\ket{\psi'_4}=\sum_{\vs\in \ZZ_q^n}\sum_{\ve\in \ZZ_q^m:\|\ve\|_\infty\le \beta_{err}}\rho_{q/\sigma}(\ve)\omega_q^{-\langle\vs,\vy\rangle}\ket{\vzero}\ket{\vzero}\ket{\vs^\trans \mA+\ve^\trans}.
\]

We show the following propositions. 
\begin{proposition}\label{prop:bound_bad_prob}
   It holds that $P_{bad}\le q^{-n}\cdot \negl(m)$.
\end{proposition}
\begin{proof}
    We have 
    \begin{align}
        P_{bad}
        &=\Pr_{\ve\gets D_{\ZZ_q^m,q/\sigma}}[\|\ve\|_\infty> \beta_{err}]\\
        &=\frac{\sum_{\ve\in\ZZ_q^m:\|\ve\|_\infty> \beta_{err}}\rho_{q/\sigma}(\ve)}{\sum_{\vz\in \ZZ_q^m}\rho_{q/\sigma}(\vz)}\\
        &=\frac{\sum_{\ve\in\ZZ_q^m:\|\ve\|_\infty> \beta_{err}}\rho_{q/\sigma}(\ve)}{\sum_{\vz\in \ZZ^m}\rho_{q/\sigma}(\vz)}
        \cdot 
        \frac{\sum_{\vz\in \ZZ^m}\rho_{q/\sigma}(\vz)}{\sum_{\vz\in \ZZ_q^m}\rho_{q/\sigma}(\vz)}\\
        &\le \frac{\rho_{q/\sigma}\left(\ZZ^m\setminus \mathcal{B}_\infty^m\left(\frac{q}{\sigma}\left(\frac{n\log q+m}{\pi}\right)^{1/2}\right)\right)}{\rho_{q/\sigma}(\ZZ^m)}\cdot \left(1-\frac{\rho_{q/\sigma}(\ZZ^m\setminus \mathcal{B}^m_{\infty}(q/2-1))}{\rho_{q/\sigma}(\ZZ^m)}\right)^{-1}\\ 
      &\le \left(2me^{-(n\log q+m)}\right)
      \cdot 
      \left(1-
      2m e^{-\pi \left(\frac{\sigma (q/2-1)}{q}\right)^2}
      \right)^{-1}\\
      &\le q^{-n}\cdot \negl(m). 
    \end{align}
    where 
 the first inequality follows from
 $\beta_{err}\ge \beta_{err}\ge \frac{q}{\sigma}\left(\frac{n\log q+m}{\pi}\right)^{1/2}$ and
$B^{m}_\infty(q/2-1)\subseteq \ZZ_q^m$,\footnote{Recall that an element of $\ZZ_q$ is represented as an element of $\ZZ \cap (-q/2,q/2]$} 
the second inequality follows from  \Cref{lem:tail_infty},
and the third inequality follows from 
$m2^{-m+1}=\negl(m)$ and 
 $\sigma>\sqrt{8m}$.
\end{proof}
\begin{proposition}\label{prop:difference_psi3}
It holds that 
$
\|\ket{\psi_3}-\ket{\psi'_3}\|^2 =\|\ket{\psi_3}\|^2\cdot P_{bad}.
$
\end{proposition}
\begin{proof}
We have 
\begin{align}
\|\ket{\psi_3}-\ket{\psi'_3}\|^2
&=
\left\|\sum_{\vs\in \ZZ_q^n}\sum_{\ve\in \ZZ_q^m:\|\ve\|_\infty> \beta_{err}}\rho_{q/\sigma}(\ve)\omega_q^{-\langle\vs,\vy\rangle}\ket{\vs}\ket{\ve}\ket{\vs^\trans \mA+\ve^\trans}\right\|^2\\
&=\sum_{\vs\in \ZZ_q^n}\sum_{\ve\in \ZZ_q^m:\|\ve\|_\infty> \beta_{err}}\left|
\rho_{q/\sigma}(\ve)\omega_q^{-\langle\vs,\vy\rangle}
\right|^2\\
&=
\sum_{\vs\in \ZZ_q^n}
\sum_{\ve\in \ZZ_q^m}\left|
\rho_{q/\sigma}(\ve)
\right|^2 \cdot
\frac{\sum_{\ve\in \ZZ_q^m:\|\ve\|_\infty> \beta_{err}}\left|
\rho_{q/\sigma}(\ve)
\right|^2}
{\sum_{\ve\in \ZZ_q^m}\left|
\rho_{q/\sigma}(\ve)
\right|^2}\\
&=\|\ket{\psi_3}\|^2\cdot P_{bad}
\end{align}
where we used $|\omega_q^{-\langle\vs,\vy\rangle}|=1$ in the third equality and the fourth equality immediately follows from the definitions of $\ket{\psi_3}$  and $P_{bad}$. 
\end{proof}
\begin{proposition}\label{prop:U_inv}
It holds that $U_{inv}\ket{\psi'_3}=\ket{\psi'_4}$.
\end{proposition}
\begin{proof}
Recall that if $\td$ is a $\beta_{err}$-LWE-trapdoor w.r.t. $\Invert$, then we have 
$\Invert(\mathbf{A},\td,\mathbf{s}^\trans \mA+\ve)=(\vs,\ve)$ by \Cref{def:LWE_trapdoor}. 
Then $U_{inv}\ket{\psi'_3}=\ket{\psi'_4}$ immediately follows from the definitions of 
$U_{inv}$, 
$\ket{\psi'_3}$, and $\ket{\psi'_4}$. 
\end{proof}
\begin{proposition}\label{prop:difference_psi4}
It holds that 
$
\|\ket{\psi_4}-\ket{\psi'_4}\|^2\le q^{n}\cdot \|\ket{\psi_3}\|^2\cdot P_{bad}.
$
\end{proposition}
\begin{proof}
\takashi{Tianwei's email: 
While going through the proof of Proposition 8.18 on page 41, I noticed
a potential point of confusion. Specifically, should the expression $\|
|\psi_4\rangle - \psi'_4 \|_2$ equal the following, considering that
$y^{\prime} \coloneqq s^\top A + e^\top$ is not necessarily identical to
$y$ in the exponent?
\[
\sum_{y^{\prime} \in \mathbb{Z}_q^m} q \sum_{s \in \mathbb{Z}_q^n, e \in
\mathbb{Z}_q^m | \|e\|_\infty > \beta_{\text{err}} \land s^\top A +
e^\top = {y^{\prime}}^\top}\rho_{q/\sigma}(e) \cdot \omega^{-\langle s, y
\rangle}
\]}
For $\ve\in \ZZ_q^m$, let
\[
\chi_{bad}(\ve)=
\begin{cases}
1 & (\|\ve\|_\infty>\beta_{err})\\
0 & (\|\ve\|_\infty\le \beta_{err})
\end{cases}.
\]
We have 
\begin{align}
\|\ket{\psi_4}-\ket{\psi'_4}\|^2
&=
\left\|\sum_{\vs\in \ZZ_q^n}\sum_{\ve\in \ZZ_q^m:\|\ve\|_\infty> \beta_{err}}\rho_{q/\sigma}(\ve)\omega_q^{-\langle\vs,\vy\rangle}\ket{\vzero}\ket{\vzero}\ket{\vs^\trans \mA+\ve^\trans}\right\|^2\\
&=\sum_{\vt\in \ZZ_q^m}\left|
\sum_{
\substack{\vs\in \ZZ_q^n,\ve\in \ZZ_q^m:\\
\|\ve\|_\infty> \beta_{err}~\land~\vs^\trans \mA+\ve^\trans=\vt^\trans
}
}
\rho_{q/\sigma}(\ve)\omega_q^{-\langle\vs,\vy\rangle}
\right|^2\\
&=\sum_{\vt\in \ZZ_q^m}\left|
\sum_{\vs\in \ZZ_q^n}
\chi_{bad}(\vt^\trans-\vs^\trans\mA)
\rho_{q/\sigma}(\vt^\trans-\vs^\trans\mA)\omega_q^{-\langle\vs,\vy\rangle}
\right|^2\\
&\le q^{n}\sum_{\vt\in \ZZ_q^m}\sum_{\vs\in \ZZ_q^n}\left|
\chi_{bad}(\vt^\trans-\vs^\trans\mA)
\rho_{q/\sigma}(\vt^\trans-\vs^\trans\mA)\omega_q^{-\langle\vs,\vy\rangle}
\right|^2\\
&= q^{n}\cdot \sum_{\vs\in \ZZ_q^n}\sum_{\ve\in \ZZ_q^m:\|\ve\|_\infty> \beta_{err}}\left|
\rho_{q/\sigma}(\ve)\omega_q^{-\langle\vs,\vy\rangle}
\right|^2\\
&=q^{n}\cdot \|\ket{\psi_3}\|^2\cdot P_{bad}
\end{align}
where 
the inequality in the fourth line follows from Cauchy-Schwarz inequality and the final equality is derived similarly to the proof of \Cref{prop:difference_psi3}.
\end{proof}

By \Cref{prop:U_inv}, we have
\begin{align} 
U_{inv}\ket{\psi_3}
&=U_{inv}\ket{\psi'_3}+U_{inv}(\ket{\psi_3}-\ket{\psi'_3})\\
&=\ket{\psi'_4}+U_{inv}(\ket{\psi_3}-\ket{\psi'_3})\\
&=\ket{\psi_4}+U_{inv}(\ket{\psi_3}-\ket{\psi'_3})+(\ket{\psi'_4}-\ket{\psi_4}).
\end{align}
Thus, by the triangle inequality and \Cref{prop:bound_bad_prob,prop:difference_psi3,prop:difference_psi4}, 
\begin{align} \label{eq:difference_state}
\|U_{inv}\ket{\psi_3}-\ket{\psi_4}\|\le \|\ket{\psi_3}-\ket{\psi'_3}\|+\|\ket{\psi'_4}-\ket{\psi_4}\|
\le \|\psi_3\|\cdot \negl(m). 
\end{align}
Thus, we have 
\begin{align}
\left\|\frac{U_{inv}\ket{\psi_3}}{\|\ket{\psi_3}\|}-\frac{\ket{\psi_4}}{\|\ket{\psi_4}\|}\right\|
&\le 
\left\|\frac{U_{inv}\ket{\psi_3}}{\|\ket{\psi_3}\|}-\frac{\ket{\psi_4}}{\|\ket{\psi_3}\|}\right\|
+
\left\|\frac{\ket{\psi_4}}{\|\ket{\psi_3}\|}-\frac{\ket{\psi_4}}{\|\ket{\psi_4}\|}\right\|\\
&\le \negl(m) 
\end{align}
where the second inequality follows from Eq. \ref{eq:difference_state}.\footnote{The first term is negligible by dividing both sides of Eq. \ref{eq:difference_state} by $\|\ket{\psi_3}\|$.
For the second term, observe 
$
\left\|\frac{\ket{\psi_4}}{\|\ket{\psi_3}\|}-\frac{\ket{\psi_4}}{\|\ket{\psi_4}\|}\right\|
=
\left|
\frac{\|\ket{\psi_4}\|}{\|\ket{\psi_3}\|}-1
\right|
$
and that Eq. \ref{eq:difference_state} implies
$\|\ket{\psi_4}\|=(1\pm \negl(m))\|\ket{\psi_3}\|$.
}
Note that for any normalized states $\ket{\psi}$ and $\ket{\psi'}$, we have $\|\ket{\psi}\bra{\psi}-\ket{\psi'}\bra{\psi'}\|_{tr}\le 4\|\ket{\psi}-\ket{\psi'}\|$ by \cite[Theorem 3.1]{BBBV97}.
Thus, the above inequality implies that the state $\frac{U_{inv}\ket{\psi_3}}{\|\ket{\psi_3}\|}$, which is obtained at the end of 
Step \ref{step:coherently_invert}, is
 within trace distance $\negl(m)$ from $\frac{\ket{\psi_4}}{\|\ket{\psi_4}\|}$.
This completes the proof of \Cref{thm:QSampGauss}. 
\end{proof}
\begin{remark}\label{rem:gap}
In the proof of \cite[Theorem 3.3]{TCC:AnaPorVai23}, they seem to claim, using our terminology, that the state at the end of Step \ref{step:coherently_invert} is
within trace distance $\negl(m)$ of the normalized variant of $\ket{\psi_4}$ only assuming $P_{bad}\le \negl(m)$. However, we do not see how to prove it. In our proof, we needed to show  $P_{bad}\le q^{-n}\cdot \negl(m)$ to deal with the loss of factor $q^{n}$ due to Cauchy-Schwarz inequality used in the proof of \Cref{prop:difference_psi4}. Fortunately, such a strong bound for $P_{bad}$ is easy to prove if we take sufficiently large $\beta_{err}$ since the tail bound of discrete Gaussian (\Cref{lem:tail_infty}) exponentially decreases. 
\end{remark}

\subsection{Construction}\label{sec:const_CS}  
We construct a coherently-signable CS scheme for the function class $\Fs$ that consists of functions expressed as circuits of 
description size $\ell=\ell(\secp)$ 
and input length $\iota=\iota(\secp)$ for bounded polynomials $\ell,\iota$. 
For $f\in \Fs$, 
we write $\descf$ to mean the description of $f$, which is a $\ell$-bit string. Without loss of generality, we assume that the first bit of $\descf$ is $1$.\footnote{This convention is made to meet the requirement of \Cref{lem:lattice_eval}.}  
In the description of the scheme, we regard $\descf\in \bit^\ell$ as a row vector in $\ZZ_q^\ell$ in a canonical manner. 
Let $U_\msg$ be a universal circuit for $\Fs$, i.e., for any $f\in \Fs$ and $\msg\in \bit^{k}$, 
we have $U_\msg(f)=f(\msg)$. 
Let $d=d(\secp)$ be the largest depth of $U_\msg$. 
Since the size of $U_\msg$ is $\poly(\secp)$, $d$ is also a polynomial. 
Whenever we consider LWE-trapdoors, they are w.r.t. the algorithm $\Invert$ of \Cref{lem:R_is_trapdoor}, and $\QSampGauss$ is the quantum algorithm in \Cref{thm:QSampGauss} w.r.t. $\Invert$. 
The choice of parameters
$n,m,q,\sigma,\beta_{sam},\beta_{ver},\beta_{SIS}$ are given below the description of the scheme. 
\begin{description}
\item[$\Setup(1^\secp)$:]~ 
\begin{itemize}
    \item Generate $(\mA,\td)\gets \TrapGen(1^n,1^m,q)$, $\mB\gets \ZZ_q^{n\times m \ell}$, and $\mC\gets \ZZ_q^{n\times m}$.  
    \item Output a verification key $\vk=(\mA,\mB,\mC)$ and a master secret key $\msk=(\vk,\td)$. 
\end{itemize}
\item[$\Constrain(\msk,f)$:]~
\begin{itemize}
\item Parse $\msk=(\vk=(\mA,\mB,\mC),\td)$. 
\item Generate $\mR \gets \SamPre(\td,\mB-\langle f\rangle\otimes \mG)$ where $\mG\in \ZZ_q^{n\times m}$ is the gadget matrix. 
\item Generate $\mR' \gets \SamPre(\td,\mC)$.
\item Output $\sigk_f=(\vk,f,\mR,\mR')$. 
\end{itemize}
\item[$\Sign(\sigk_f,\msg)$:]~
\begin{itemize}
\item Parse  $\sigk_f=(\vk=(\mA,\mB,\mC),f,\mR,\mR')$. 
\item Compute $\mH\gets \EvalF(U_{\msg},\mB)$
and  $\widehat{\mH}:=\EvalFX(U_{\msg},\langle f\rangle,\mB)$.
\item Generate $(\rho,\tau)\gets \QSampGauss([\mA\mid \mB\mH+\mC],\ket{\mR\widehat{\mH}+\mR'},\vzero,\sigma)$.\footnote{Here, we run $\QSampGauss$ with $L=0$ in which case $z$ is a null string.} 
\item Measure $\tau$ in the computational basis to obtain a vector $\vx\in \ZZ_q^{m}$.
\item Output a signature $\signature = \vx$.\footnote{In this subsection, we denote the signature by $\signature$ instead of $\sigma$ to avoid a notational collision.}
\end{itemize} 
\item[$\Vrfy(\vk,\msg,\signature)$:]~
\begin{itemize}
\item Parse $\vk=(\mA,\mB,\mC)$ and $\signature=\mathbf{x}$.
\item Compute $\mH\gets \EvalF(U_{\msg},\mB)$.
\item Output $\top$ if $[\mA\mid \mB\mH+\mC]\vx=\vzero$, $\|\vx\|_\infty\le \beta_{ver}$, and $\vx\neq \vzero$, 
and output $\bot$ otherwise.
\end{itemize} 
\end{description}
\begin{remark}\label{rem:CS_difference} \takashi{I added this remark.}
The above scheme is almost identical to the one obtained by applying the conversion of \cite{TCC:Tsabary17} to the context-hiding homomorphic signature scheme of \cite[Section 6]{STOC:GorVaiWic15}. 
One difference is that we introduce an additional matrix $\mC$ in $\vk$ and define a signature to be a solution to the SIS problem with respect to $[\mA\mid \mB\mH+\mC]$ instead of one to $[\mA\mid \mB\mH]$.  This additional matrix $\mC$ is introduced to make sure that $[\mA\mid \mB\mH+\mC]$ is uniformly distributed so that the algorithm $\QSampGauss$ of \Cref{thm:QSampGauss} works for $[\mA\mid \mB\mH+\mC]$ with an overwhelming probability. This modification almost does not affect the security proof and thus our security proof is essentially identical to that in \cite{STOC:GorVaiWic15}. 
\end{remark}

\paragraph{\bf Parameters.} 
We choose parameters $n,m,\beta_{sam}=\poly(\secp)$ and $q,\sigma,\beta_{ver},\beta_{SIS}=2^{\poly(\secp)}$ in such a way that all of the following conditions are satisfied for sufficiently large $\secp$.\footnote{$\beta_{SIS}$ does not appear in the scheme and only appears as a parameter for the underlying SIS problem.} 
\begin{enumerate}
\item Both items of \Cref{lem:trapdoor} hold for $n,m,q,\beta_{sam},$ and $k\in \{m, m\ell\}$.  \label{item:trapdoor}
\item $m\geq 2n\lceil \log q\rceil$ \label{item:m_lower}
\item 
$\sigma \ge (5m^3\ell\beta_{sam}(2m)^d + 5m )\cdot \left(\frac{n\log q+m}{\pi}\right)^{1/2}$ 
\label{item:sigma_lower}
\item $\sigma\in (\sqrt{8m},q/\sqrt{8m})$ 
\label{item:sigma_interval}
\item $\beta_{ver}= \omega(\log m) \sigma$ \label{item:beta_ver_lower}
\item $\beta_{SIS}\ge 4m^2 \ell \beta_{sam}(2m)^d\beta_{ver}$ \label{item:beta_SIS_lower}
\item The quantum hardness of $\SIS_{n, m, q, \beta_{SIS}}$ holds. \label{item:SIS}
\end{enumerate}
In the proofs of correctness, coherent-signability, and security, whenever we use one of the above conditions, we refer to the number of the corresponding condition.

To meet the above conditions, for example, we can choose the parameters as follows:
\begin{itemize}
\item We set $\beta_{SIS}\seteq 2^{d\lceil\log^2 \secp\rceil}$ and take a positive integer $n=\poly(\secp)$ and a prime $q=2^{\poly(\secp)}>\beta_{SIS}$ in such a way that for any $m=\poly(\secpar)$, the quantum hardness of $\SIS_{n, m, q, \beta_{SIS}}$ holds. (Recall that the SIS assumption in this paper ensures that we can take such $n$ and $q$. See the paragraph below \Cref{def:sis}.) 
Then Condition \ref{item:SIS} is satisfied. 
\item We take $m^*=O(n\log q)$ and $\beta_{sam}=O(n\sqrt{\log q})$ as in \Cref{lem:trapdoor} for the above chosen $n$ and $q$, and set $m=\max\{m^*,2n\lceil\log q \rceil\}$. Then Conditions  \ref{item:trapdoor} and \ref{item:m_lower} are satisfied. 
\item We set $\sigma\seteq   (5m^3\ell\beta_{sam}(2m)^d + 5m )\cdot \left(\frac{n\log q+m}{\pi}\right)^{1/2}$. 
Then Condition \ref{item:sigma_lower} is satisfied. Condition \ref{item:sigma_interval} is also satisfied for sufficiently large $\secp$ noting that 
$(5m^3\ell\beta_{sam}(2m)^d + 5m )\cdot \left(\frac{n\log q+m}{\pi}\right)^{1/2}=\poly(\secp)^d=2^{O(d\log \lambda)}$  and $q/\sqrt{8m}\ge \beta_{SIS}/\sqrt{8m}=2^{d\lceil\log^2 \secp\rceil}/\poly(\secp)=2^{\Omega(d \log^2 \secp)}$. 
\item We set $\beta_{ver}\seteq \lceil\log^2 m\rceil\sigma$. Then Condition \ref{item:beta_ver_lower} is satisfied. 
We can see that Condition \ref{item:beta_SIS_lower} is also satisfied for sufficiently large $\secp$ 
noting that $\beta_{SIS}\seteq 2^{d\lceil\log^2 \secp\rceil}$ and $4m^2 \ell \beta_{sam}(2m)^d\beta_{ver}=\poly(\secp)^d=2^{O(d\log \lambda)}$. 
\end{itemize}

\paragraph{\bf Correctness.} 
Suppose that $f$ and $\msg$ satisfy $f(\msg)=1$. 
Let  
$\sigk_f=(\vk=(\mA,\mB,\mC),f,\mR,\mR')$ be an honestly generated signing key for a function $f$.  
By \Cref{lem:trapdoor}, which is applicable by Condition \ref{item:trapdoor}, 
we have 
$\|\mR\|_\infty,\|\mR'\|_\infty\le \beta_{sam}$, 
\[\mA\mR=\mB-\descf \otimes \mG, 
\]
and 
\[
\mA\mR'=\mC. 
\]
By \Cref{lem:lattice_eval}, 
we have 
\[
[\mB-\descf \otimes G]\widehat{\mH}=\mB\mH-\mG
\]
and $\|\widehat{\mH}\|_\infty\le (2m)^d$ 
where $\mH$ and $\widehat{\mH}$ are as in the signing algorithm 
and we used $U_{\msg}(f)=f(\msg)=1$. 
Combining the above, we have 
\[
\mA(\mR\widehat{\mH}+\mR')+\mG=\mB\mH+\mC.
\]
Thus, by \Cref{lem:R_is_trapdoor} and $m\geq 2n\lceil \log q\rceil\geq  n\lceil \log q\rceil$ (Condition \ref{item:m_lower}), 
$\mR\widehat{\mH}+\mR'$ is a $\beta_{err}$-LWE-trapdoor of $[\mA \mid \mB\mH+\mC]$ where 
$\beta_{err}=q/(5m(m\|\mR \widehat{\mH}+\mR'\|_\infty+1))\ge q/(5m^3\ell\beta_{sam}(2m)^d + 5m ) \ge \frac{q}{\sigma}\left(\frac{n\log q+m}{\pi}\right)^{1/2}$ where the final inequality follows from Condition \ref{item:sigma_lower}.
We note that  $[\mA \mid \mB\mH+\mC]$ is uniformly distributed over $\ZZ_q^{n\times 2m}$. 
Thus, by \Cref{thm:QSampGauss}, which is applicable by $m\ge  2n\lceil \log q\rceil\ge n$ (Condition \ref{item:m_lower}) and $\sigma\in (\sqrt{8m},q/\sqrt{8m})$ (Condition \ref{item:sigma_interval}), 
for at least $(1-\negl(\secp))$-fraction of $\vk$,  
the second output $\tau$ of $\QSampGauss([\mA\mid \mB\mH+\mC],\ket{\mR\widehat{\mH}+\mR'},\vzero,\sigma)$ is a state negligibly close to the normalized variant of the following state: 
\[
\sum_{
\vx\in \ZZ_q^m\cap \Lambda_{[\mA\mid \mB\mH+\mC]}^\bot
}
\rho_{\sigma}(\vx)\ket{\vx}
\]
If we measure the above state, the outcome is distributed according to the distribution $D_{\ZZ_q^m\cap \Lambda_{[\mA\mid \mB\mH+\mC]}^\bot,\sigma}$. 
By \Cref{cor:distribution_close} and $q/\sigma>\sqrt{8m}$ (Condition \ref{item:sigma_interval}), the distribution is within statistical distance $\negl(m)$ of the distribution $D_{\Lambda_{[\mA\mid \mB\mH+\mC]}^\bot,\sigma}$. 
Thus, it suffices to prove that $\vx$ passes the verification with an overwhelming probability assuming that $\vx$ is sampled from $D_{\Lambda_{[\mA\mid \mB\mH+\mC]}^\bot,\sigma}$. 
We show this below. 

First, $\vx$ in the support of $D_{\Lambda_{[\mA\mid \mB\mH+\mC]}^\bot,\sigma}$ satisfies $[\mA\mid \mB\mH+\mC]\vx=\vzero$. 
Second, by \Cref{lem:tail_infty} and $\beta_{ver}= \omega(\log m) \sigma$ (Condition \ref{item:beta_ver_lower}), 
\begin{align}
\Pr_{\vx \gets D_{\Lambda_{[\mA\mid \mB\mH+\mC]}^\bot,\sigma}}[\|\vx\|_\infty\ge \beta_{ver}]
\le 2m e^{-\pi \left(\frac{\beta_{ver}}{\sigma}\right)^2}\le \negl(m).
\end{align}
Moreover, by \Cref{lem:prob_zero}, $m\ge 2n\log q$ (Condition \ref{item:m_lower}), and $\sigma=\omega(\sqrt{\log m})$ (Condition \ref{item:sigma_interval}),  
\begin{align}
\Pr_{\vx \gets D_{\Lambda_{[\mA\mid \mB\mH+\mC]}^\bot,\sigma}}[\vx=\vzero]
\le 2^{-\Omega(m)}. 
\end{align}
Thus, an honestly generated signature $\signature=\vx$ passes the verification except for a negligible probability. 

\paragraph{\bf Coherent-signability.}  \takashi{Changed the explanation according to the modification of the definition.}
Let 
$f_z\in \Fs$ for $z\in \bit^L$,  
$(\vk,\msk)\in \Setup(1^\secp)$,  
and 
$\sigk_{f_z}=(\vk,f_z,\mR_z,\mR'_z)\in \Constrain(\msk,f_z)$ for  $z\in \bit^L$. 
The coherent signing algorithm $\QSign$
takes $\sum_{z\in \bit^L}\alpha_z\ket{z}\ket{\sigk_{f_z}}$ and $\msg\in \Ms$ as input where $\sum_{z\in \bit^L}|\alpha_z|^2=1$   and $f_z(\msg)=1$  for all $z\in \bit^L$ and works as follows.
\begin{enumerate}
\item Measure $\vk=(\mA,\mB,\mC)$. Note that this does not collapse the state since all the branches share the same $\vk$.
Thus, the state remains to be the following:
\begin{align}
\sum_{z\in \bit^L}\alpha_z\ket{z}\ket{\vk,f_z,\mR_z,\mR'_z}. \label{eq:original}
\end{align}
\item Compute $\mH= \EvalF(U_\msg,\mB)$ and $\widehat{\mH}=\EvalFX(U_\msg,\descf,\mB)$.
\item \label{step:coherently_compute_td}
By coherently computing $\mR_z\widehat{\mH}+\mR'_z$ in a new register,  
generate 
\begin{align}
\ket{\psi}\seteq \sum_{z\in \bit^L}\alpha_z\ket{z}\ket{\vk,f_z,\mR_z,\mR'_z}\ket{\mR_z\widehat{\mH}+\mR'_z}. \label{eq:compute_RH}
\end{align}
\item \label{step:QSampGauss}
Run $(\rho,\tau)\gets \QSampGauss([\mA \mid \mB\mH+\mC],\ket{\psi},\vzero,\sigma)$ where the concatenation of the first two registers of $\ket{\psi}$ is regarded as the first register of $\ket{\psi}$ in \Cref{thm:QSampGauss}.  
\item Measure $\tau$ in the computational basis and let $\vx$ be the outcome. 
\item Uncompute the third register of $\rho$ (i.e., apply the inverse of Step \ref{step:coherently_compute_td}), discard the third register, and let $\rho'$ be the resulting state.
\item Output $\rho'$ and a signature $\mathsf{sig}=\vx$. 
\end{enumerate}

We show that the above describe algorithm $\QSign$ satisfies the requirements of coherent-signability as stated in  \Cref{def:coherent_signing}.
\Cref{item:coherent_one} of \Cref{def:coherent_signing} immediately follows from the definitions of $\QSign$ and $\Sign$ where we observe that when $\ket{\psi}$ is a computational basis state $\ket{z}{\sigk_{f_z}}$, $\QSign(\ket{z}{\sigk_{f_z}},\msg)$ works exactly in the same way as $\Sign(\sigk_{f_z},\msg)$. 
We show \Cref{item:coherent_two} of \Cref{def:coherent_signing} below.
By the proof of the correctness, for every $z\in \bit^L$, 
$\mR_z\widehat{\mH}+\mR'_z$ is a $\beta_{err}$-LWE-trapdoor of $[\mA\mid \mB\mH+\mC]$ for $\beta_{err}\ge \frac{q}{\sigma}\left(\frac{n\log q+m}{\pi}\right)^{1/2}$ 
(assuming Conditions \ref{item:trapdoor}, \ref{item:m_lower}, and \ref{item:sigma_lower}). 
Thus, 
by \Cref{thm:QSampGauss},
which is applicable by $m\ge  2n\lceil \log q\rceil\ge n$ (Condition \ref{item:m_lower}) and $\sigma\in (\sqrt{8m},q/\sqrt{8m})$ (Condition \ref{item:sigma_interval}), 
 for a $(1-\negl(\secp))$-fraction of $\vk$, 
the first output $\rho$ of $\QSampGauss$ obtained in Step \ref{step:QSampGauss} is negligibly close to $\ket{\psi}$.
Thus, $\rho$ is negligibly close to the original state $\sum_{z\in \bit^L}\alpha_z\ket{z}\ket{\sigk_{f_z}}$. 

\paragraph{Selective single-key security.}
We reduce selective single-key security to the quantum hardness of $\SIS_{n, m, q, \beta_{SIS}}$.  
This is very similar to the security proof of the leveled fully homomorphic signature scheme of \cite{STOC:GorVaiWic15}, but we prove it for completeness. 
For an adversary $\A$, we consider the following hybrids.
\begin{description} 
 \item[$\hybi{1}$:]
 This is the original experiment for selective single-key security. 
  \item[$\hybi{2}$:] 
  This hybrid is identical to the previous one except for that the way of generating 
  $\mB$, $\mC$ (which are part of $\vk$), and
  $\mR,\mR'$ (which are part of $\sigk_f$) as follows:   Upon receiving $f$ from $\A$, 
  the challenger
   generates $\mR\gets \Sam(1^m,1^{m\ell},q)$
   and $\mR'\gets \Sam(1^m,1^m,q)$, 
   and 
 sets $\mB:=\mA\mR+\descf \otimes \mG$ and $\mC:=\mA\mR'$. 
 By \Cref{lem:trapdoor}, which is applicable by Condition \ref{item:trapdoor}, 
 this hybrid is indistinguishable from the previous one. 
  \item[$\hybi{3}$:] 
  This hybrid is identical to the previous one except that the challenger generates 
  $\mA$ as 
  $\mA\gets \ZZ_q^{n\times m}$. Note that this modification can be made since the challenger is no longer using $\td$ in the previous hybrid. By \Cref{lem:trapdoor}, which is applicable by Condition \ref{item:trapdoor},  
 this hybrid is indistinguishable from the previous one. 
 \end{description}
Suppose that in $\hybi{3}$, 
$\A$ outputs 
$\msg$ and 
$\signature=\vx$ such that 
$f(\msg)=0$ and $\Vrfy(\vk,\msg,\signature)=\top$, i.e., 
$[\mA\mid \mB\mH+\mC]\vx=\vzero$, $\|\vx\|_\infty\le \beta_{ver}$, and $\vx\neq \vzero$. 
By \Cref{lem:lattice_eval},  and $m\geq 2n\lceil \log q\rceil\geq  n\lceil \log q\rceil$ (Condition \ref{item:m_lower}), 
we have 
\[
[\mB-\descf \otimes G]\widehat{\mH}=\mB\mH
\] 
where we used $U_{\msg}(f)=f(\msg)=0$. 
Since we set $\mB=\mA\mR+\descf \otimes \mG$ and $\mC=\mA\mR'$, the above implies 
\[
\mA(\mR\widehat{\mH}+\mR')=\mB\mH+\mC.
\]
Thus, we have 
\[
\mA [\mI\mid \mR\widehat{\mH}+\mR']\vx=\vzero.  
\]
By \Cref{lem:trapdoor}, which is applicable by Condition \ref{item:trapdoor},  
we have $\|\mR\|_\infty,\|\mR'\|_\infty\le \beta_{sam}$ and 
by \Cref{lem:lattice_eval} and $m\geq 2n\lceil \log q\rceil\geq  n\lceil \log q\rceil$ (Condition \ref{item:m_lower}), we have $\|\widehat{\mH}\|_\infty\le (2m)^d$. Thus,  
$\|[\mI\mid \mR\widehat{\mH}+\mR']\vx\|_{\infty}\le 4m^2 \ell \beta_{sam}(2m)^d\beta_{ver}\le \beta_{SIS}$ where the final inequality follows from Condition \ref{item:beta_SIS_lower}. 
By the assumed quantum hardness of $\SIS_{n, m, q, \beta_{SIS}}$ (Condition \ref{item:SIS}),  
this occurs only with a negligible probability. 
This means that above scheme is selectively single-key secure.

\fi 
\ifnum\anonymous=1
\else
\section*{Acknowledgement}
We thank Alexander Poremba for collaboration at the early stage of this work, Jiahui Liu for a useful discussion on  DS-SKL, and Tianwei Zhang for pointing out a notational issue. 
TM
is supported by JST CREST JPMJCR23I3, JST Moonshot JPMJMS2061-5-1-1, JST FOREST, MEXT QLEAP,
the Grant-in Aid for Transformative Research Areas (A) 21H05183, and the Grant-in-Aid for Scientific Research (A)
No.22H00522.
\fi

	\ifnum\llncs=1
\bibliographystyle{splncs04}
\bibliography{abbrev3,crypto,siamcomp_jacm,other}
	\else
\bibliographystyle{alpha} 
\bibliography{abbrev3,crypto,siamcomp_jacm,other}
	\fi

\ifnum\cameraready=0
	\ifnum\llncs=0
	\appendix
 \ifnum\llncs=0
\section{On Definition of TEPRF}\label{sec:TEPRF}
\else 
\section{Two-key Equivocal PRF}\label{sec:TEPRF} \fi
\ifnum\llncs=1
\subsection{Definition}
Our constructions of PRF-SKL and DS-SKL rely on a primitive called two-key equivocal PRF (TEPRF) introduced in  \cite{C:HJOSW16}.
\begin{definition}[Two-Key Equivocal PRF]\label{def:TEPRF} 
A two-key equivocal PRF (TEPRF) with input length $\ell$ (and output length $1$)\footnote{Though we can also define it for larger output length, we assume the output length to be $1$ by default similarly to \cite{C:HJOSW16}.} is a tuple of two algorithms $(\KG,\Eval)$.
\begin{description}
\item[$\KG(1^\secp,s^*) \ra (\key_0,\key_1)$:] The key generation algorithm is a PPT algorithm that takes as input the security parameter $1^\secp$ and a string $s^*\in \bit^\ell$, and outputs two keys $\key_0$ and $\key_1$. 
\item[$\Eval(\key,s) \ra b$:] The evaluation algorithm is a deterministic classical polynomial-time algorithm that takes as input a key $\key$ and an input $s\in \bit^\ell$, and outputs a bit $b\in \bit$. 
\end{description}
\begin{description}
    \item[Equality:]
    For all $\secp\in \mathbb{N}$, $s^*\in \bit^\ell$, $(\key_0,\key_1)\gets \KG(1^\secpar,s^*)$, $s\in \bit^\ell\setminus \{s^*\}$, 
    \[
    \Eval(\key_0,s)=\Eval(\key_1,s).
    \]
    \item[Different values on target:] 
     For all $\secp\in \mathbb{N}$, $s^*\in \bit^n$, $(\key_0,\key_1)\gets \KG(1^\secpar,s^*)$, 
    \[
    \Eval(\key_0,s^*)\ne \Eval(\key_1,s^*).
    \]
  \item[Differing point hiding.] 
For any (stateful) QPT adversary $\A$, 
        \begin{align}
    \abs{
    \begin{array}{l}
    \Pr\left[
    \A(\key_b)=1:
    \begin{array}{l}
    (s^*_0,s^*_1,b) \gets \A(1^\secpar)\\
    (\key_0,\key_1)\gets \KG(1^\secpar,s^*_0)\\
    \end{array}
    \right]\\
    -
   \Pr\left[
    \A(\key_b)=1:
    \begin{array}{l}
    (s^*_0,s^*_1,b) \gets \A(1^\secpar)\\
    (\key_0,\key_1)\gets \KG(1^\secpar,s^*_1)\\
    \end{array}
    \right]
    \end{array}
    }\le \negl(\secp).
    \end{align}
\end{description}
\end{definition}
Our definition of TEPRFs is slightly different from the original one in \cite{C:HJOSW16}. In \Cref{sec:relation_def}, we show that their definition implies ours. Since they show that OWFs imply TEPRFs under their definition, a similar implication holds for our definition of TEPRFs. 
\begin{theorem}\label{thm:OWF_to_TEPRF}
Assuming the existence of OWFs, there exist TEPRFs with input length $\ell$  for any polynomial $\ell=\ell(\secpar)$.
\end{theorem}
\fi 

\ifnum\llncs=1\subsection{Relation to the Original Definition}\label{sec:relation_def}\fi 
We show that TEPRFs as defined in \cite{C:HJOSW16} imply those as defined in \Cref{def:TEPRF}. 
For the sake of distinction, we call TEPRFs under their definition \emph{original TEPRFs}.

\begin{definition}[Original TEPRF~{\cite[Definition 8]{C:HJOSW16}}]\label{def:TEPRF_HJO}
An original TEPRF with input length $\ell$ (and output length $1$) is a tuple of three algorithms $(\ObvKG,\KG,\Eval)$.
\begin{description}
\item[$\ObvKG(1^\secp) \ra \key$:] The oblivious key generation algorithm is a PPT algorithm that takes as input the security parameter $1^\secp$, and outputs a key $\key$.
\item[$\KG(1^\secp,s^*) \ra (\key_0,\key_1)$:] The key generation algorithm is a PPT algorithm that takes as input the security parameter $1^\secp$ and a string $s^*\in \bit^\ell$, and outputs two keys $\key_0$ and $\key_1$. 
\item[$\Eval(\key,s) \ra b$:] The evaluation algorithm is a deterministic classical polynomial-time algorithm that takes as input a key $\key$ and an input $s\in \bit^\ell$, and outputs a bit $b\in \bit$. 
\end{description}
\begin{description}
    \item[Equality:]
    For all $\secp\in \mathbb{N}$, $s^*\in \bit^\ell$, $(\key_0,\key_1)\gets \KG(1^\secpar,s^*)$, $s\in \bit^\ell\setminus \{s^*\}$, 
    \[
    \Eval(\key_0,s)=\Eval(\key_1,s).
    \]
    \item[Different values on target:] 
     For all $\secp\in \mathbb{N}$, $s^*\in \bit^n$, $(\key_0,\key_1)\gets \KG(1^\secpar,s^*)$, 
    \[
    \Eval(\key_0,s^*)\ne \Eval(\key_1,s^*).
    \]
    \item[Indistinguishability:] 
    For any (stateful) QPT adversary $\A$, 
    \ifnum\llncs=0
    \begin{align}
    \abs{
    \Pr\left[
    \A(\key_b)=1:
    \begin{array}{l}
    (s^*,b) \gets \A(1^\secpar)\\
    (\key_0,\key_1)\gets \KG(1^\secpar,s^*)\\
    \end{array}
    \right]
    -
\Pr\left[
    \A(\key)=1:
    \begin{array}{l}
    (s^*,b) \gets \A(1^\secpar)\\
    \key\gets \ObvKG(1^\secpar)\\
    \end{array}
    \right]
    }\le \negl(\secp).
    \end{align}
    \else
    \begin{align}
    \abs{
        \begin{array}{l}
    \Pr\left[
    \A(\key_b)=1:
    \begin{array}{l}
    (s^*,b) \gets \A(1^\secpar)\\
    (\key_0,\key_1)\gets \KG(1^\secpar,s^*)\\
    \end{array}
    \right]\\
    -
\Pr\left[
    \A(\key)=1:
    \begin{array}{l}
    (s^*,b) \gets \A(1^\secpar)\\
    \key\gets \ObvKG(1^\secpar)\\
    \end{array}
    \right]
    \end{array}}\le \negl(\secp).
    \end{align}
\fi
\end{description}
\end{definition}
\begin{theorem}[{\cite[Claim 4]{C:HJOSW16}}]\label{thm:OWF_to_TEPRF_original}
Assuming the existence of OWFs, there exist original TEPRFs (as defined in \Cref{def:TEPRF_HJO}) with input length $\ell$  for any polynomial $\ell=\ell(\secpar)$.
\end{theorem}
\begin{theorem}\label{thm:TEPRF_original_to_TEPRF}
If original TEPRFs (as defined in \Cref{def:TEPRF_HJO}) with input length $\ell$  exist, then TEPRFs (as defined in \Cref{def:TEPRF}) with input length $\ell$  exist.
\end{theorem}
\begin{proof}
Let $(\ObvKG,\KG,\Eval)$ be an original TEPRF (that satisfies \Cref{def:TEPRF_HJO}). Then we show that $(\KG,\Eval)$ is a TEPRF (that satisfies \Cref{def:TEPRF}).  
The equality and different values on target properties immediately follow from those of the original TEPRF. Below, we show the different values on target property. 
    By indistinguishability of the original TEPRF, 
    for $\beta\in \bit$, we have
    \ifnum\llncs=0 
        \begin{align}
    \abs{
    \Pr\left[
    \A(\key_b)=1:
    \begin{array}{l}
    (s^*_0,s^*_1,b) \gets \A(1^\secpar)\\
    (\key_0,\key_1)\gets \KG(1^\secpar,s^*_\beta)\\
    \end{array}
    \right]
    -
   \Pr\left[
    \A(\key)=1:
    \begin{array}{l}
    (s^*_0,s^*_1,b) \gets \A(1^\secpar)\\
    \key\gets \ObvKG(1^\secpar)\\
    \end{array}
    \right]
    }\le \negl(\secp).
    \end{align}
    \else
            \begin{align}
    \abs{
    \begin{array}{l}
    \Pr\left[
    \A(\key_b)=1:
    \begin{array}{l}
    (s^*_0,s^*_1,b) \gets \A(1^\secpar)\\
    (\key_0,\key_1)\gets \KG(1^\secpar,s^*_\beta)\\
    \end{array}
    \right]\\
    -
   \Pr\left[
    \A(\key)=1:
    \begin{array}{l}
    (s^*_0,s^*_1,b) \gets \A(1^\secpar)\\
    \key\gets \ObvKG(1^\secpar)\\
    \end{array}
    \right]
    \end{array}
    }\le \negl(\secp).
    \end{align}
    \fi
Note that the second probability does not depend on $\beta$. 
Then differing point hiding follows from the above inequality and the triangle inequality. 
\end{proof}
\Cref{thm:OWF_to_TEPRF} follows from 
\Cref{thm:OWF_to_TEPRF_original,thm:TEPRF_original_to_TEPRF}.

\section{Comparisons among Security Definitions of PKE-SKL}\label{sec:comparison_def}
We introduce the definitions of IND-KLA and OW-KLA security by \cite{EC:AKNYY23} adopted to the classical revocation setting. 
\begin{definition}[IND-KLA Security~\cite{EC:AKNYY23}]\label{def:IND-CPA_PKESKL}
We say that a PKE-SKL scheme $\PKESKL$ with classical revocation for the message space $\Ms$ is IND-KLA secure, if it satisfies the following requirement, formalized by the experiment $\expb{\PKESKL,\qA}{ind}{kla}(1^\secp,\coin)$ between an adversary $\qA$ and the challenger $\qC$:
        \begin{enumerate}
            \item  $\qC$ runs $(\ek,\qdk,\dvk)\gets\qKG(1^\secp)$ and sends $\ek$ and $\qdk$ to $\qA$. 
            \item Throughout the experiment, $\qA$ can access the following (stateful) deletion verification oracle $\Oracle{\DelVrfy}$ where $V$ is initialized to be $\unreturned$:
            \begin{description}
               \item[ $\Oracle{\DelVrfy}(\cert)$:] It runs $d \gets \DelVrfy(\dvk,\cert)$ and returns $d$.  
               If $V=\bot$ and $d=\top$, it updates $V\seteq \top$. 
            \end{description}
            \item $\qA$ sends $(\msg_0^*,\msg_1^*)\in \Ms^2$
            to $\qC$. If $V=\unreturned$, $\qC$ output $0$ as the final output of this experiment. Otherwise, $\qC$ generates $\ct^*\la\Enc(\ek,\msg_\coin^*)$ and sends $\ct^*$ to $\qA$.
            \item $\qA$ outputs a guess $\coin^\prime$ for $\coin$. $\qC$ outputs $\coin'$ as the final output of the experiment.
        \end{enumerate}
        For any QPT $\qA$, it holds that
\begin{align}
\advb{\PKESKL,\qA}{ind}{kla}(\secp) \seteq \abs{\Pr[\expb{\PKESKL,\qA}{ind}{kla} (1^\secp,0) \ra 1] - \Pr[\expb{\PKESKL,\qA}{ind}{kla} (1^\secp,1) \ra 1] }\leq \negl(\secp).
\end{align}  
\end{definition}

\begin{definition}[OW-KLA Security~\cite{EC:AKNYY23}]\label{def:OW-CPA_PKESKL}
We say that a PKE-SKL scheme $\PKESKL$ with classical revocation for the message space $\Ms$ is OW-KLA secure, if it satisfies the following requirement, formalized by the experiment $\expb{\PKESKL,\qA}{ow}{kla}(1^\secp)$ between an adversary $\qA$ and the challenger $\qC$:
        \begin{enumerate}
            \item  $\qC$ 
            runs $(\ek,\qdk,\dvk)\gets\qKG(1^\secp)$ and sends $\ek$ and $\qdk$ to $\qA$. 
              \item Throughout the experiment, $\qA$ can access the following (stateful) deletion verification oracle $\Oracle{\DelVrfy}$ where $V$ is initialized to be $\unreturned$:
              \begin{description}
               \item[   $\Oracle{\DelVrfy}(\cert)$:] It runs $d \gets \DelVrfy(\dvk,\cert)$ and returns $d$. 
               If $V=\bot$ and $d=\top$, it updates $V\seteq \top$. 
            \end{description}
            \item 
            $\qA$ sends $\requestchallenge$
            to $\qC$. If $V=\unreturned$, $\qC$ outputs $0$ as the final output of this experiment. Otherwise, $\qC$ 
            chooses $\msg^*\gets \Ms$, 
            generates $\ct^*\la\Enc(\ek,\msg^*)$ and sends $\ct^*$ to $\qA$.
            \item $\qA$ outputs $\msg$. $\qC$ outputs $1$ if $\msg=\msg^*$ and otherwise outputs $0$ as the final output of the experiment.
        \end{enumerate}
        For any QPT $\qA$, it holds that
\begin{align}
\advb{\PKESKL,\qA}{ow}{kla}(\secp) \seteq \Pr[\expb{\PKESKL,\qA}{ow}{kla} (1^\secp) \ra 1]\leq \negl(\secp).
\end{align}
\end{definition}

\begin{lemma}\label{lem:IND-VRA_to_IND-KLA}
    If a PKE-SKL scheme with classical revocation for $\Ms$ such that $|\Ms|=\poly(\secp)$ satisfies IND-VRA security, then it also satisfies IND-KLA security.  
\end{lemma}
\begin{lemma}\label{lem:OW-VRA_to_OW-KLA}
    If a PKE-SKL scheme with classical revocation satisfies OW-VRA security, then it also satisfies OW-KLA security.  
\end{lemma}
Since the proofs of \Cref{lem:IND-VRA_to_IND-KLA} amd \Cref{lem:OW-VRA_to_OW-KLA} are almost identical, we only prove \Cref{lem:IND-VRA_to_IND-KLA} after which we briefly explain the difference for the proof of \Cref{lem:OW-VRA_to_OW-KLA}.
\begin{proof}[Proof of \Cref{lem:IND-VRA_to_IND-KLA}]
Let $(\qKG, \Enc, \qDec,\qDel,\DelVrfy)$ be an IND-VRA secure PKE-SKL scheme with classical revocation for message space $\Ms$ such that $|\Ms|=\poly(\secp)$. Toward construction, we assume that there is a QPT adversary $\qA$ that breaks IND-KLA security. 
Without loss of generality, we assume that $\qA$ makes exactly $q$ queries to the deletion verification oracle for some polynomial $q=q(\secp)$. 
Then we construct a QPT adversary $\qB$ that breaks IND-VRA security as follows. 
\begin{description}
\item[$\qB$:]  
Upon receiving $\ek$ and $\qdk$ from the challenger, it works as follows. 
Choose $k\gets [q]$ and $\msg_0,\msg_1\gets \Ms$.  
Run $\qA$ on input $(\ek,\qdk)$ until it makes $k$-th deletion verification query where $\qB$ simulates the deletion verification oracle by returning $\bot$ for the first $k-1$ queries and returning $\top$ for the $k$-th query.  
For $i\in [k]$, let $\cert_i$ be the $i$-th deletion verification query by $\qA$. 
Send $\cert_k$ and $(\msg_0,\msg_1)$ to the challenger and   
receive $\dvk$ and $\ct^*$ from the challenger.  
If $\DelVrfy(\dvk,\cert_i)=\top$ for some  $i\in [k-1]$ or  $\DelVrfy(\dvk,\cert_k)=\bot$, output a random bit and immediately halt. Otherwise, run $\qA$ until it sends $(\msg_0^*,\msg^*_1)$ to the challenger where $\qB$ simulates the deletion verification oracle by using $\dvk$. 
If $(\msg_0^*,\msg^*_1)\neq (\msg_0,\msg_1)$,  output a random bit and immediately halt.  Otherwise, send $\ct^*$ to $\qA$ and run $\qA$ until it outputs the final output $\coin'$ where $\qB$ simulates the deletion verification oracle by using $\dvk$. 
Output $\coin'$. 
\end{description}
By the description of $\qB$, one can see that $\qB$ perfectly simulates the IND-KLA experiment for $\qA$ conditioned on that $\qB$ does not abort, which occurs with probability $\frac{1}{q|\Ms|^2}$. Thus, we have 
\[
\advb{\PKESKL,\qB}{ind}{vra}(\secp)
=
\frac{1}{q|\Ms|^2}\advb{\PKESKL,\qA}{ind}{kla}(\secp). 
\]
Since $q=\poly(\secp)$ and we assume $|\Ms|=\poly(\secp)$, 
if $\qA$ breaks IND-KLA security, then $\qB$ breaks IND-VRA security. This is contradiction. Thus, the scheme is IND-KLA secure.  
\end{proof}
\begin{proof}[Proof of \Cref{lem:OW-VRA_to_OW-KLA} (sketch)]
The proof can be done almost in the same manner as the proof of \Cref{lem:IND-VRA_to_IND-KLA}. One difference is that in the OW-KLA experiment, the challenge message is chosen by the challenger instead of the adversary. Thus, we do not need to guess the challenge messages $(\msg_0^*,\msg_1^*)$ unlike the proof of \Cref{lem:IND-VRA_to_IND-KLA}. This is why the proof works for schemes with an  arbitrarily large message space.
\end{proof}
\begin{remark}
The above proof implicitly shows that IND-VRA security implies a stronger variant of  IND-KLA security where the adversary is also given the deletion verification key after passing the deletion verification for the first time when the message space is polynomial-size. We can show a similar implication for the case of one-wayness without the restriction to a polynomial-size message space.
\end{remark}

 \section{Transformation from UPF-SKL to PRF-SKL}\label{sec:UPF-to-PRF}
In this section, we show a generic conversion from UPF-SKL to PRF-SKL using the quantum Goldreich-Levin lemma. Though it is very similar to the OW-to-IND conversion for PKE-SKL~\cite[Lemma 3.12]{EC:AKNYY23}, we provide a full proof for completeness.  
\subsection{Preparation}
We need the quantum Goldreich-Levin lemma established by \cite{C:CLLZ21} based on \cite{AC02}.
\begin{lemma}[Quantum Goldreich-Levin with Quantum Auxiliary Input~{\cite[Lemma~B.12]{C:CLLZ21}}]\label{lem:QGL}
There exists a QPT algorithm $\qExt$ that satisfies the following.
Let $n \in \mathbb{N}$, $x\in \bit^n$, $\epsilon\in [0,1/2]$,  and $\qA$ be a quantum algorithm with a quantum auxiliary input $\qaux$ such that    
\begin{align}
    \Pr\left[\qA(\qaux,r)\ra x\cdot r ~\middle |~ r\gets \bit^n\right]
    \geq \frac{1}{2}+\epsilon.
\end{align}
Then, we have
\begin{align}
    \Pr\left[\qExt([\qA],\qaux)\rightarrow x \right]\geq 4\epsilon^2.
\end{align}
where $[\qA]$ means the description of $\qA$. 
\end{lemma}

\subsection{Transformation}
\begin{lemma}[Restatement of \Cref{lem:upf-prf}]\label{lem:upf-prf_app}
If there exists a UP-VRA secure UPF-SKL scheme with classical revocation, then there exists a PR-VRA secure PRF-SKL scheme with classical revocation.
\end{lemma}
\begin{proof}
$\UPFSKL=(\UPFSKL.\qKG,\UPFSKL.\Eval, \UPFSKL.\qLEval,\allowbreak\UPFSKL.\qDel,\UPFSKL.\DelVrfy)$ be a UP-VRA secure UPF-SKL scheme with classical revocation whose input space is $\bit^n$ and output space is $\bit^\ell$. Then, we construct an IND-VRA secure PKE-SKL scheme with classical revocation $\PRFSKL=(\PRFSKL.\qKG,\PRFSKL.\Eval, \allowbreak \PRFSKL.\qLEval,\allowbreak\PRFSKL.\qDel,\PRFSKL.\DelVrfy)$  whose input space is $\bit^{n+\ell}$ and output space is $\bit$ as follows. 
\begin{description}
\item[$\PRFSKL.\qKG(1^\secp)\ra(\msk,\qsk,\dvk)$:] On input the security parameter $1^\lambda$, 
run $(\msk,\qsk,\dvk)\gets \UPFSKL.\qKG(1^\secp)$ and output $(\msk,\qsk,\dvk)$.

\item[$\PRFSKL.\Eval(\msk,s)\ra t$:]On input a master secret key $\msk$ and an input $s \in \bit^{n+\ell}$, 
parse $s=s^\prime \concat r$ where $s^\prime\in \bit^n$ and $r\in \bit^\ell$,  
generate $t^\prime \gets \UPFSKL.\Eval(\msk,s^\prime)$,  
and output $t\seteq t^\prime \cdot r$. 

\item[$\PRFSKL.\qLEval(\qsk,s)\ra t$:] On input a secret key $\qsk$ and an input $s \in \bit^{n+\ell}$, 
parse $s=s^\prime \concat r$ where $s^\prime\in \bit^n$ and $r\in \bit^\ell$,  
generate $t^\prime \gets \UPFSKL.\qLEval(\qsk,s^\prime)$,  
and output $t\seteq t^\prime \cdot r$. 

\item[$\PRFSKL.\qDel(\qsk)$:] On input a secret key $\qsk$, output $\cert\gets\UPFSKL.\qDel(\qsk)$.

\item[$\PRFSKL.\DelVrfy(\dvk,\cert)\ra\top/\bot$:] On input a verification key $\dvk$ and a  certificate $\cert$, run $\UPFSKL.\DelVrfy(\dvk,\cert)$ and
output whatever $\UPFSKL.\DelVrfy$ outputs.  
\end{description}
The evaluation correctness and verification correctness of $\PRFSKL$ immediately follow from those of $\UPFSKL$.  
In the following, we prove that $\PRFSKL$ is PR-VRA secure assuming that $\UPFSKL$ is UP-VRA secure. 
Toward contradiction, suppose that $\PRFSKL$ is not PR-VRA secure.  Let $\expb{\PRFSKL,\qA}{pr}{vra'} (1^\secp,\coin)$ be an experiment that works similarly to $\expb{\PRFSKL,\qA}{pr}{vra} (1^\secp,\coin)$ (defined in \Cref{def:PR-VRA}) except that we set $t^*_1\seteq t^*_0 \oplus 1$ instead of choosing $t^*_1$ uniformly from the output space. (Recall that $\PRFSKL$ has 1-bit outputs.) 
Then the indistinguishability between $\expb{\PRFSKL,\qA}{pr}{vra} (1^\secp,0)$ and $\expb{\PRFSKL,\qA}{pr}{vra} (1^\secp,1)$ is equivalent to that between $\expb{\PRFSKL,\qA}{pr}{vra'} (1^\secp,0)$ and $\expb{\PRFSKL,\qA}{pr}{vra'} (1^\secp,1)$. Thus,  by our assumption that $\PRFSKL$ is not PR-VRA secure, 
there is a QPT adversary $\qA$ such that 
\begin{align} \label{eq:IND_A_wins_prf}
\Pr_{\coin\gets \bit}[\expb{\PRFSKL,\qA}{pr}{vra'} (1^\secp,\coin) \ra \coin]\ge 1/2+\epsilon(\secp)
\end{align}
for a non-negligible $\epsilon(\secp)$. 
We divide $\qA$ into the following two stages $\qA_0$ and $\qA_1$: 
\begin{description}
\item[$\qA_0(\qsk)\rightarrow (\cert,\qst_\qA)$:] Upon receiving $\qsk$ from the challenger, output a classical certificate $\cert$ and a quantum state $\qst_\qA$. 
\item[$\qA_1(\qst_\qA,\dvk,s^*,t^*)\rightarrow \coin'$:] Upon receiving the state $\qst_\A$ from $\qA_0$, and $(\dvk,s^*,t^*)$ from the challenger, output $\coin'$.  
\end{description} 

We have 
\begin{align}
&\Pr_{\coin\gets \bit}[\expb{\PRFSKL,\qA}{pr}{vra'} (1^\secp,\coin) \ra \coin]\\
&=\Pr_{\coin\gets \bit}[\expb{\PRFSKL,\qA}{pr}{vra'} (1^\secp,\coin) \ra \coin \wedge \DelVrfy(\dvk,\cert)=\top]\\
&+\Pr_{\coin\gets \bit}[\expb{\PRFSKL,\qA}{pr}{vra'} (1^\secp,\coin) \ra \coin \wedge \DelVrfy(\dvk,\cert)=\bot]\\
&=
\Pr[\DelVrfy(\dvk,\cert)=\returned]
\cdot \Pr_{\coin\gets \bit}[\expb{\PRFSKL,\qA}{pr}{vra'} (1^\secp,\coin) \ra \coin \mid \DelVrfy(\dvk,\cert)=\returned]\\
&+\frac{1}{2}(1-\Pr[\DelVrfy(\dvk,\cert)=\returned]). \label{eq:advantage_A_decompose_prf}
\end{align}
By \cref{eq:IND_A_wins_prf,eq:advantage_A_decompose_prf}, 
we have 
\begin{align}\label{eq:conditioned_on_returned_prob_prf}
 \Pr_{\coin\gets \bit}[\expb{\PRFSKL,\qA}{pr}{vra'} (1^\secp,\coin) \ra \coin \mid \DelVrfy(\dvk,\cert)=\returned]\ge \frac{1}{2}+\frac{\epsilon(\secp)}{\Pr[\DelVrfy(\dvk,\cert)=\returned]}.
\end{align}

Then, we construct an adversary $\qB=(\qB_0,\qB_1)$ against UP-VRA security of $\UPFSKL$ that works as follows. 
\begin{description}
\item[$\qB_0(\qdk)\rightarrow (\cert,\qst_\qA)$:]  
This is identical to $\qA_0$. Specifically, run $(\cert,\qst_\qA)\gets \qA_0(\qsk)$ and output $\cert$ and $\qst_\qA$.   
\item[$\qB_1(\qst_\qA,\dvk,s^\prime )\rightarrow t^\prime$:] Upon receiving $\qst_\qA$ from $\qB_0$  and $(\dvk,s^\prime)$ from the challenger, set $\qaux\seteq (\qst_\qA,\dvk,s^\prime)$ and define an algorithm $\qA'$ as follows.
\begin{description}
\item[$\qA'(\qaux,r)$:]
On input $\qaux=(\qst_\qA,\dvk,s^\prime)$ and $r\in \bit^\ell$,  
choose $t^*\gets \bit$ and $r \gets \bit^\ell$,  
set $s^*=s^\prime \concat r$, 
run $\coin' \gets \qA_1(\qst_\qA,\dvk,s^*,t^*)$, and output $\coin'\oplus t^*$.  
\end{description}
Run $t^\prime \gets \qExt([\qA'],\qaux)$, and output $t^\prime$ where  
$\qExt$ is the algorithm as in \Cref{lem:QGL} and
$[\qA']$ is the description of $\qA'$.
\end{description}

In the following, we show that $\qB$ breaks UP-VRA security of $\UPFSKL$. 
Let $\qG$ be an algorithm that works as follows.
\begin{description}
\item[$\qG(1^\secp)$:] 
Generate $(\msk,\qsk,\dvk)\gets\UPFSKL.\qKG(1^\secp)$, 
$(\cert,\qst_\qA)\gets \A_0(\qsk)$,    
$s^\prime \gets \bit^\ell$, and 
$t^\prime \seteq \UPFSKL(\UPFSKL.\msk,s^\prime)$. 
Let $V\seteq \returned$ if $\DelVrfy(\dvk,\cert)=\returned$ and $V\seteq \unreturned$ otherwise.    
Output $(V,\qst_\qA,\dvk,s^\prime,t^\prime)$. 
\end{description}

By \cref{eq:conditioned_on_returned_prob_prf} and a standard averaging argument,  
with probability at least $\frac{\epsilon(\secp)}{2\Pr[V=\returned]}$ over the choice of  $(V,\qst_\qA,\dvk,s^\prime,t^\prime)\gets \qG(1^\secp)$  conditioned on $V=\returned$, we have 
\begin{align} 
    \Pr\left[
    \A_1(\qst_\qA,\dvk,s^*,t^*)\rightarrow \coin
    \right]\ge \frac{1}{2}+\frac{\epsilon(\secp)}{2\Pr[V=\returned]}\ge \frac{1}{2}+\frac{\epsilon(\secp)}{2}
\end{align}
where  
    $r\gets \bit^\ell$, 
    $s^*\seteq s^\prime\concat r$, 
    $\coin\gets \bit$,  
    and 
    $t^*\seteq t^\prime\cdot r\oplus \coin$.
    
Therefore, with probability at least $\frac{\epsilon(\secp)}{2}$ over the choice of  $(V,\qst_\qA,\dvk,s^\prime,t^\prime)\gets \qG(1^\secp)$, we have 
\begin{align}\label{eq:IND_A_wins_rewrite_again} 
    \Pr\left[
    \A_1(\qst_\qA,\dvk,s^*,t^*)\rightarrow \coin
    \right]\ge \frac{1}{2}+\frac{\epsilon(\secp)}{2}
\end{align}
where  
    $r\gets \bit^\ell$, 
    $s^*\seteq s^\prime\concat r$,   
   $\coin\gets \bit$, 
    and 
    $t^*\seteq t^\prime\cdot r\oplus \coin$.

For such $(V,\qst_\qA,\dvk,s^\prime,t^\prime)$, if we let $\qaux=(\qst_\qA,\dvk,s^\prime)$, \cref{eq:IND_A_wins_rewrite_again} directly implies  
\begin{align}
\Pr_{r\gets \bit^n}\left[\qA'(\qaux,r)\ra t^\prime\cdot r\right]\geq \frac{1}{2}+\frac{\epsilon(\secp)}{2}.
\end{align}
Therefore, by \Cref{lem:QGL}, we have 
\begin{align} \label{eq:ext_succeed}
    \Pr\left[\qExt([\qA'],\qaux)\rightarrow t^\prime \right]\geq \epsilon(\secp)^2.
\end{align}
Since \cref{eq:ext_succeed} \emph{and} $V=\returned$ hold at the same time with probability at least $\frac{\epsilon(\secp)}{2}$ over the choice of $(V,\qst_\qA,\dvk,s^\prime,t^\prime)$,  we have  
\begin{align}
    \Pr_{(V,\qst_\qA,\dvk,s^\prime,t^\prime)\gets \qG(1^\secp) }\left[
    V=\returned
    \wedge
    \qB_1(\qst_\qA,\dvk,s^\prime)\rightarrow t^\prime 
    \right]\ge \frac{\epsilon(\secp)^3}{2}.
\end{align} 
By the definitions of $\qB=(\qB_0,\qB_1)$ and $\qG$ and the assumption that $\epsilon(\secp)$ is non-negligible, this implies that $\qB$ breaks UP-VRA security of $\UPFSKL$. 
\end{proof}

\begin{remark}
The above conversion does not work as a compiler from a plain UPF to a plain PRF. For such a compiler, we need to take a seed for the Goldreich-Levin predicate from a key rather than input~\cite{C:NaoRei98}. 
However, this is not a problem for our purpose since we can easily add security as a plain PRF on top of any PR-VRA secure PRF-SKL as noted in \Cref{rem:security_plain_PRF}. 
\end{remark}
 \section{Proof of \Cref{lem:R_is_trapdoor}}\label{sec:proof_MR}
We rely on the following lemma shown by \cite{EC:MicPei12}. 
\begin{lemma}[{\cite[Theorem 4.1]{EC:MicPei12}}]\label{lem:G_invert}
    For $n,q,k\in \NN$ such that 
$q\ge 2$ and
$k\geq n\lceil \log q\rceil$, 
let $\mG\in \ZZ_q^{n\times k}$ be the gadget matrix. 
There is a basis $\mS\in \ZZ_q^{k\times k}$ of $\Lambda_\mG^\bot$ such that $\|\mS\|_\infty\le 2$.\footnote{The original statement of \cite[Theorem 4.1]{EC:MicPei12} does not give a bound for the $\ell_\infty$ norm of $\mS$, but this can be easily seen from its proof.} Moreover, there is a classical deterministic algorithm $\Invert_\mG$ that takes $\vs^\trans \mG+\ve^\trans$ and outputs $(\vs,\ve)$ for any $\vs\in \ZZ_q^n$ and $\ve\in \ZZ_q^k$ such that $\|\mS^\trans \ve\|_\infty< q/2$. 
The running time of $\Invert_\mG$ is $\poly(n,k,\log q)$.
\end{lemma}
Then we prove \Cref{lem:R_is_trapdoor}
\begin{proof}[Proof of \Cref{lem:R_is_trapdoor}]
The algorithm $\Invert$ works as follows. On input $\left[\mathbf{A}\mid \mathbf{A}\mathbf{R}+\mathbf{G}\right]\in \ZZ_q^{n\times (m+k)}$, $\mR\in \ZZ_q^{m\times k}$, and $\vy \in \ZZ_q^{m+k}$, 
run $(\vs,\ve')\gets \Invert_\mG(\vy_2^\trans-\vy_1^\trans \mR)$
where $\vy_1\in \ZZ_q^m$ and $\vy_2\in \ZZ_q^k$ are vectors such that $\vy^\trans=(\vy_1^\trans\mid \vy_2^\trans)$. 
Then it sets $\ve\seteq 
\vy-
\vs^\trans \left[\mathbf{A}\mid \mathbf{A}\mathbf{R}+\mathbf{G}\right] \mod q$ and outputs $(\vs,\ve)$. 
Below, we show that the above algorithm works correctly, 
that is, 
if $\vy^\trans=\vs^\trans \left[\mathbf{A}\mid \mathbf{A}\mathbf{R}+\mathbf{G}\right]+\ve^\trans \mod q$ where $\vs\in \ZZ_q^n$ and $\ve\in \ZZ_q^{m+k}$ such that $\|\ve\|_\infty\le q/(5k(m\|\mR\|_\infty+1))$, then it outputs $(\vs,\ve)$. 
To show this, we observe that $\vy_2^\trans-\vy_1^\trans \mR=\vs^\trans \mG+ \ve_2^\trans-\ve_1^\trans \mR$ where $\ve_1\in \ZZ_q^m$ and $\ve_2\in \ZZ_q^k$ are vectors such that $\ve^\trans=(\ve_1^\trans\mid \ve_2^\trans)$. 
Let $\ve'\seteq \ve_2^\trans-\ve_1^\trans \mR$ and $\mS\in \ZZ_q^{k\times k}$ be the matrix as in \Cref{lem:G_invert}. 
Then we have $\|\mS^\trans \ve'\|_\infty\le 
k\|\mS\|_{\infty}(\|\ve\|_\infty+m\|\ve\|_\infty\|\mR\|_\infty)\le 2k\|\ve\|_\infty(m\|\mR\|_\infty+1)\le 2q/5<q/2$. Thus, by \Cref{lem:G_invert},  $\Invert_\mG(\vy_2^\trans-\vy_1^\trans \mR)$ outputs $(\vs,\ve')$. Thus, $\Invert(\left[\mathbf{A}\mid \mathbf{A}\mathbf{R}+\mathbf{G}\right], \mR,\vy)$ outputs $(\vs,\ve)$. 
\end{proof}

\else
	\newpage
	 	\appendix
	 	\setcounter{page}{1}
 	{
	\noindent
 	\begin{center}
	{\Large SUPPLEMENTAL MATERIALS}
	\end{center}
 	}
	\setcounter{tocdepth}{2}

  \section{Definitions of PRF-SKL and DS-SKL}\label{sec:definitions_PRF-SKL_and_DS-SKL}
\subsection{Pseudorandom and Unpredictable Functions with Secure Key Leasing}\label{def:PRF-SKL}
In this subsection, we define pseudorandom functions with secure key leasing (PRF-SKL) with classical revocation. 
The syntax is similar to that of key-revocable PRFs in \cite{TCC:AnaPorVai23}   
except that we introduce a deletion algorithm that generates a classical certificate of deletion of a key.\footnote{\cite{TCC:AnaPorVai23} also mentions a classical revocation variant, but they do not provide its formal definition.}

\begin{definition}[PRF-SKL with classical revocation]
A PRF-SKL scheme  $\PRFSKL$ with classical revocation for the domain $\Domprf$ and the range $\Ranprf$ is a tuple of five algorithms $(\qKG, \Eval, \qLEval,\qDel,\DelVrfy)$. 
\begin{description}
\item[$\qKG(1^\secp)\ra(\msk,\qsk,\dvk)$:] The key generation algorithm is a QPT algorithm that takes a security parameter $1^\lambda$, and outputs a master secret key $\msk$, a secret key $\qsk$, and a deletion verification key $\dvk$.

\item[$\Eval(\msk,\prfinp)\ra \prfout$:] The evaluation algorithm is a deterministic classical polynomial-time algorithm that takes a master secret key $\msk$ and an input $\prfinp \in \Domprf$, and outputs a value $\prfout$.

\item[$\qLEval(\qsk,\prfinp)\ra \prfout$:] The leased evaluation algorithm is a QPT algorithm that takes a secret key $\qsk$ and an input $\prfinp \in \Domprf$, and outputs a value $\prfout$.

\item[$\qDel(\qsk)\ra\cert$:] The deletion algorithm is a QPT algorithm that takes a secret key $\qsk$, and outputs a classical string $\cert$.

\item[$\DelVrfy(\dvk,\cert)\ra\top/\bot$:] The deletion verification algorithm is a deterministic classical polynomial-time algorithm that takes a deletion verification key $\dvk$ and a deletion certificate $\cert$, and outputs $\top$ or $\bot$.

\item[Evaluation correctness:]For every $\prfinp \in \Domprf$, we have
\begin{align}
\Pr\left[
\qLEval(\qsk,\prfinp) \allowbreak = \Eval(\msk,\prfinp)
\ \middle |
\begin{array}{ll}
(\msk,\qsk,\dvk)\gets\qKG(1^\secp)\\
\end{array}
\right] 
=1-\negl(\secp).
\end{align}

\item[Deletion verification correctness:] We have 
\begin{align}
\Pr\left[
\DelVrfy(\dvk,\cert)=\top
\ \middle |
\begin{array}{ll}
(\msk,\qsk,\dvk)\gets\qKG(1^\secp)\\
\cert\gets\qDel(\qsk)
\end{array}
\right] 
=1-\negl(\secp).
\end{align}
\end{description}
\end{definition}
\begin{remark}[Reusability]\label{rem:reusability＿PRF}
We can assume without loss of generality that a key of a PRF-SKL scheme is reusable, i.e., it can be reused to evaluate on  (polynomially) many inputs. This is because the output of the leased evaluation algorithm is almost unique by evaluation correctness, and thus such an operation can be done without almost disturbing the input state by the gentle measurement lemma~\cite{Winter99}.    
\end{remark}

Our security definition is similar to that in \cite{TCC:AnaPorVai23} except that we allow the adversary to receive the verification key after submitting a valid certificate.  
\begin{definition}[PR-VRA security]\label{def:PR-VRA}
We say that a PRF-SKL scheme  $\PRFSKL$ with classical revocation for the domain $\Domprf$ and the range $\Ranprf$ is PR-VRA secure, if it satisfies the following requirement, formalized by the experiment $\expb{\PRFSKL,\qA}{pr}{vra}(1^\secp,\coin)$ between an adversary $\qA$ and the challenger:
        \begin{enumerate}
            \item  The challenger runs $(\msk,\qsk,\dvk)\gets\qKG(1^\secp)$ and sends $\qsk$ to $\qA$. 
            \item $\qA$ sends $\cert$ to the challenger. If $\DelVrfy(\dvk,\cert)=\bot$, the challenger outputs $0$ as the final output of this experiment. Otherwise, the challenger generates $\prfinp^*\la\Domprf$, $\prfout^*_0\la\Eval(\msk,\prfinp^*)$, and $\prfout^*_1\la\Ranprf$, and sends $(\dvk,\prfinp^*,\prfout^*_\coin)$ to $\qA$.
            \item $\qA$ outputs a guess $\coin^\prime$ for $\coin$. The challenger outputs $\coin'$ as the final output of the experiment.
        \end{enumerate}
        For any QPT $\qA$, it holds that
\begin{align}
\advb{\PRFSKL,\qA}{pr}{vra}(\secp) \seteq \abs{\Pr[\expb{\PRFSKL,\qA}{pr}{vra} (1^\secp,0) = 1] - \Pr[\expb{\PRFSKL,\qA}{pr}{vra} (1^\secp,1) = 1] }\leq \negl(\secp).
\end{align} 
\end{definition}
\begin{remark}[Relationship with the definition of \cite{TCC:AnaPorVai23}]
It is easy to see that PR-VRA security implies the security of \cite{TCC:AnaPorVai23} since theirs is a weaker variant of PR-VRA security where the verification key is not given to the adversary.\footnote{Strictly speaking, \cite{TCC:AnaPorVai23} only gives a security definition in the quantum revocation setting, but it can be extended to the classical revocation setting in a straightforward manner.}
\end{remark}
\begin{remark}[Variant with the deletion verification oracle]
We can consider a seemingly stronger definition where we allow the adversary to make arbitrarily many queries to the deletion verification oracle in the first phase. However, we can show that the PR-VRA security implies security in such a setting as well by a similar argument to the proof of \Cref{lem:IND-VRA_to_IND-KLA}. Thus, we omit the deletion verification oracle from the definition for simplicity. 
\end{remark}
\begin{remark}[Security as a plain PRF]\label{rem:security_plain_PRF}
PR-VRA security does not imply security as a plain PRF though it \emph{does} imply security as a plain \emph{weak} PRF. However, there is a very simple transformation to add security as a plain PRF as observed in \cite[Claim 9.8]{TCC:AnaPorVai23}.  
We can simply take an XOR of PR-VRA secure PRF-SKL and a plain PRF to achieve both security simultaneously. 
Thus we focus on achieving PR-VRA security in this paper. 
\end{remark}

As an intermediate goal towards constructing PRF-SKL, we introduce a primitive which we call unpredictable functions with secure key leasing (UPF-SKL).\footnote{Though UPF-SKL and PRF-SKL are syntactically identical, we treat them as different primitives for clarity.}

\begin{definition}[UPF-SKL]\label{def:UPF-SKL}
A UPF-SKL scheme $\mathsf{UPFSKL}$ with classical revocation has the same syntax as PRF-SKL with classical revocation and satisfies the following security, which we call UP-VRA security, 
formalized by the experiment $\expb{\UPFSKL,\qA}{up}{vra}(1^\secp)$ between an adversary $\qA$ and the challenger:
        \begin{enumerate}
            \item  The challenger runs $(\msk,\qsk,\dvk)\gets\qKG(1^\secp)$ and sends $\qsk$ to $\qA$. 
            \item $\qA$ sends $\cert$ to the challenger. If $\DelVrfy(\dvk,\cert)=\bot$, the challenger outputs $0$ as the final output of this experiment. Otherwise, the challenger generates $\prfinp^*\la\Domprf$, and sends $(\dvk,\prfinp^*)$ to $\qA$.
            \item $\qA$ outputs a guess $t^\prime$ for the output on $s^*$. The challenger outputs $1$ if $t^\prime=\Eval(\msk,s^*)$ and otherwise outputs $0$. 
        \end{enumerate}
        For any QPT $\qA$, it holds that
\begin{align}
\advb{\UPFSKL,\qA}{up}{vra}(\secp) \seteq \Pr[\expb{\UPFSKL,\qA}{up}{vra} (1^\secp) = 1] \leq \negl(\secp).
\end{align} 
\end{definition}

By the quantum Goldreich-Levin lemma~\cite{AC02,C:CLLZ21} we have the following theorem. 
\begin{lemma}\label{lem:upf-prf}
If there exists a UP-VRA secure UPF-SKL scheme with classical revocation, then there exists a PR-VRA secure PRF-SKL scheme with classical revocation.
\end{lemma}
See \Cref{sec:UPF-to-PRF} for the proof. 

\subsection{Digital Signatures with Secure Key Leasing}
In this subsection, we define digital signatures with secure key leasing (DS-SKL) with classical revocation. 
A recent work by Morimae et al.~\cite{TQC:MorPorYam24} introduced 
a similar primitive called digital signatures with revocable signing keys. A crucial difference from their definition is that we require the quantum signing key to be static, i.e., the state of the signing key almost does not change before and after signing. 
\mor{static is the same as reusability defined below? In that case, it is better to use a single terminology.
reusability is formaly defined, but static is not. In remark 3.18, you use static.}
\takashi{
Yes, that is the same as reusability defined below. 
However, reusability is a confusing term when comparing with \cite{TQC:MorPorYam24} since theirs is also "reusable" in some sense. 
To clarify the difference, I renamed it to "reusability with static signing keys".
}

\begin{definition}[DS-SKL with classical revocation]
A DS-SKL scheme $\DSSKL$  with classical revocation  
is a tuple of five algorithms $(\qKG, \qSign, \SigVrfy,\qDel,\DelVrfy)$. 
Below, let $\cM$  be the message space of $\DSSKL$. 
\begin{description}
\item[$\qKG(1^\secp)\ra(\qsigk, \sigvk,\dvk)$:] The key generation algorithm is a QPT algorithm that takes a security parameter $1^\lambda$, and outputs a signing key $\qsigk$, a signature verification key $\sigvk$, and a deletion verification key $\dvk$.

\item[$\qSign(\qsigk,\msg)\ra(\qsigk',\sigma)$:] The signing algorithm is a QPT algorithm that takes a signing key $\qsigk$ and a message $\msg \in \cM$, and outputs a subsequent signing key $\qsigk'$ and a signature $\sigma$.  \takashi{I modified the syntax to output $\qsigk'$.}

\item[$\SigVrfy(\sigvk,\msg,\sigma)\ra \top/\bot$:] The signature verification algorithm is a deterministic classical polynomial-time algorithm that takes a signature verification key $\sigvk$, a message $\msg \in \cM$, and a signature $\sigma$, and outputs $\top$ or $\bot$. 

\item[$\qDel(\qsigk)\ra\cert$:] The deletion algorithm is a QPT algorithm that takes a signing key $\qsigk$, and outputs a deletion certificate $\cert$.

\item[$\DelVrfy(\dvk,\cert)\ra\top/\bot$:] The deletion verification algorithm is a deterministic classical polynomial-time algorithm that takes a deletion  verification key $\dvk$ and a deletion certificate $\cert$, and outputs $\top$ or $\bot$.

\item[Signature verification correctness:]For every $\msg \in \cM$, we have
\begin{align}
\Pr\left[
\SigVrfy(\sigvk, \msg,\sigma) \allowbreak = \top
\ \middle |
\begin{array}{ll}
(\qsigk,\sigvk,\dvk)\gets\qKG(1^\secp)\\
(\qsigk',\sigma)\gets\qSign(\qsigk,\msg)
\end{array}
\right] 
=1-\negl(\secp).
\end{align}

\item[Deletion verification correctness:] We have 
\begin{align}
\Pr\left[
\DelVrfy(\dvk,\cert)=\top
\ \middle |
\begin{array}{ll}
(\qsigk,\sigvk,\dvk)\gets\qKG(1^\secp)\\
\cert\gets\qDel(\qsigk)
\end{array}
\right] 
=1-\negl(\secp).
\end{align}
\item[Reusability with static signing keys:]
Let $\qsigk$ be an honestly generated signing key and $\msg\in \mathcal{M}$ be any message. 
Suppose that we run $(\qsigk',\sigma)\gets\qSign(\qsigk,\msg)$. 
Then we have 
\begin{align}
    \|\qsigk-\qsigk'\|_{tr}=\negl(\secp).
\end{align}
\end{description}
\end{definition}

\begin{remark}[Reusability]\label{rem:reusability_DS}
The previous work \cite{TQC:MorPorYam24} considered a weaker definition of reusability, where $\sigk'$ is not required to be close to $\sigk$ as long as it can be still used to generate signatures on further messages.
To emphasize the difference from their definition, we call the above property reusability with static signing keys. 
Unlike the cases of PKE-SKL and PRF-SKL, we cannot assume (even the weaker version of) reusability without loss of generality since signatures generated by the signing algorithm may not be unique. Thus, we explicitly state it as a requirement.
\end{remark}

We next introduce the security notions for DS-SKL with classical revocation.

\begin{definition}[RUF-VRA security]\label{def:RUF-VRA} 
We say that a DS-SKL scheme  $\DSSKL$ with classical revocation  for the message space $\cM$ is RUF-VRA secure,\footnote{"RUF" stands for "\textbf{R}andom message \textbf{U}n\textbf{F}orgeability".} if it satisfies the following requirement, formalized by the experiment $\expb{\DSSKL,\qA}{ruf}{vra}(1^\secp)$  between an adversary $\qA$ and the\mor{a dattari the dattari surunode touitsu shitahougayoi}\takashi{let's use "the".} challenger:
        \begin{enumerate}
            \item  The challenger runs $(\qsigk,\sigvk,\dvk)\gets\qKG(1^\secp)$ and sends $\qsigk$ and $\sigvk$ to $\qA$. 
            \item $\qA$ sends $\cert$  to the challenger. If $\DelVrfy(\dvk,\cert)=\bot$, the challenger outputs $0$ as the final output of this experiment. Otherwise, the challenger chooses $\msg^*\la\mathcal{M}$, and sends $\dvk$ and $\msg^*$ to $\qA$.
            \item $\qA$ outputs a signature $\sigma^\prime$. The challenger outputs $1$ if $\SigVrfy(\sigvk,\msg^*,\sigma^\prime)=\top$ and otherwise outputs $0$.  
        \end{enumerate}
        For any QPT $\qA$, it holds that
\begin{align}
\advb{\DSSKL,\qA}{ruf}{vra}(\secp) \seteq \Pr[\expb{\DSSKL,\qA}{ruf}{vra} (1^\secp) = 1] \leq \negl(\secp).
\end{align} 
\end{definition}

\begin{remark}[Comparison with \cite{TQC:MorPorYam24} ]\label{rem:comparison_DS}
We cannot directly compare our security definition with that in \cite{TQC:MorPorYam24} since the syntax is significantly different. In \cite{TQC:MorPorYam24}, the deletion verification algorithm additionally takes a set $S$ of messages as input. Intuitively, the deletion verification algorithm verifies that the signer generated signatures only on messages that belong to $S$. Thus, their security definition requires that the adversary cannot generate a valid signature on any message outside $S$ after submitting a certificate that passes verification w.r.t. $S$. While their definition looks conceptually stronger, such security is impossible to achieve when the signing key is static as in our definition since the adversary may generate signatures on any messages without being known by the verifier. Thus, we adopt the above random message style definition analogously to a security definition of copy-protection of signing keys~\cite{TCC:LLQZ22}. 
\end{remark}
\begin{remark}[Variant with the deletion  verification oracle]
We can consider a seemingly stronger definition where we allow the adversary to make arbitrarily many queries to the deletion verification oracle in the first phase similarly. However, we can show that the RUF-VRA security implies security in such a setting as well by a similar argument to the proof of \Cref{lem:IND-VRA_to_IND-KLA}.  Thus, we omit the deletion verification oracle from the definition for simplicity.
\end{remark}
\begin{remark}[Security as a plain DS]
 RUF-VRA security does not imply EUF-CMA security as a plain DS. However, similarly to the case of PRF, there is a very simple transformation to add EUF-CMA security as a plain DS. 
We can simply parallelly run an RUF-VRA secure DS-SKL scheme and EUF-CMA secure plain DS scheme where the signature verification algorithm accepts if the signature consists of valid signatures of both schemes. It is straightforward to see that the resulting scheme satisfies both RUF-VRA security as an DS-SKL scheme and EUF-CMA security as a plain DS scheme. Thus, we focus on constructing an RUF-VRA secure DS-SKL scheme.  
\end{remark}

	\setcounter{tocdepth}{2}
	\tableofcontents

	\fi
\fi

\end{document}